\documentclass[11pt]{article}
\usepackage[utf8]{inputenc}

\usepackage[margin=1in]{geometry}
\usepackage{amsmath,amsfonts,amssymb,amsthm}
\usepackage{aliascnt}
\usepackage{mathtools}
\mathtoolsset{centercolon}
\usepackage{mathrsfs}
\usepackage{bm}

\usepackage{float}

\usepackage{tikz}
\usetikzlibrary{arrows.meta,positioning}

\usepackage{enumitem}
\usepackage{booktabs}
\usepackage{xcolor}
\usepackage[colorlinks]{hyperref}
\usepackage[capitalize]{cleveref}

\allowdisplaybreaks
\numberwithin{equation}{section}

\newtheorem{theorem}{Theorem}[section]

\newaliascnt{lemma}{theorem}
\newtheorem{lemma}[lemma]{Lemma}
\aliascntresetthe{lemma}

\newaliascnt{corollary}{theorem}
\newtheorem{corollary}[corollary]{Corollary}
\aliascntresetthe{corollary}

\newaliascnt{proposition}{theorem}
\newtheorem{proposition}[proposition]{Proposition}
\aliascntresetthe{proposition}

\newaliascnt{assumption}{theorem}
\newtheorem{assumption}[assumption]{Assumption}
\aliascntresetthe{assumption}

\theoremstyle{definition}

\newaliascnt{definition}{theorem}
\newtheorem{definition}[definition]{Definition}
\aliascntresetthe{definition}

\newaliascnt{notation}{theorem}
\newtheorem{notation}[notation]{Notation}
\aliascntresetthe{notation}

\newaliascnt{remark}{theorem}
\newtheorem{remark}[remark]{Remark}
\aliascntresetthe{remark}

\newaliascnt{example}{theorem}

\aliascntresetthe{example}

\crefname{problem}{problem}{problems}
\Crefname{problem}{Problem}{Problems}
\crefname{theorem}{theorem}{theorems}
\Crefname{theorem}{Theorem}{Theorems}
\crefname{lemma}{lemma}{lemmas}
\Crefname{lemma}{Lemma}{Lemmas}
\crefname{corollary}{corollary}{corollaries}
\Crefname{corollary}{Corollary}{Corollaries}
\crefname{proposition}{proposition}{propositions}
\Crefname{proposition}{Proposition}{Propositions}
\crefname{assumption}{assumption}{assumptions}
\Crefname{assumption}{Assumption}{Assumptions}
\crefname{definition}{definition}{definitions}
\Crefname{definition}{Definition}{Definitions}
\crefname{notation}{notation}{notations}
\Crefname{notation}{Notation}{Notations}
\crefname{remark}{remark}{remarks}
\Crefname{remark}{Remark}{Remarks}
\crefname{example}{example}{examples}
\Crefname{example}{Example}{Examples}

\newcommand{\C}{\mathbb{C}}
\newcommand{\R}{\mathbb{R}}
\newcommand{\I}{\mathrm{I}}
\newcommand{\spec}{\operatorname{spec}}
\newcommand{\rank}{\operatorname{rank}}
\newcommand{\range}{\operatorname{range}}
\newcommand{\signf}{\operatorname{sign}}
\newcommand{\dist}{\operatorname{dist}}
\newcommand{\sepop}{\operatorname{sep}}

\DeclareMathOperator{\diag}{diag}
\newcommand{\BE}{\mathrm{BE}}
\newcommand{\eps}{\varepsilon}
\newcommand{\wt}[1]{\widetilde{#1}}
\newcommand{\half}{\tfrac12}
\newcommand{\abs}[1]{\left\lvert #1 \right\rvert}
\newcommand{\norm}[1]{\left\lVert #1 \right\rVert}

\newcommand{\bracks}[1]{\left[ #1 \right]}
\newcommand{\paren}[1]{\left( #1 \right)}
\newcommand{\set}[1]{\left\{ #1 \right\}}
\newcommand{\suchthat}{\;:\;}
\newcommand{\RePart}{\operatorname{Re}}
\newcommand{\ImPart}{\operatorname{Im}}
\newcommand{\qA}{q_A}
\newcommand{\qB}{q_B}
\newcommand{\qC}{q_C}
\newcommand{\qE}{q_E}
\newcommand{\qD}{q_D}
\newcommand{\ket}[1]{\lvert #1 \rangle}
\newcommand{\bra}[1]{\langle #1 \rvert}

\newcommand{\secref}[1]{\hyperref[sec:#1]{Section~\ref*{sec:#1}}}
\newcommand{\appref}[1]{\hyperref[app:#1]{Appendix~\ref*{app:#1}}}
\newcommand{\thmref}[1]{\hyperref[thm:#1]{Theorem~\ref*{thm:#1}}}
\newcommand{\lemref}[1]{\hyperref[lem:#1]{Lemma~\ref*{lem:#1}}}
\newcommand{\corref}[1]{\hyperref[cor:#1]{Corollary~\ref*{cor:#1}}}
\newcommand{\fig}[1]{\hyperref[fig:#1]{Figure~\ref*{fig:#1}}}

\usepackage{tabularx}
\usepackage{ragged2e}
\usepackage{makecell}
\usepackage{placeins}
\newcolumntype{Y}{>{\RaggedRight\arraybackslash}X}

\usepackage{array,threeparttable,multirow,pifont}

\usepackage{graphicx}

\title{Sign Embedding Quantum Algorithms for \\ Matrix Equations and Matrix Functions}
\author{Yanqiao Wang$^{1,2,3}$, Jin-Peng Liu$^{1,4,5,\thanks{Corresponding author: liujinpeng@tsinghua.edu.cn}}$
\\ 
\footnotesize $^{1}$ Yau Mathematical Sciences Center, Tsinghua University\\
\footnotesize $^{2}$ Qiuzhen College, Tsinghua University\\
\footnotesize $^{3}$ Institute for AI Industry Research (AIR), Tsinghua University\\
\footnotesize $^{4}$ Institute for Applied Mathematics, Tsinghua University\\
\footnotesize $^{5}$ Beijing Institute of Mathematical Sciences and Applications\\
}
\date{}

\begin{document}

\maketitle

\begin{abstract}
We develop operator-output quantum algorithms for matrix equations and matrix functions using matrix-sign embeddings. For each problem, an augmented matrix $M$ is chosen so that the target operator is a block of $\signf(M)$ or is recovered from blocks of $(\I-\signf(M))/2$. We approximate the half-plane sign by a logarithmic-sinc formula and implement the shifted inverse families by scaled multiplexing with node-dependent rebalancing. For ordinary Sylvester equations, this yields a block-encoding of the solution with query complexity linear in the relevant conditioning parameters and logarithmic in the inverse error tolerance, under either a field-of-values (FoV) gap or a strip-resolvent bound. The same method extends to generalized Sylvester and Lyapunov equations, principal square and inverse square roots, matrix geometric means, and continuous-time algebraic Riccati equations (CARE), with explicit query-complexity and block-encoding-normalization bounds that cover non-normal and non-diagonalizable cases.
\end{abstract}

\tableofcontents

\section{Introduction}\label{sec:introduction}
Matrix equations arise throughout numerical linear algebra, systems and control, and scientific computing. A basic example is the Sylvester equation
\begin{equation}\label{eq:sylvester}
AX + XB = C,
\end{equation}
which appears in Lyapunov and Riccati theory, model reduction, eigenvalue assignment, matrix functions, descriptor systems, and dynamical-systems analysis. Classical solvers for Sylvester and Lyapunov equations include Schur-based dense methods and projection/Krylov schemes \cite{BartelsStewart1972,Higham2008,Simoncini2016}. Matrix iterations have also been used for related model-reduction, Riccati, and generalized Lyapunov problems \cite{Roberts1980,Byers1987,KenneyLaub1995,BennerQuintanaOrti1999,LancasterRodman1995}. Their matrix-valued solutions motivate quantum algorithms with operator output.

We also consider principal square roots, inverse square roots, geometric means, and invariant-subspace projectors. Classical algorithms for these matrix functions include contour-integral, rational-approximation, and iterative methods \cite{Higham2008,HaleHighamTrefethen2008,NakatsukasaFreund2016,Gawlik2019}. Recent work on Sylvester equations includes inverse-free iterative solvers based on an augmented block formulation \cite{BallewTrogdonWilber2025}.

The HHL algorithm prepares a state proportional to the solution of a linear system; later algorithms improved its precision dependence \cite{HarrowHassidimLloyd2009,ChildsKothariSomma2017,Dervovic2018}. Qubitization, quantum signal processing, and block-encoding led to Quantum Singular Value Transformation (QSVT), which applies polynomial transformations to the singular values of block-encoded matrices \cite{LowChuang2017,LowChuang2019,Gilyen2018}. Linear Combination of Hamiltonian Simulation (LCHS) and Schr\"odingerisation address non-unitary dynamics and PDEs, with extensions to Laplace-transform eigenvalue transformations and improved matrix-query bounds \cite{AnLiuLin2023,AnChildsLin2023,AnChildsLinYing2024,NovikauJoseph2025,SchleichKharaziLiLiuAspuruWiebe2025,LowSomma2025,HuangAn2025,YangLiu2025RandomLCHS,JinLiuYu2023,JinLiuYu2024,JinLiuMa2024,JinLiuMaPengYu2025}. Quantum eigenvalue processing acts on the eigenvalues of block-encoded non-normal matrices \cite{LowSu2024}. Other algorithms address pseudospectrum membership \cite{YangLengWuLin2026} and non-Hermitian matrix functions through contour-integral or Poisson-summation formulas \cite{TakahiraOhashiSogabeUsuda2022,WangZhuangDouChenGuo2026}.

Quantum algorithms for matrix equations include ordinary Sylvester solvers, Lyapunov solvers with state, application-specific operator, or mixed-state output, and algorithms for matrix geometric means \cite{SommaLowBerryBabbush2025,Clayton2024,Benedetti2025,LiuWangWildeZhang2025}.

\subsection{Matrix-sign embeddings}
The normalization determines the scale of an operator output. A normalized vectorization represents an entry as $X_{jk}/\norm{X}_{\mathrm F}$, and a block-encoding with normalization $\beta_X$ represents it as $X_{jk}/\beta_X$. For an $N\times N$ unitary or permutation matrix, $\norm{X}_{\mathrm F}=\sqrt{N}$. Thus, whenever $\beta_X=O(1)$, a block-encoding avoids this dimension-dependent loss in entry estimation. It also provides coherent access to $X/\beta_X$ for applying $X$ to selected input states and composing it with other encoded operators \cite{SommaLowBerryBabbush2025}.

The motivating example is the Sylvester equation \eqref{eq:sylvester}. Apply a shift-rotation and common positive scaling that place $A$ and $-B$ in opposite half-planes. These transformations leave the solution unchanged. The augmented matrix
\[
M=
\begin{bmatrix}
A&C\\
0&-B
\end{bmatrix}
\]
satisfies the identity
\[
\signf(M)=
\begin{bmatrix}
\I&2X\\
0&-\I
\end{bmatrix},
\]
so the solution is the off-diagonal block of the sign. The construction can be written as
\[
\text{target }X
\Longrightarrow
\text{sign embedding }M
\Longrightarrow
\text{rational approximation of }\signf(M)
\Longrightarrow
\text{block-encoding of }X.
\]
Analogous identities apply to generalized Sylvester and Lyapunov equations, principal square roots and inverse square roots, matrix geometric means, and continuous-time algebraic Riccati equations (CAREs). The target is recovered from a sign block or a sign projector. For principal roots $A^{1/2}$ and $A^{-1/2}$, one uses
\[
K(A)=
\begin{bmatrix}
0 & A\\
\I & 0
\end{bmatrix},
\qquad
\signf(K(A))=
\begin{bmatrix}
0 & A^{1/2}\\
A^{-1/2} & 0
\end{bmatrix};
\]
for geometric means $A\# B$, one uses $G(A,B)$, or any positive scaling of it since $\signf(\chi G)=\signf(G)$ for $\chi>0$,
\[
G(A,B)=
\begin{bmatrix}
0 & B\\
A^{-1} & 0
\end{bmatrix},
\qquad
\signf(G(A,B))=
\begin{bmatrix}
0 & A\# B\\
(A\# B)^{-1} & 0
\end{bmatrix};
\]
and for continuous-time algebraic Riccati
equations $A^*X + XA - XGX + Q = 0$, one uses the Hamiltonian
\[
H=
\begin{bmatrix}
A & -G\\
-Q & -A^*
\end{bmatrix},
\qquad
\Pi_-:=\frac{\I-\signf(H)}{2},
\qquad
X=\Pi_{21}\Pi_{11}^{-1}.
\]
Each problem has its own augmented matrix and final block extraction. The sign approximation and shifted-inverse implementation are shared, and CARE adds an inversion of the leading projector block. Together, these identities give one construction for deterministic block-encodings of ordinary and generalized Sylvester equations, generalized Lyapunov equations, principal square and inverse square roots, matrix geometric means, and CAREs. The construction has three components:
\begin{enumerate}[leftmargin=2em,label=(\roman*)]
\item a \emph{sign representation} of the target as a block of the sign or as a quotient of projector blocks; see \cref{thm:higham,thm:sqrt-sign,thm:care-projector};
\item a \emph{logarithmic-sinc approximation} of the half-plane sign on a real logarithmic coordinate; see \cref{thm:sign-approx};
\item a \emph{scaled-multiplexing implementation} of the shifted inverses, with nodewise rebalancing used to improve the block-encoding normalization relative to a worst-case conditioning estimate; see \cref{thm:profile-implementation,thm:single-family-algorithm}.
\end{enumerate}
The analysis uses resolvent estimates and does not require an eigendecomposition. It therefore includes non-normal and non-diagonalizable matrices covered by the corresponding theorems, including instances for which diagonalization-based bounds are unavailable.

Figure~\ref{fig:my_image_0} depicts the constructions for the ordinary Sylvester equation, matrix roots, geometric means, and CARE.

\begin{figure}[H]
  \centering
  \includegraphics[width=1.0\textwidth]{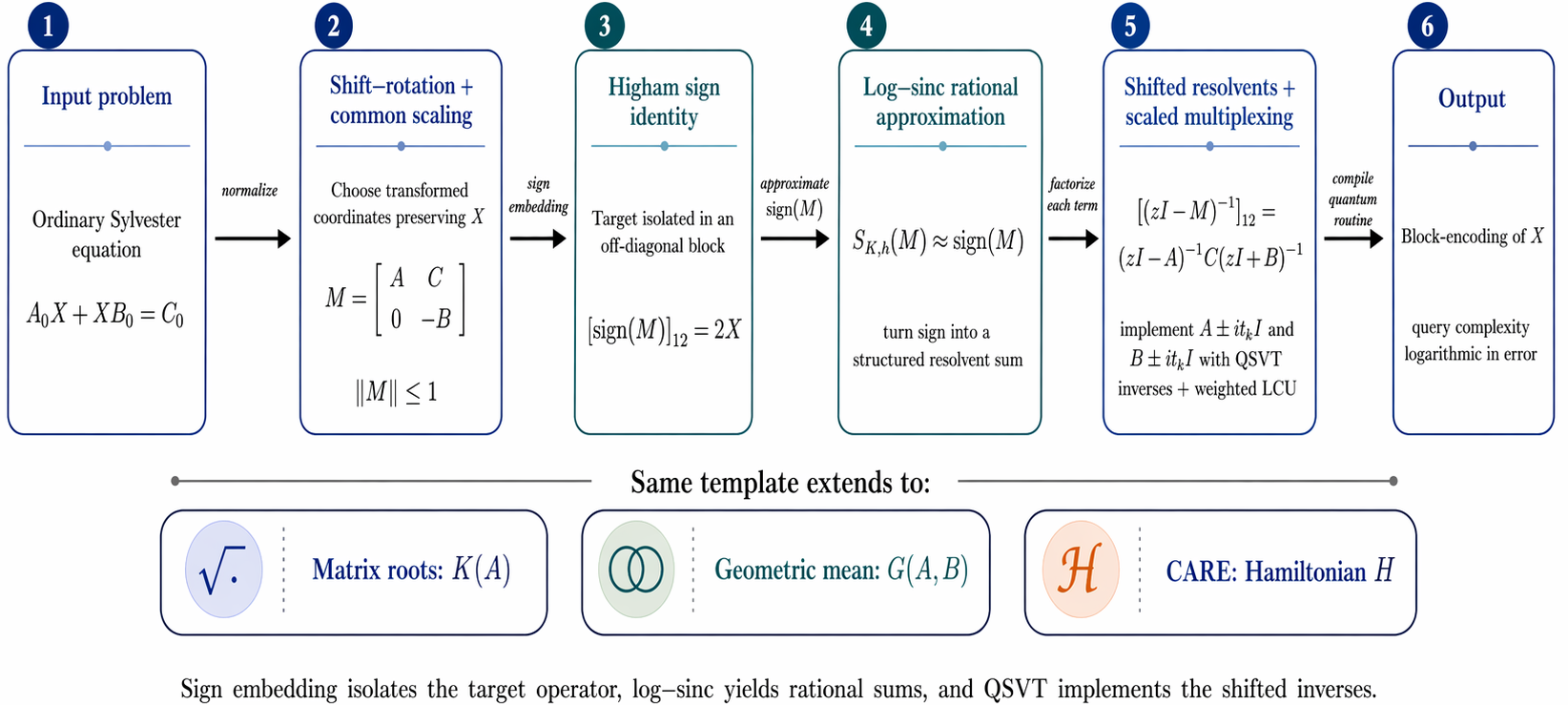}
  \caption{Sign-embedding algorithms for the ordinary Sylvester equation, matrix roots, geometric means, and CARE.}
  \label{fig:my_image_0}
\end{figure}

Generic contour-integral algorithms discretize a complex contour chosen for the target function. Here $t=e^x$ maps the half-plane sign integral to the real line $x\in\R$, and uniform log-sinc nodes $x_k=kh$ give geometric shifts $t_k=e^{kh}$. Strip analyticity and decaying tails yield root-exponential error $O(e^{-\sqrt{2\pi\beta K}})$. The same shifts match the scaled resolvent factorizations and permit nodewise rebalancing.

For ordinary Sylvester equations, the matrix-oracle query complexity is linear in $1/\mu$ or $\gamma$ and logarithmic in the inverse error tolerance; the right-hand-side oracle is queried only $O(1)$ times. In the banded FoV regime, nodewise rebalancing sharpens the uniform $O(\mu^{-2})$ normalization to $O(\mu^{-1}+\tau\mu^{-2})$, reaching the optimal $O(\mu^{-1})$ dependence for Hermitian inputs.

\subsection{Main results}
The results cover the following matrix equations and matrix functions.

\begin{enumerate}[leftmargin=2em,label=(\arabic*)]
\item \textbf{Ordinary Sylvester equation.}
For the transformed Sylvester problem, the rebalanced sign-embedding bound is given in \cref{thm:generic-overlap}. We state two specializations.

\emph{Regime $(\mathrm{R1})$: field-of-values (FoV) gap.}
After a suitable shift-rotation, suppose
\[
H(A)\succeq \mu \I,
\qquad
H(B)\succeq \mu \I.
\]
Then \cref{thm:R1-main,cor:fov-main} give a unitary $U_X$ and an ancilla count $a_X$ such that
\[
\left\|
X-\beta_X(\bra{0^{a_X}}\otimes \I)U_X(\ket{0^{a_X}}\otimes \I)
\right\|\le \eps,
\qquad
\beta_X=O\!\left(\frac{1}{\mu^2}\right),
\]
with query complexity
\[
\qA=O\!\left(\frac{1}{\mu}\log\frac{1}{\eps\mu}\right),
\qquad
\qB=O\!\left(\frac{1}{\mu}\log\frac{1}{\eps\mu}\right),
\qquad
\qC=O(1).
\]
A geometric FoV gap $\dist(W(A_0),-W(B_0))>0$ can be converted into this regime by an optimal shift-rotation. Under this FoV hypothesis, \cref{lem:R1-strip-bound} also gives a strip-resolvent bound for the augmented matrix.

\emph{Banded-overlap refinement.} \cref{thm:R1-banded} improves the normalization to
\[
\beta_X=O\!\left(\frac{1}{\mu}+\frac{\tau}{\mu^2}\right)
\]
whenever the numerical ranges lie in a horizontal strip of half-width $\tau$. Thus $\beta_X=O(1/\mu)$ when $\tau=O(\mu)$. For Hermitian inputs, $\tau=0$ and \cref{thm:R1-banded} gives $\beta_X\le 4/\mu$, with optimal dependence on $\mu$.

\emph{Regime $(\mathrm{R2})$: strip-resolvent bound.}
If the augmented matrix
\[
M=
\begin{bmatrix}
A & C\\
0 & -B
\end{bmatrix}
\]
is separated from the imaginary axis and admits a strip-resolvent bound $\gamma$ on $\Omega_a$, then \cref{thm:R2-main} gives a unitary $U_X$ and an ancilla count $a_X$ such that
\[
\left\|
X-\beta_X(\bra{0^{a_X}}\otimes \I)U_X(\ket{0^{a_X}}\otimes \I)
\right\|\le \eps,
\qquad
\beta_X=O(\gamma^2),
\]
with query complexity
\[
\qA=O\!\left(\gamma\log\frac{\gamma}{\eps}\right),
\qquad
\qB=O\!\left(\gamma\log\frac{\gamma}{\eps}\right),
\qquad
\qC=O(1).
\]
This result includes non-normal and non-diagonalizable inputs for which an FoV gap is unavailable. \Cref{thm:generic-overlap} refines the worst-case product bound to a nodewise overlap sum formed from the shifted resolvent norms.

\item \textbf{Generalized Sylvester and generalized Lyapunov equations.}
\Cref{thm:gen-main,cor:lyap-main} reduce nonsingular generalized Sylvester and generalized Lyapunov equations to ordinary Sylvester problems with coefficients
\[
\widetilde A=E^{-1}A,
\qquad
\widetilde B=BD^{-1},
\qquad
\widetilde C=E^{-1}CD^{-1}.
\]
Applying the ordinary Sylvester construction to these coefficients gives the generalized algorithms. In the weighted FoV-gap specialization, the block-encoding normalizations scale as $O(\norm{E^{-1}}\norm{D^{-1}}/\mu^2)$ for generalized Sylvester and $O(\norm{E^{-1}}^2/\mu^2)$ for generalized Lyapunov, while the query costs remain linear in the corresponding weighted inverse-conditioning parameters.

\item \textbf{Matrix functions: matrix roots and geometric means.}
\Cref{thm:sqrt-sign,thm:sqrt-generic,thm:sqrt-accretive,thm:geom-mean,thm:geom-mean-accretive} give block-encodings of $A^{-1/2}$, $A^{1/2}$, and the operator geometric mean $A\#B := A(A^{-1}B)^{1/2}$. For the square-root embedding
\[
K(A)=
\begin{bmatrix}
0 & A\\
\I & 0
\end{bmatrix},
\qquad
\signf(K(A))=
\begin{bmatrix}
0 & A^{1/2}\\
A^{-1/2} & 0
\end{bmatrix},
\]
both $A^{-1/2}$ and $A^{1/2}$ are recovered as sign blocks. In the FoV-based regime $H(A)\succeq \mu\I$, \cref{thm:sqrt-accretive} gives normalization $4/\sqrt{\mu}$ for both outputs and query complexity
\[
O\!\paren{\frac{1}{\mu}\log\frac{1}{\eps\mu}}.
\]
The inverse-square-root normalization has optimal order: the scalar case $A=\mu$ already forces $\|A^{-1/2}\|=\mu^{-1/2}$. The square-root output is obtained by multiplying the inverse-square-root construction by $A$ once and has the same normalization. For the geometric-mean embedding, the single-family algorithm gives profile-dependent bounds. In the FoV specialization, \cref{thm:geom-mean-accretive} gives normalization at most $4/\sqrt{\mu_A\mu_B}$ for $(A\#B)^{-1}$. Left and right multiplication of this construction gives a block-encoding of $A\#B$ at the same scale.

\item \textbf{Continuous-time algebraic Riccati equation.}
\Cref{thm:care-projector,thm:care-sign,thm:care-main} give an algorithm based on the Hamiltonian matrix associated with the standard continuous-time algebraic Riccati equation,
\[
H=
\begin{bmatrix}
A & -G\\
-Q & -A^*
\end{bmatrix},
\qquad
\Pi_-:=\frac{\I-\signf(H)}{2}.
\]
The stabilizing solution is extracted as a quotient of projector blocks, $X=\Pi_{21}\Pi_{11}^{-1}$. The algorithm therefore consists of a Hamiltonian sign stage and one projector-block inversion. The final theorem returns a unitary $U_X$ satisfying
\[
\left\|
X-\beta_X(\bra{0^{a_X}}\otimes \I)U_X(\ket{0^{a_X}}\otimes \I)
\right\|\le \eps,
\]
with normalization controlled by the Hamiltonian sign stage and the invertibility gap of the leading projector block. This construction applies to the standard CARE and includes non-Hermitian, non-normal Hamiltonians covered by \cref{thm:care-main}.
\end{enumerate}

\begin{table}[H]
\centering
\small
\setlength{\tabcolsep}{4pt}
\renewcommand{\arraystretch}{1.2} 
\caption{Sign representations, query bounds, and block-encoding normalizations}
\label{tab:this_paper_internal_summary}
\begin{tabularx}{\textwidth}{>{\raggedright\arraybackslash}p{3.55cm} >{\raggedright\arraybackslash}p{3.55cm} >{\raggedright\arraybackslash}X >{\raggedright\arraybackslash}p{2.45cm}}
\toprule
Problem / regime & Sign representation & Representative query complexity & Normalization $\beta_X$\\
\midrule
\makecell[l]{Sylvester\\(FoV gap;\\ \cref{thm:R1-main},\\ \cref{cor:fov-main})}
& $[\signf\!\bigl(\begin{smallmatrix}A&C\\0&-B\end{smallmatrix}\bigr)]_{(1,2)}$
& $\qA,\qB=O\!\paren{\frac{1}{\mu}\log\frac{1}{\eps\mu}}$, $\qC=O(1)$
& $O\!\left(\mu^{-2}\right)$\\ \addlinespace[10pt] 

\makecell[l]{Sylvester\\(FoV gap, banded;\\ \cref{thm:R1-banded})}
& $[\signf\!\bigl(\begin{smallmatrix}A&C\\0&-B\end{smallmatrix}\bigr)]_{(1,2)}$
& $\qA,\qB=O\!\paren{\frac{1}{\mu}\log\frac{\tau}{\eps\mu}}$, $\qC=O(1)$
& $O\!\left(\mu^{-1}+\tau\mu^{-2}\right)$\\ \addlinespace[10pt]

\makecell[l]{Sylvester\\(strip-resolvent;\\ \cref{thm:R2-main})}
& $[\signf\!\bigl(\begin{smallmatrix}A&C\\0&-B\end{smallmatrix}\bigr)]_{(1,2)}$
& $\qA,\qB=O\!\left(\gamma\log\frac{\gamma}{\eps}\right)$, $\qC=O(1)$
& $O\!\left(\gamma^2\right)$\\ \addlinespace[10pt] 

\makecell[l]{Generalized Sylvester /\\ generalized Lyapunov\\(FoV gap;\\ \cref{thm:gen-main},\\ \cref{cor:lyap-main})}
& $\left[\signf\!\bigl(\begin{smallmatrix}E^{-1}A&E^{-1}CD^{-1}\\0&-BD^{-1}\end{smallmatrix}\bigr)\right]_{(1,2)}$
& \makecell[l]{
$\qA,\qE=O\!\left(\frac{\norm{E^{-1}}}{\mu}\log\frac{\norm{E^{-1}}\norm{D^{-1}}}{\eps\mu}\right)$,\\
$\qB,\qD=O\!\left(\frac{\norm{D^{-1}}}{\mu}\log\frac{\norm{E^{-1}}\norm{D^{-1}}}{\eps\mu}\right)$
}
& \makecell[l]{
$O\!\left(\norm{E^{-1}}\norm{D^{-1}}\mu^{-2}\right)$,\\
$O\!\left(\norm{E^{-1}}^{2}\mu^{-2}\right)$
}\\ \addlinespace[10pt]

\makecell[l]{Inverse square root /\\ square root\\(FoV gap;\\
\cref{thm:sqrt-accretive})}
& \makecell[l]{
$[\signf\!\bigl(\begin{smallmatrix}0&A\\ \I&0\end{smallmatrix}\bigr)]_{(2,1)}$,\\
$[\signf\!\bigl(\begin{smallmatrix}0&A\\ \I&0\end{smallmatrix}\bigr)]_{(1,2)}$
}
& $\qA=O\!\paren{\frac{1}{\mu}\log\frac{1}{\eps\mu}}$
& $O\!\left(\mu^{-1/2}\right)$\\ \addlinespace[10pt] 

\makecell[l]{Geometric mean /\\ Inverse geometric mean\\(FoV + scaled\\ sign-resolvent bound;\\
\cref{thm:geom-mean-accretive})}
& \makecell[l]{
$\left[\signf\!\bigl(\begin{smallmatrix}0&B\\A^{-1}&0\end{smallmatrix}\bigr)\right]_{(2,1)}$,\\
$\left[\signf\!\bigl(\begin{smallmatrix}0&B\\A^{-1}&0\end{smallmatrix}\bigr)\right]_{(1,2)}$
}
& $O\!\paren{
\frac{1}{\min\{\mu_A,\mu_B\}}
\log(\frac{1}{\eps \min\{\mu_A,\mu_B\}})
}$
& \makecell[l]{$O\!\left((\mu_A\mu_B)^{-1/2}\right)$}\\ \addlinespace 

\makecell[l]{CARE\\(\cref{thm:care-main})}
& \makecell[l]{
$\Pi_-=(\I-\signf(H))/2$,\\
$X=\Pi_{21}\Pi_{11}^{-1}$
}
& \makecell[l]{
$\begin{aligned}
&Q_{\mathrm{total}}
=O\!\Bigl(R_H\log\frac{R_H}{\eps_H}\cdot
\\
&\bigl[1+\frac{\beta_\Pi}{2\sigma-\eps_{\mathrm{sign}}}
\log\frac{1}{(2\sigma-\eps_{\mathrm{sign}})\eps_{\Pi^{-1}}}
\bigr]\Bigr)
\end{aligned}$
}
& \makecell[l]{
$O\!\left(\dfrac{1+\beta_{\mathrm{sign}}}{2\sigma-\eps_{\mathrm{sign}}}\right)$
}
\\
\bottomrule
\end{tabularx}
\end{table}

\Cref{tab:this_paper_internal_summary} lists the sign representation, representative query complexity, and block-encoding normalization for each problem and regime. The cited theorems give the assumptions and definitions of the conditioning parameters.

\subsection{Relation to prior work}
Quantum Singular Value Transformation (QSVT) and block-encoding provide general algorithms for transformations of encoded matrices \cite{Gilyen2018}. Algorithms for non-Hermitian matrix functions based on holomorphic functional calculus discretize complex contours \cite{TakahiraOhashiSogabeUsuda2022}. Poisson-summation formulas give discrete contour transformations for generalized quantum matrix transformations \cite{WangZhuangDouChenGuo2026}. The real-line log-sinc rule used here exploits the common half-plane representation of the embedded problems. Its uniform logarithmic nodes give geometric resolvent shifts and root-exponential accuracy. These shifts also match the block-resolvent identities, allowing the quadrature weights, inverse costs, and output normalization to be analyzed together.

For ordinary Sylvester equations with operator output, Somma, Low, Berry, and Babbush construct a block-encoding through the vectorized Kronecker operator \cite{SommaLowBerryBabbush2025}. Their representative analysis assumes diagonalizability. The augmented-block construction used here requires no eigendecomposition and includes non-normal and defective matrices covered by the corresponding theorems. Its matrix-oracle query cost is linear in $1/\mu$ or $\gamma$, up to a logarithmic dependence on the target precision, and it uses only $O(1)$ queries to the right-hand-side oracle. Nodewise rebalancing replaces a uniform product of worst-case inverse bounds by an overlap sum of the two shifted-resolvent profiles. This produces the $O(\mu^{-1}+\tau\mu^{-2})$ banded normalization and the optimal $O(\mu^{-1})$ dependence for Hermitian inputs.

The operator output gives coherent access to $X/\beta_X$. State-output methods encode a matrix entry at the Frobenius scale $X_{jk}/\norm{X}_{\mathrm F}$. When $\beta_X$ remains small and $\norm{X}_{\mathrm F}$ grows with the dimension, the operator representation avoids the corresponding loss in entry estimation and supports direct composition with other block-encoded operators.

The same framework extends from ordinary Sylvester equations to generalized Sylvester and Lyapunov equations, principal roots, geometric means, and CAREs. Existing Lyapunov algorithms study vectorized-state, application-specific operator, and mixed-state outputs \cite{SunZhang2017,Clayton2024,Benedetti2025}; the results here give deterministic block-encodings of the matrix solutions. Liu, Wang, Wilde, and Zhang treat geometric means and related Riccati problems for Hermitian positive-definite inputs \cite{LiuWangWildeZhang2025}. The geometric-mean result here also includes non-normal inputs covered by its theorem. The CARE construction acts directly on the standard Hamiltonian and recovers the stabilizing solution from its stable invariant subspace.

Appendix~\ref{app:comparison} compares the assumptions and bounds with representative prior results; see Tables~\ref{tab:sylvester-compare}--\ref{tab:care-compare}.

\subsection{Organization}
\Cref{sec:model} introduces the computational model and inverse algorithm. \Cref{sec:sign-reduction,sec:sinc} derive the sign representation for ordinary Sylvester equations and the log-sinc approximation; \cref{prop:sign-approx-explicit} gives the constants. \Cref{sec:scaled-multiplexing} gives the rebalanced and single-family implementations and separates nodewise normalization estimates from the problem-specific sign reductions. \Cref{sec:main-results} contains the generic, FoV-gap, and strip-resolvent results for ordinary Sylvester equations. \Cref{sec:generalized,sec:sqrt,sec:care} treat generalized Sylvester and Lyapunov equations, principal square roots and geometric means, and continuous-time algebraic Riccati equations. Appendix~\ref{app:be-calculus} collects the block-encoding rules and multiplexed constructions, Appendix~\ref{app:mu-conditioning} relates the FoV gap $\mu$ to standard Sylvester conditioning quantities, and Appendix~\ref{app:comparison} gives the comparisons with quantum algorithms for these problems.

\section{Model and QSVT inverse primitive}\label{sec:model}
\subsection{Norms and transformed coordinates}
All norms are operator norms.
For a square matrix $K$, denote the Hermitian part as
\[
H(K) := \half (K+K^*)
\]
and denote its numerical range as
\[
W(K) := \set{ x^* K x \suchthat x \in \C^n,\ \norm{x}_2 = 1 }
\].

For the Sylvester operator $S_{A,B}(X) := AX + XB$, define
\begin{equation}\label{eq:sep}
\sepop(A,-B) := \min_{X\neq 0} \frac{\norm{AX+XB}}{\norm{X}} = \frac{1}{\norm{S_{A,B}^{-1}}}.
\end{equation}

We will repeatedly use transformations that leave the solution $X$ invariant.
If $\eta\in \C$ has $\abs{\eta}=1$, $\omega\in\C$, and $\lambda>0$, then
\[
A \mapsto \lambda\eta(A-\omega \I),\qquad B \mapsto \lambda\eta(B+\omega \I),\qquad C \mapsto \lambda\eta C
\]
leaves the solution of \eqref{eq:sylvester} unchanged.
Accordingly, whenever we speak about the matrices $A,B,C$ in the ordinary Sylvester part of the paper, we mean the matrices \emph{after} this shift-rotation and common positive scaling has been applied.

For each problem class, we use normalized coordinates in which every coefficient matrix supplied by an oracle has operator norm at most one.
For ordinary and generalized Sylvester-type equations, this normalization can be arranged by transformations that leave the solution operator unchanged.
For square roots and geometric means, scalar rescaling changes the target by a known scalar factor; the corresponding recovery formulas are given in the relevant sections.
For CARE, positive rescaling of the Hamiltonian leaves the sign projector unchanged.

For the sign-approximation theorem we also need a normalized augmented matrix.
Hence in the ordinary Sylvester sections we assume
\begin{equation}\label{eq:Mnorm}
\norm{M} \le 1,
\qquad
M := \begin{bmatrix} A & C \\ 0 & -B \end{bmatrix}.
\end{equation}
Since $A$, $B$, and $C$ are compressions of $M$, this implies $\norm{A}\le 1$, $\norm{B}\le 1$, and $\norm{C}\le 1$ as well.

\subsection{Block-encodings}
\begin{definition}[Block-encoding]
Let $T$ be an operator on $s$ qubits.
A unitary $U_T$ acting on $a+s$ qubits is an $(\alpha,a,\eps)$-block-encoding of $T$ if
\[
\norm{T - \alpha (\langle 0^a\!\mid \otimes \I) U_T (\mid 0^a\rangle \otimes \I)} \le \eps.
\]
\end{definition}

\begin{assumption}[Unit-normalized block-encoding access]\label{ass:be}
We are given exact unit block-encodings
\[
U_A \in \BE(1,a_A,0),\qquad U_B \in \BE(1,a_B,0),\qquad U_C \in \BE(1,a_C,0)
\]
of $A$, $B$, and $C$, respectively.
Controlled versions and adjoints of these unitaries are available at the same asymptotic query cost.
\end{assumption}

We use standard block-encoding product and LCU rules throughout; a concise summary is recorded in \cref{app:be-calculus}.
In the present unit-normalized input model, attaching an input oracle contributes no additional normalization factor beyond the explicit factors created by the construction itself.

\subsection{The QSVT inverse primitive}
Given an exact unit block-encoding of an invertible matrix, quantum singular value transformation (QSVT) produces a block-encoding of its inverse with linear dependence on $\norm{T^{-1}}$ and logarithmic dependence on the target precision.  The construction uses the Hermitian dilation and therefore also applies when $T$ is non-Hermitian or non-normal.

\begin{proposition}[QSVT inverse for a unit-normalized block-encoding]\label{prop:qsvt-inverse}
Let $U_T$ be a $(1,a,0)$-block-encoding of an invertible matrix $T$.
Then for every target operator-norm error $\eps_{\mathrm{inv}}\in(0,1]$ there exists a QSVT circuit $V_T$ that is a
\begin{equation}\label{eq:inverse-output-scaling}
(2\norm{T^{-1}}, a+O(1), \eps_{\mathrm{inv}})\text{-block-encoding of }T^{-1}
\end{equation}
using
\begin{equation}\label{eq:qsvt-query-cost}
O\!\paren{\norm{T^{-1}} \log \frac{\norm{T^{-1}}}{\eps_{\mathrm{inv}}}}
\end{equation}
queries to $U_T$ and $U_T^\dagger$.
\end{proposition}

\begin{proof}
Set
\[
\mathcal H(T):=
\begin{bmatrix}
0&T\\
T^*&0
\end{bmatrix}.
\]
A standard one-qubit selector construction, using $U_T$ on one branch and $U_T^\dagger$ on the other, gives an exact unit block-encoding of $\mathcal H(T)$ with $a+O(1)$ ancillas and $O(1)$ queries to $U_T$ and $U_T^\dagger$.  The eigenvalues of $\mathcal H(T)$ are $\pm\sigma_j(T)$, hence
\[
\norm{\mathcal H(T)^{-1}}=\norm{T^{-1}},
\qquad
\mathcal H(T)^{-1}=
\begin{bmatrix}
0&(T^{-1})^*\\
T^{-1}&0
\end{bmatrix}.
\]

Let $\delta:=\sigma_{\min}(T)=1/\norm{T^{-1}}$.  Applying the Hermitian QSVT inverse theorem of Gily\'en \emph{et al.}\cite{Gilyen2018} to the unit block-encoding of $\mathcal H(T)$ yields, for any scaled target error $\eta\in(0,1]$, a QSVT circuit whose encoded block approximates
\[
\frac{\delta}{2}\mathcal H(T)^{-1}
=
\frac{1}{2\norm{T^{-1}}}\mathcal H(T)^{-1}
\]
with error at most $\eta$, using $O(\delta^{-1}\log(1/\eta))$ queries.  Taking the $(2,1)$ block of the encoded operator gives the same scaled approximation to $T^{-1}$.  Choosing
\[
\eta=\min\left\{1,\frac{\eps_{\mathrm{inv}}}{2\norm{T^{-1}}}\right\}
\]
therefore gives operator-norm error at most $\eps_{\mathrm{inv}}$ on $T^{-1}$, which proves \eqref{eq:inverse-output-scaling} and \eqref{eq:qsvt-query-cost}.
\end{proof}

\subsection{Two regimes of interest}
We use two sets of assumptions. In regime $(\mathrm{R1})$, a suitable shift-rotation places the numerical ranges of $A$ and $-B$ in opposite half-planes with margin $\mu$, as shown in \fig{my_image}. The resulting uniform bounds on the shifted resolvents give explicit complexity estimates for the log-sinc discretization. Regime $(\mathrm{R2})$ assumes that the augmented matrix $M$ is separated from the imaginary axis and that its resolvent is bounded by a known constant $\gamma$ on a vertical strip around $i\R$. This regime includes non-Hermitian and non-normal inputs for which a FoV gap is unavailable, with $\gamma$ retained in the complexity bounds.

A FoV gap is obtained from the numerical ranges and bounds the resolvent through the distance to those ranges. This estimate remains valid for defective matrices, so regime $(\mathrm{R1})$ covers normal, non-normal, diagonalizable, and defective inputs. The banded Sylvester refinement and the accretive square-root and geometric-mean corollaries use the same estimate.

\begin{assumption}[Field-of-values (FoV) gap regime $(\mathrm{R1})$]\label{ass:R1}
After the shift-rotation and common scaling described above, the transformed matrices satisfy
\begin{equation}\label{eq:R1ass}
H(A) \succeq \mu \I,
\qquad
H(B) \succeq \mu \I,
\qquad
\norm{M}\le 1
\end{equation}
for a known $\mu\in(0,1]$.
\end{assumption}

\begin{figure}[htbp]
  \centering
  \includegraphics[width=1.0\textwidth]{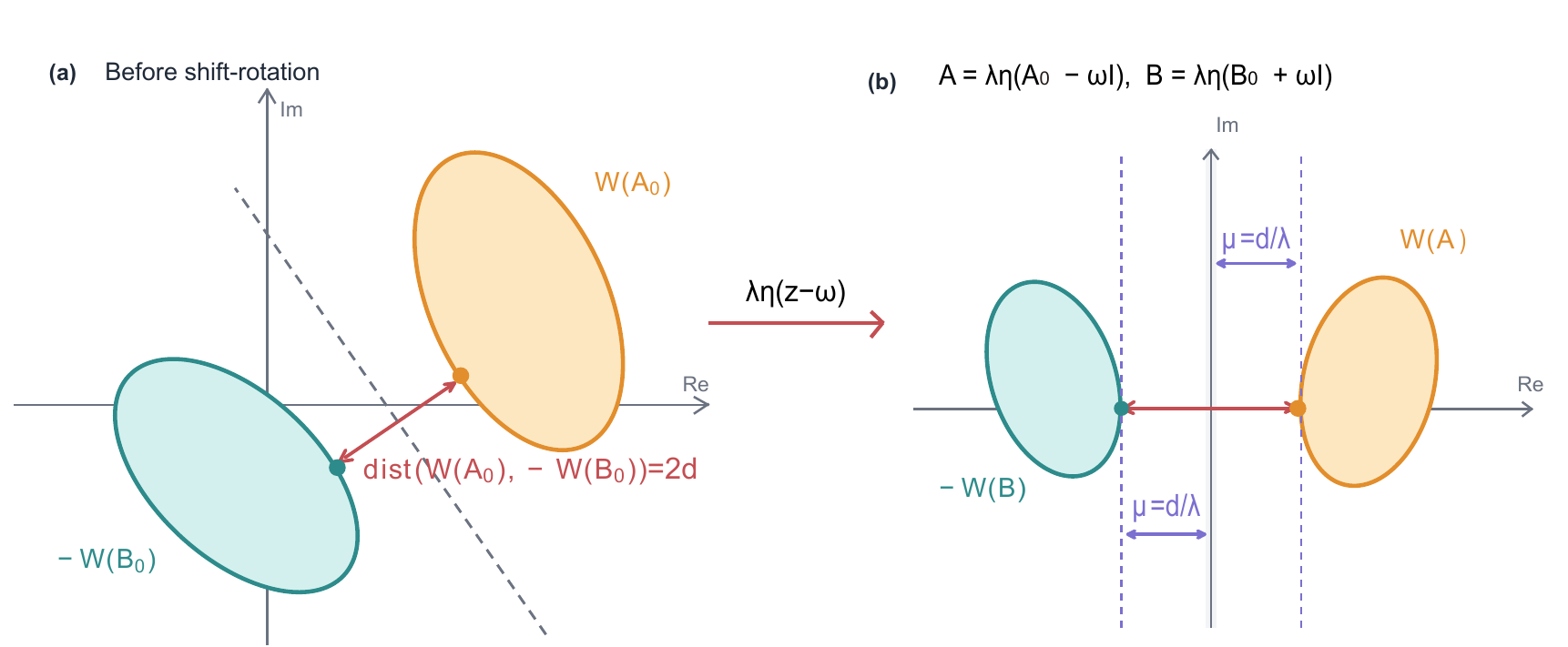} 
  \caption{shift-rotation and common scaling}
  \label{fig:my_image}
\end{figure}

\begin{assumption}[Strip-resolvent regime $(\mathrm{R2})$]\label{ass:R2}
The transformed augmented matrix $M$ satisfies $\spec(M)\cap i\R = \varnothing$, $\norm{M}\le 1$, and there exist $a\in(0,1)$ and a known finite constant
\begin{equation}\label{eq:strip-gamma}
\gamma \ge \gamma_{\Omega_a}(M) := \sup_{z\in\Omega_a} \norm{(z\I-M)^{-1}},
\qquad
\Omega_a := \set{ z\in\C \suchthat \abs{\RePart z}\le a }.
\end{equation}
\end{assumption}

In regime $(\mathrm{R1})$, the strip constant follows from the FoV assumptions. \Cref{lem:R1-strip-bound} proves the explicit bound
\[
\gamma_{\Omega_a}(M)\le \frac{2}{\mu-a}+\frac{\norm{C}}{(\mu-a)^2},
\qquad 0<a<\mu.
\]
Thus the sign-approximation theorem applies in regime $(\mathrm{R1})$ with $\gamma$ determined by $\mu$, $a$, and $\norm{C}$.

\section{Sign embedding representation for ordinary Sylvester equations}\label{sec:sign-reduction}
For the ordinary Sylvester equation, the solution operator is the $(1,2)$ block of the sign of an augmented matrix. The representation is established below.

\subsection{Half-plane sign}
\begin{definition}[Half-plane sign]
For a scalar $z\in\C$ with $\RePart z \neq 0$, define
\[
\signf(z) = \begin{cases}
+1, & \RePart z > 0,\\
-1, & \RePart z < 0.
\end{cases}
\]
If $T$ is a matrix with $\spec(T)\cap i\R = \varnothing$, then $\signf(T)$ is defined by holomorphic functional calculus on any domain containing $\spec(T)$ and disjoint from $i\R$.
\end{definition}

\begin{remark}
Equivalently,
\[
\signf(T) = T(T^2)^{-1/2},
\]
where the square root uses the principal branch.
In particular, $\signf(T)^2 = \I$ and $\signf(T)$ commutes with $T$.
\end{remark}

\subsection{Higham's sign function identity}
\begin{theorem}[Higham sign function identity]\label{thm:higham}
Let
\[
M = \begin{bmatrix} A & C \\ 0 & -B \end{bmatrix}.
\]
Assume that $\spec(A)\subset \set{z\suchthat \RePart z>0}$ and $\spec(-B)\subset \set{z\suchthat \RePart z<0}$.
Then
\[
\signf(M) = \begin{bmatrix} \I & 2X \\ 0 & -\I \end{bmatrix},
\]
where $X$ is the unique solution of \eqref{eq:sylvester}.
\end{theorem}

\begin{proof}
Because $M$ is block upper triangular and $\signf$ is holomorphic near $\spec(M)$, the matrix $\signf(M)$ is block upper triangular with diagonal blocks $\signf(A)=\I$ and $\signf(-B)=-\I$.
So
\[
\signf(M) = \begin{bmatrix} \I & Y \\ 0 & -\I \end{bmatrix}
\]
for some $Y$.
Since $\signf(M)$ is a holomorphic function of $M$, it commutes with $M$.
Comparing the $(1,2)$-blocks of $M\signf(M)$ and $\signf(M)M$ gives
\[
AY + YB = 2C.
\]
Hence $Y=2X$, where $X$ solves $AX+XB=C$.
Uniqueness follows from the half-plane separation, which implies $\spec(A)\cap(-\spec(B))=\varnothing$.
\end{proof}

\subsection{Shift-rotation invariance}
\begin{lemma}[Shift-rotation invariance]\label{lem:shift-rotation}
Let $\eta\in\C$ satisfy $\abs{\eta}=1$, and let $\omega\in\C$, $\lambda>0$.
Define
\[
A' := \lambda\eta(A-\omega\I),\qquad B' := \lambda\eta(B+\omega\I),\qquad C' := \lambda\eta C.
\]
Then $X$ solves $AX+XB=C$ if and only if it solves $A'X+XB'=C'$.
Moreover, if
\[
M' := \begin{bmatrix} A' & C' \\ 0 & -B' \end{bmatrix},
\]
then $M' = \lambda\eta(M-\omega \I)$.
In particular, whenever $\signf(M')$ is defined and its diagonal blocks lie in opposite open half-planes, the $(1,2)$-block of $\signf(M')$ is again $2X$.
\end{lemma}

\begin{proof}
The equation
\[
(A-\omega\I)X + X(B+\omega\I) = AX+XB
\]
is immediate, and multiplying by $\eta$ gives the transformed Sylvester equation.
The block identity for $M'$ is obvious.
The last statement then follows from \cref{thm:higham}.
\end{proof}

\section{Log-sinc approximation of the sign operator}\label{sec:sinc}
We approximate the half-plane sign by logarithmic-sinc quadrature. The asymptotic statement is followed by explicit constants and tail estimates used in the applications.

After the substitution $t=e^x$, the sign integral has a holomorphic operator-valued integrand on a horizontal strip and exponentially decaying tails as $x\to\pm\infty$. The trapezoidal rule then converges exponentially; the proof obtains a quantitative bound from contour shifts and Poisson summation. The geometric nodes $t_k=e^{kh}$ also match the scaled shifted-inverse identities, which keeps the LCU coefficient sums bounded for the Sylvester, square-root, and geometric-mean formulas.

\subsection{Integral formula and logarithmic resolvent sum}
For scalars with nonzero real part,
\[
\signf(z) = \frac{2}{\pi}\int_0^{\infty} \frac{z}{z^2+t^2}\,dt.
\]
Hence, if $\spec(M)\cap i\R = \varnothing$, then holomorphic functional calculus yields
\begin{equation}\label{eq:sign-integral}
\signf(M) = \frac{2}{\pi}\int_0^{\infty} M(M^2+t^2\I)^{-1} \, dt.
\end{equation}
Using the partial fraction identity
\[
M(M^2+t^2\I)^{-1} = \half\paren{ (M-it\I)^{-1} + (M+it\I)^{-1} },
\]
we can rewrite the sign integral on the logarithmic line.
For $z\in\C$, define
\begin{equation}\label{eq:FM-def}
F_M(z)
:=
\frac1\pi e^z\Big((M-i e^z \I)^{-1}+(M+i e^z \I)^{-1}\Big)
=
\frac{2}{\pi}e^z M(M^2+e^{2z}\I)^{-1}.
\end{equation}
For a mesh width $h>0$ and an integer $K\ge 1$, define the logarithmic sinc nodes
\begin{equation}\label{eq:logsinc-nodes}
t_k := e^{kh},\qquad k=-K,-K+1,\dots,K,
\end{equation}
and the truncated log-sinc approximant
\begin{equation}\label{eq:SKh}
S_{K,h}(M)
:=
h\sum_{k=-K}^{K}F_M(kh)
=
\frac{h}{\pi}\sum_{k=-K}^{K} t_k
\Big((M-it_k\I)^{-1}+(M+it_k\I)^{-1}\Big).
\end{equation}

The asymptotic estimate is stated next. Explicit constants and the full strip and tail bounds are given in \cref{prop:sign-approx-explicit}.

\begin{theorem}[Log-sinc sign approximation]\label{thm:sign-approx}
Assume $\norm{M}\le 1$, fix $a\in(0,1)$, and let
\[
\gamma \ge \gamma_{\Omega_a}(M):=\sup_{z\in\Omega_a}\norm{(z\I-M)^{-1}}<\infty,
\qquad
\Omega_a:=\set{z\in\C\suchthat \abs{\RePart z}\le a}.
\]
Choose any $\beta\in(0,\arcsin a)$.
Then there is an explicit constant $C_{\mathrm{sgn}}(\beta,a,\gamma)>0$, given in \cref{prop:sign-approx-explicit}, such that under the balanced rule \eqref{eq:balanced-h},
\begin{equation}\label{eq:sign-approx-clean}
\norm{\signf(M)-S_{K,h}(M)}
\le
C_{\mathrm{sgn}}(\beta,a,\gamma)\,e^{-\sqrt{2\pi\beta K}}.
\end{equation}
Consequently, to guarantee deterministic sign-approximation error at most $\eps_{\mathrm{sgn}}$, it suffices to take
\begin{equation}\label{eq:Kchoice-sinc-asymptotic}
K = O\!\paren{\frac{1}{\beta}\log^2\frac{C_{\mathrm{sgn}}(\beta,a,\gamma)}{\eps_{\mathrm{sgn}}}},
\end{equation}
and the explicit sufficient inequality used later is \eqref{eq:Kchoice-sinc}.
\end{theorem}

\begin{proof}
\Cref{prop:sign-approx-explicit}(3) gives \eqref{eq:sign-approx-clean} with
$C_{\mathrm{sgn}}(\beta,a,\gamma)=\wt C_{\beta,a,\gamma}$.
Solving
\[
\wt C_{\beta,a,\gamma}e^{-\sqrt{2\pi\beta K}}\le \eps_{\mathrm{sgn}}
\]
for $K$ yields \eqref{eq:Kchoice-sinc-asymptotic}, while \eqref{eq:Kchoice-sinc} is the corresponding explicit sufficient choice.
\end{proof}

\begin{proposition}[Explicit strip and tail bounds for the log-sinc rule]\label{prop:sign-approx-explicit}
Assume $\norm{M}\le 1$, fix $a\in(0,1)$, and let
\[
\gamma \ge \gamma_{\Omega_a}(M):=\sup_{z\in\Omega_a}\norm{(z\I-M)^{-1}}<\infty,
\qquad
\Omega_a:=\set{z\in\C\suchthat \abs{\RePart z}\le a}.
\]
Fix any
\[
0<\beta<\arcsin a,
\qquad
\Sigma_\beta := \set{z\in\C\suchthat \abs{\ImPart z}<\beta}.
\]
Define
\begin{equation}\label{eq:Cbetaa}
C_{\beta,a,\gamma}
:=
\frac{4}{\pi}\left(\frac{a\gamma}{\sin\beta}
+\frac{\sin\beta}{a-\sin\beta}\right)
\end{equation}
and
\begin{equation}\label{eq:eps-sgn-sinc}
\eps_{\mathrm{sgn}}^{\mathrm{sinc}}(K,h;\beta,a,\gamma)
:=
\frac{C_{\beta,a,\gamma}}{e^{2\pi\beta/h}-1}
+\frac{2\gamma}{\pi}e^{-Kh}
+\frac{2}{\pi(e^{Kh}-1)}.
\end{equation}
Then the following hold.
\begin{enumerate}[leftmargin=2em]
\item The function $F_M$ is holomorphic on $\Sigma_\beta$, and
\begin{equation}\label{eq:sign-as-log-integral}
\signf(M)=\int_{-\infty}^{\infty}F_M(x)\,dx.
\end{equation}
\item For every $h>0$ and $K\ge 1$,
\begin{equation}\label{eq:sign-approx-bound}
\norm{\signf(M)-S_{K,h}(M)}
\le
\eps_{\mathrm{sgn}}^{\mathrm{sinc}}(K,h;\beta,a,\gamma).
\end{equation}
\item If one chooses the balanced step
\begin{equation}\label{eq:balanced-h}
h=\sqrt{\frac{2\pi\beta}{K}},
\qquad
s:=Kh=\sqrt{2\pi\beta K},
\end{equation}
and if $s\ge \log 2$, then
\begin{equation}\label{eq:balanced-sign-bound}
\norm{\signf(M)-S_{K,h}(M)}
\le
\wt C_{\beta,a,\gamma}e^{-s},
\end{equation}
where
\begin{equation}\label{eq:Ctilde}
\wt C_{\beta,a,\gamma}
:=
2C_{\beta,a,\gamma}+\frac{2\gamma+4}{\pi}
=
\frac{8}{\pi}\left(\frac{a\gamma}{\sin\beta}
+\frac{\sin\beta}{a-\sin\beta}\right)
+\frac{2\gamma+4}{\pi}.
\end{equation}
Consequently, it suffices to choose
\begin{equation}\label{eq:Kchoice-sinc}
K\ge
\left\lceil
\frac{1}{2\pi\beta}
\max\!\left\{
(\log 2)^2,
\log^2\!\left(1+\frac{\wt C_{\beta,a,\gamma}}{\eps_{\mathrm{sgn}}}\right)
\right\}
\right\rceil
\end{equation}
and $h=\sqrt{2\pi\beta/K}$ to ensure
\[
\norm{\signf(M)-S_{K,h}(M)}\le \eps_{\mathrm{sgn}}.
\]
\end{enumerate}
\end{proposition}

\begin{proof}
All operator-valued contour integrals and Fourier transforms can be performed entrywise.
Because the matrix space is finite-dimensional, the resulting statements are equivalent to the corresponding operator-valued ones, and
\[
\norm{\int G(z)\,dz}\le \int \norm{G(z)}\,|dz|
\]
is available throughout.\\
\smallskip
\noindent\emph{Integral representation on the logarithmic line.}
From \eqref{eq:sign-integral} and the identity
\[
M(M^2+t^2\I)^{-1}
=
\half\Big((M-it\I)^{-1}+(M+it\I)^{-1}\Big),
\]
substitute $t=e^x$, $dt=e^x\,dx$, to obtain
\[
\signf(M)
=
\frac1\pi\int_{-\infty}^{\infty}
 e^x\Big((M-i e^x\I)^{-1}+(M+i e^x\I)^{-1}\Big)\,dx
=
\int_{-\infty}^{\infty}F_M(x)\,dx.
\]
This proves \eqref{eq:sign-as-log-integral}.

\smallskip
\noindent\emph{Holomorphy of $F_M$ on $\Sigma_\beta$.}
It suffices to prove that $M\pm i e^z\I$ is invertible for every $z=x+iy\in\Sigma_\beta$.
Suppose, to the contrary, that $M\pm i e^z\I$ is singular.
Then there exists $\lambda\in\spec(M)$ such that $\lambda=\mp i e^z$, hence
\[
\abs{\RePart \lambda}=e^x\abs{\sin y}=\abs{\lambda}\,\abs{\sin y}
\le \abs{\lambda}\sin\beta
\le \norm{M}\sin\beta
\le \sin\beta
< a.
\]
But $\gamma<\infty$ implies $\spec(M)\cap\Omega_a=\varnothing$, a contradiction.
Therefore $F_M$ is holomorphic on $\Sigma_\beta$.

\smallskip
\noindent\emph{Two pointwise bounds.}
Let $z=x+iy\in\Sigma_\beta$.
If
\begin{equation}\label{eq:strip-case}
e^x\abs{\sin y}\le a,
\end{equation}
then $\abs{\RePart(\pm i e^z)}=e^x\abs{\sin y}\le a$, so $\pm i e^z\in\Omega_a$.
Hence
\[
\norm{(M\mp i e^z\I)^{-1}} = \norm{(\pm i e^z\I-M)^{-1}}\le \gamma,
\]
and therefore
\begin{equation}\label{eq:FM-bound-A}
\norm{F_M(z)}
\le
\frac1\pi e^x\Big(\gamma+\gamma\Big)
=
\frac{2\gamma}{\pi}e^x.
\end{equation}
If instead $e^x>1$, then by
\[
(M-it\I)^{-1}+(M+it\I)^{-1}=2(M-it\I)^{-1}M(M+it\I)^{-1}
\]
with $t=e^z$,
\[
F_M(z)=\frac{2}{\pi}e^z(M-i e^z\I)^{-1}M(M+i e^z\I)^{-1}.
\]
Using $\norm{M}\le 1$ and the Neumann estimate
\[
\norm{(M\pm i e^z\I)^{-1}}\le \frac{1}{|e^z|-\norm{M}}\le \frac{1}{e^x-1},
\]
we obtain
\begin{equation}\label{eq:FM-bound-B}
\norm{F_M(z)}
\le
\frac{2}{\pi}e^x\frac{\norm{M}}{(e^x-1)^2}
\le
\frac{2}{\pi}\frac{e^x}{(e^x-1)^2}.
\end{equation}
In particular, $F_M|_{\R}\in L^1(\R)$, and the same bounds show that the boundary restrictions $F_M(x\pm i\beta)$ are integrable.

\smallskip
\noindent\emph{Boundary $L^1$ bounds on $\ImPart z=\pm\beta$.}
Set
\[
x_\beta:=\log\!\frac{a}{\sin\beta}>0.
\]
For
\[
B_+ := \int_{-\infty}^{\infty}\norm{F_M(x+i\beta)}\,dx,
\qquad
B_- := \int_{-\infty}^{\infty}\norm{F_M(x-i\beta)}\,dx,
\]
we split the integral at $x_\beta$.
If $x\le x_\beta$, then $e^x\sin\beta\le a$, so \eqref{eq:FM-bound-A} gives
\[
\norm{F_M(x\pm i\beta)}\le \frac{2\gamma}{\pi}e^x.
\]
If $x\ge x_\beta$, then $e^x>1$, so \eqref{eq:FM-bound-B} gives
\[
\norm{F_M(x\pm i\beta)}\le \frac{2}{\pi}\frac{e^x}{(e^x-1)^2}.
\]
Therefore
\begin{align}
B_\pm
&\le
\frac{2\gamma}{\pi}\int_{-\infty}^{x_\beta}e^x\,dx
+
\frac{2}{\pi}\int_{x_\beta}^{\infty}\frac{e^x}{(e^x-1)^2}\,dx \nonumber\\
&=
\frac{2a\gamma}{\pi\sin\beta}
+
\frac{2\sin\beta}{\pi(a-\sin\beta)}.
\label{eq:Bpm-bound}
\end{align}
Summing the two bounds yields
\begin{equation}\label{eq:Bpm-sum}
B_+ + B_-
\le
\frac{4}{\pi}\left(\frac{a\gamma}{\sin\beta}
+\frac{\sin\beta}{a-\sin\beta}\right)
=
C_{\beta,a,\gamma}.
\end{equation}

\smallskip
\noindent\emph{The infinite trapezoidal discretization error.}
Define the Fourier transform
\[
\widehat F_M(\xi):=\int_{-\infty}^{\infty} e^{-i\xi x}F_M(x)\,dx.
\]
Then $\widehat F_M(0)=\signf(M)$ by \eqref{eq:sign-as-log-integral}.
For $\xi>0$, shift the contour from the real axis to $\ImPart z=-\beta$.
The vertical sides vanish as the rectangle width tends to infinity, because \eqref{eq:FM-bound-B} gives $\norm{F_M(R+iy)}=O(e^{-R})$ as $R\to\infty$, while for sufficiently large $R$, \eqref{eq:FM-bound-A} gives $\norm{F_M(-R+iy)}=O(e^{-R})$.
Thus
\[
\widehat F_M(\xi)
=
 e^{-\beta\xi}
\int_{-\infty}^{\infty} e^{-i\xi x}F_M(x-i\beta)\,dx,
\qquad \xi>0,
\]
so
\begin{equation}\label{eq:Fhat-plus}
\norm{\widehat F_M(\xi)}\le e^{-\beta\xi}B_-,\qquad \xi>0.
\end{equation}
Similarly, for $\xi<0$, shifting to $\ImPart z=+\beta$ gives
\begin{equation}\label{eq:Fhat-minus}
\norm{\widehat F_M(\xi)}\le e^{-\beta|\xi|}B_+,\qquad \xi<0.
\end{equation}
Now define the $h$-periodic function
\[
G_h(x):=\sum_{m\in\mathbb Z}F_M(x+mh).
\]
The real-axis decay from part~3 implies absolute and uniform convergence on $[0,h]$, so $G_h$ is continuous and $h$-periodic.
Its Fourier coefficients are
\[
\frac1h\int_0^h G_h(x)e^{-2\pi i n x/h}\,dx
=
\frac1h\widehat F_M\!\left(\frac{2\pi n}{h}\right).
\]
Because \eqref{eq:Fhat-plus}--\eqref{eq:Fhat-minus} give exponential decay, the Fourier series converges absolutely, and evaluating it at $x=0$ yields the Poisson summation formula
\[
h\sum_{k\in\mathbb Z}F_M(kh)
=
\sum_{n\in\mathbb Z}\widehat F_M\!\left(\frac{2\pi n}{h}\right).
\]
Write
\[
S_h^{\infty}(M):=h\sum_{k\in\mathbb Z}F_M(kh).
\]
Then
\begin{align*}
\norm{\signf(M)-S_h^{\infty}(M)}
&=
\left\|\widehat F_M(0)-\sum_{n\in\mathbb Z}\widehat F_M\!\left(\frac{2\pi n}{h}\right)\right\|\\
&\le
\sum_{n=1}^{\infty}\left\|\widehat F_M\!\left(\frac{2\pi n}{h}\right)\right\|
+
\sum_{n=1}^{\infty}\left\|\widehat F_M\!\left(-\frac{2\pi n}{h}\right)\right\|\\
&\le
\sum_{n=1}^{\infty} e^{-2\pi\beta n/h}(B_-+B_+)\\
&=
\frac{B_+ + B_-}{e^{2\pi\beta/h}-1}.
\end{align*}
Using \eqref{eq:Bpm-sum}, we obtain
\begin{equation}\label{eq:infinite-sinc-error}
\norm{\signf(M)-S_h^{\infty}(M)}
\le
\frac{C_{\beta,a,\gamma}}{e^{2\pi\beta/h}-1}.
\end{equation}

\smallskip
\noindent\emph{Truncation from the infinite sum to $2K+1$ nodes.}
Since
\[
S_h^{\infty}(M)-S_{K,h}(M)
=
 h\sum_{|k|>K}F_M(kh),
\]
we have
\[
\norm{S_h^{\infty}(M)-S_{K,h}(M)}
\le
h\sum_{k=K+1}^{\infty}\norm{F_M(kh)}
+
 h\sum_{k=K+1}^{\infty}\norm{F_M(-kh)}.
\]
For the negative half-axis, \eqref{eq:FM-bound-A} on the real line gives
\[
\norm{F_M(-kh)}\le \frac{2\gamma}{\pi}e^{-kh},
\]
whence
\begin{equation}\label{eq:negative-tail-sinc}
h\sum_{k=K+1}^{\infty}\norm{F_M(-kh)}
\le
\frac{2\gamma}{\pi}
 h\sum_{k=K+1}^{\infty}e^{-kh}
\le
\frac{2\gamma}{\pi}\int_{Kh}^{\infty}e^{-x}\,dx
=
\frac{2\gamma}{\pi}e^{-Kh}.
\end{equation}
For the positive half-axis, define
\[
q(x):=\frac{e^x}{(e^x-1)^2},\qquad x>0.
\]
Then \eqref{eq:FM-bound-B} gives $\norm{F_M(kh)}\le \frac{2}{\pi}q(kh)$, and
\[
q'(x)= -\frac{e^x(e^x+1)}{(e^x-1)^3}<0,
\]
so $q$ is decreasing.  Hence
\begin{equation}\label{eq:positive-tail-sinc}
h\sum_{k=K+1}^{\infty}\norm{F_M(kh)}
\le
\frac{2}{\pi}h\sum_{k=K+1}^{\infty} q(kh)
\le
\frac{2}{\pi}\int_{Kh}^{\infty} q(x)\,dx
=
\frac{2}{\pi(e^{Kh}-1)}.
\end{equation}
Combining \eqref{eq:infinite-sinc-error}, \eqref{eq:negative-tail-sinc}, and \eqref{eq:positive-tail-sinc} proves \eqref{eq:sign-approx-bound}.

\smallskip
\noindent\emph{Balanced choice.}
Under \eqref{eq:balanced-h}, we have $2\pi\beta/h=Kh=s$.
Substituting into \eqref{eq:sign-approx-bound} gives
\[
\norm{\signf(M)-S_{K,h}(M)}
\le
\frac{C_{\beta,a,\gamma}}{e^s-1}
+\frac{2\gamma}{\pi}e^{-s}
+\frac{2}{\pi(e^s-1)}.
\]
If $s\ge \log 2$, then $e^s\ge 2$, so $(e^s-1)^{-1}\le 2e^{-s}$.
Therefore
\[
\norm{\signf(M)-S_{K,h}(M)}
\le
\left(2C_{\beta,a,\gamma}+\frac{2\gamma+4}{\pi}\right)e^{-s}
=
\wt C_{\beta,a,\gamma}e^{-s},
\]
which is \eqref{eq:balanced-sign-bound}.
Finally, if $s\ge \log(1+\wt C_{\beta,a,\gamma}/\eps_{\mathrm{sgn}})$, then
\[
\wt C_{\beta,a,\gamma}e^{-s}
\le
\wt C_{\beta,a,\gamma}
\left(1+\frac{\wt C_{\beta,a,\gamma}}{\eps_{\mathrm{sgn}}}\right)^{-1}
< \eps_{\mathrm{sgn}}.
\]
Since $s^2=2\pi\beta K$, this is exactly \eqref{eq:Kchoice-sinc}.
\end{proof}

\subsection{Resolvent factorization and the Sylvester approximant}
\begin{lemma}[Resolvent block factorization]\label{lem:res-factor}
Let
\[
M = \begin{bmatrix} A & C \\ 0 & -B \end{bmatrix},
\]
and let $z\in\C$ be such that $z\I-A$ and $z\I+B$ are invertible.
Then
\begin{equation}\label{eq:block-factor}
\bracks{(z\I-M)^{-1}}_{12} = (z\I-A)^{-1} C (z\I+B)^{-1}.
\end{equation}
\end{lemma}

\begin{proof}
Direct block inversion gives
\[
(z\I-M)^{-1} = \begin{bmatrix}
(z\I-A)^{-1} & (z\I-A)^{-1} C (z\I+B)^{-1} \\
0 & (z\I+B)^{-1}
\end{bmatrix}.
\]
\end{proof}

Assume from now on that the transformed matrices satisfy the hypotheses of \cref{thm:higham}, so that $[\signf(M)]_{12}=2X$.
Define the unscaled log-sinc Sylvester approximant
\begin{equation}\label{eq:XKh-unscaled}
X_{K,h} := \frac{h}{2\pi}\sum_{k=-K}^{K} t_k \paren{
(A-it_k\I)^{-1} C (B+it_k\I)^{-1} + (A+it_k\I)^{-1} C (B-it_k\I)^{-1}
}.
\end{equation}

\begin{corollary}[Deterministic log-sinc approximation of the Sylvester solution]\label{cor:XN-error}
Assume the half-plane separation of \cref{thm:higham}, and let $M$ satisfy the hypotheses of \cref{thm:sign-approx} with parameters $(a,\beta,\gamma)$.
Then
\begin{equation}\label{eq:XN-error}
\norm{X-X_{K,h}}
\le
\half\norm{\signf(M)-S_{K,h}(M)}
\le
\half\eps_{\mathrm{sgn}}^{\mathrm{sinc}}(K,h;\beta,a,\gamma).
\end{equation}
In particular, under the balanced choice \eqref{eq:balanced-h},
\begin{equation}\label{eq:X-balanced-error}
\norm{X-X_{K,h}}
\le
\half\wt C_{\beta,a,\gamma}e^{-\sqrt{2\pi\beta K}}.
\end{equation}
\end{corollary}

\begin{proof}
By \cref{lem:res-factor}, the $(1,2)$-block of $S_{K,h}(M)$ is exactly $2X_{K,h}$, while \cref{thm:higham} gives $[\signf(M)]_{12}=2X$.
Therefore
\[
2\norm{X-X_{K,h}} = \norm{\bracks{\signf(M)-S_{K,h}(M)}_{12}} \le \norm{\signf(M)-S_{K,h}(M)}.
\]
The bounds follow from \cref{prop:sign-approx-explicit}.
\end{proof}

\section{Scaled-multiplexing for rational approximants}\label{sec:scaled-multiplexing}
Two implementations of the rational sums are used below. Nodewise rebalancing handles the bilinear shifted-inverse sums arising from Sylvester sign blocks. A separate construction handles sums from one inverse family, as needed for square roots, geometric means, and Hamiltonian sign approximants. Both use the standard controlled-LCU operations for multiplexed direct sums summarized in \cref{app:be-calculus}. Their normalization bounds retain the nodewise overlap quantities $\Theta_{\rho}^{\mathrm{syl}}$ and $\Theta_\rho^{(1)}$; replacing every nodewise bound by a uniform maximum gives the corresponding worst-case estimates.

Uniform FoV or strip-resolvent bounds therefore suffice for the uniform-profile corollaries. When sharper nodewise information is available, the overlap sums can give smaller output normalizations.

The active quadrature nodes $k=-K,\dots,K$ are stored in a $q=\lceil \log_2(2K+1)\rceil$-qubit register, with inactive states added when necessary.

\subsection{Rebalanced implementation for rational sums}
An LCU implementation of \eqref{eq:XKh-unscaled} would retain the node-dependent weights $ht_k/(2\pi)$. Scaling the shifted inverses leaves a bounded scalar coefficient for each term:
\begin{equation}\label{eq:scaled-identity}
\begin{aligned}
X_{K,h}
&=
\sum_{k=-K}^{K}\frac{h t_k}{2\pi(1+t_k)^2}
\Big(
((1+t_k)(A-it_k\I)^{-1})\,C\,((1+t_k)(B+it_k\I)^{-1})
\\  
&\quad +  
((1+t_k)(A+it_k\I)^{-1})\,C\,((1+t_k)(B-it_k\I)^{-1})
\Big)
\end{aligned}
\end{equation}
Define
\begin{equation}\label{eq:Lambda-syl}
\Lambda^{\mathrm{syl}}_{K,h}
:=
\sum_{k=-K}^{K}\frac{h t_k}{\pi(1+t_k)^2}.
\end{equation}

\begin{lemma}[Exact scaled identity and bounded total weight]\label{lem:constant-weight}
The representation \eqref{eq:scaled-identity} is exact. Moreover,
\begin{equation}\label{eq:constant-weight}
\Lambda^{\mathrm{syl}}_{K,h}
\le
\frac{1}{\pi}\left(1+\frac{h}{4}\right).
\end{equation}
In particular, if $h\le \pi$, then
\begin{equation}\label{eq:constant-weight-balanced}
\Lambda^{\mathrm{syl}}_{K,h}
\le
\Lambda^{\mathrm{syl}}_\star
:=
\frac{1+\pi/4}{\pi}
<0.57.
\end{equation}
\end{lemma}

\begin{proof}
Equality \eqref{eq:scaled-identity} is just \eqref{eq:XKh-unscaled} with each term multiplied and divided by $(1+t_k)^2$.
For the coefficient sum, define
\[
q(x):=\frac{e^x}{(1+e^x)^2}.
\]
Then $q$ is positive, even, decreasing on $[0,\infty)$, and
\[
\frac{h t_k}{\pi(1+t_k)^2}=\frac{h}{\pi}q(kh).
\]
Therefore
\[
\Lambda^{\mathrm{syl}}_{K,h}
=
\frac{h}{\pi}\sum_{k=-K}^{K} q(kh)
\le
\frac{h}{\pi}\sum_{k\in\mathbb Z} q(kh).
\]
Since $q$ is decreasing on $[0,\infty)$,
\[
h\sum_{k=1}^{\infty}q(kh)\le \int_0^{\infty}q(x)\,dx.
\]
Using $q(0)=1/4$ and $\int_0^{\infty}q(x)\,dx=1/2$, we obtain
\[
h\sum_{k\in\mathbb Z} q(kh)
\le
h q(0)+2\int_0^{\infty}q(x)\,dx
=
\frac{h}{4}+1,
\]
which proves \eqref{eq:constant-weight}. The last statement is immediate.
\end{proof}

The rebalanced implementation assumes a classically available nodewise profile. These values determine the controlled rotations in the rebalancing step and may be either exact shifted-resolvent norms or efficiently available upper bounds.

\begin{definition}[Sylvester nodewise inverse profile]\label{def:sylvester-profile}
A \emph{Sylvester nodewise inverse profile} for the nodes $t_k=e^{kh}$ is a collection of positive numbers
\[
\rho_{A,k}^{\pm},\qquad \rho_{B,k}^{\pm}
\qquad (-K\le k\le K)
\]
such that
\begin{equation}\label{eq:syl-profile-valid}
(1+t_k)\norm{(A\pm it_k\I)^{-1}}\le \rho_{A,k}^{\pm},
\qquad
(1+t_k)\norm{(B\pm it_k\I)^{-1}}\le \rho_{B,k}^{\pm}.
\end{equation}
Associated with such a profile, define
\begin{equation}\label{eq:syl-profile-R}
R_A:=\max_{k,\pm}\rho_{A,k}^{\pm},
\qquad
R_B:=\max_{k,\pm}\rho_{B,k}^{\pm},
\end{equation}
and
\begin{equation}\label{eq:Theta-syl-profile}
\Theta_{\rho}^{\mathrm{syl}}
:=
\sum_{k=-K}^{K}\frac{h t_k}{2\pi(1+t_k)^2}
\Bigl(
\rho_{A,k}^{-}\rho_{B,k}^{+}
+
\rho_{A,k}^{+}\rho_{B,k}^{-}
\Bigr).
\end{equation}
\end{definition}

The exact uniform shifted-resolvent maxima are
\begin{equation}\label{eq:rA-rB}
r_A := \max_{k,\pm}(1+t_k)\norm{(A\pm it_k\I)^{-1}},
\qquad
r_B := \max_{k,\pm}(1+t_k)\norm{(B\pm it_k\I)^{-1}}.
\end{equation}
If only uniform upper bounds $r_A,r_B$ are used, the constant profile
\[
\rho_{A,k}^{\pm}\equiv r_A,
\qquad
\rho_{B,k}^{\pm}\equiv r_B
\]
is admissible and gives
\begin{equation}\label{eq:Theta-AB-plain}
\Theta_{\rho}^{\mathrm{syl}}\le r_A r_B \Lambda^{\mathrm{syl}}_{K,h}.
\end{equation}
If the exact profile is used, then $\Theta_{\rho}^{\mathrm{syl}}$ coincides with the overlap-weighted coefficient sum
\[
\sum_{k=-K}^{K}\frac{h t_k}{2\pi}
\Bigl(
\norm{(A-it_k\I)^{-1}}\norm{(B+it_k\I)^{-1}}
+
\norm{(A+it_k\I)^{-1}}\norm{(B-it_k\I)^{-1}}
\Bigr).
\]

\begin{theorem}[Rebalanced implementation for the Sylvester sign sum]\label{thm:profile-implementation}
Let $\rho$ be a classically available Sylvester nodewise inverse profile in the sense of \cref{def:sylvester-profile}, with parameters $R_A$, $R_B$, and $\Theta_{\rho}^{\mathrm{syl}}$ defined by \eqref{eq:syl-profile-R}--\eqref{eq:Theta-syl-profile}.  Let $\eps_A,\eps_B\in(0,1]$.
Then there exists a unitary $U^{\mathrm{reb}}_{X_{K,h}}$ that is a
\begin{equation}\label{eq:profile-be}
\paren{
4\Theta_{\rho}^{\mathrm{syl}},
\ a_A+a_B+a_C+O(\log K),
\ \Theta_{\rho}^{\mathrm{syl}}\norm{C}\left(
\frac{\eps_A}{R_A}
+
\frac{\eps_B}{R_B}
+
\frac{\eps_A\eps_B}{R_A R_B}
\right)
}
\text{-block-encoding of }X_{K,h}.
\end{equation}
The corresponding query costs are
\begin{equation}\label{eq:qAqB-rho}
\qA = O\!\paren{R_A \log\frac{R_A}{\eps_A}},
\qquad
\qB = O\!\paren{R_B \log\frac{R_B}{\eps_B}},
\qquad
\qC = O(1).
\end{equation}
\end{theorem}

\begin{proof}
Define the rebalancing factors
\[
d^\pm_{A,k}:=\frac{R_A}{\rho_{A,k}^{\pm}},
\qquad
d^\pm_{B,k}:=\frac{R_B}{\rho_{B,k}^{\pm}}.
\]
By construction, $d^\pm_{A,k}\ge 1$ and $d^\pm_{B,k}\ge 1$.

Now form the multiplexed direct sums
\[
A_{\rho}^\pm
:=
\sum_{k=-K}^{K}\ket{k}\!\bra{k}\otimes \frac{A\pm it_k\I}{(1+t_k)d^\pm_{A,k}},
\qquad
B_{\rho}^\pm
:=
\sum_{k=-K}^{K}\ket{k}\!\bra{k}\otimes \frac{B\pm it_k\I}{(1+t_k)d^\pm_{B,k}}.
\]
For each sign choice and each $k$, the coefficients of $A$ and $\I$ in the $k$th block of $A_{\rho}^\pm$ have total absolute value $1/d^\pm_{A,k}\le 1$, and similarly for $B_{\rho}^\pm$.
Hence \cref{lem:mux-one-matrix} yields unit-normalized block-encodings of $A_{\rho}^\pm$ and $B_{\rho}^\pm$ using $O(1)$ queries to $U_A$ and $U_B$, respectively.

Their inverse blocks are
\[
(A_{\rho}^\pm)^{-1}
=
\sum_{k=-K}^{K}\ket{k}\!\bra{k}\otimes
d^\pm_{A,k}(1+t_k)(A\pm it_k\I)^{-1},
\]
\[
(B_{\rho}^\pm)^{-1}
=
\sum_{k=-K}^{K}\ket{k}\!\bra{k}\otimes
d^\pm_{B,k}(1+t_k)(B\pm it_k\I)^{-1}.
\]
By \eqref{eq:syl-profile-valid},
\[
\norm{d^\pm_{A,k}(1+t_k)(A\pm it_k\I)^{-1}}\le d^\pm_{A,k}\rho_{A,k}^{\pm}=R_A,
\]
and similarly for $B$, so
\[
\norm{(A_{\rho}^\pm)^{-1}}\le R_A,
\qquad
\norm{(B_{\rho}^\pm)^{-1}}\le R_B.
\]
Applying \cref{prop:qsvt-inverse} gives QSVT block-encodings
\[
V_{A_{\rho}^\pm}\in \BE(2R_A,a_A+O(1),\eps_A),
\qquad
V_{B_{\rho}^\pm}\in \BE(2R_B,a_B+O(1),\eps_B),
\]
with query bounds \eqref{eq:qAqB-rho}.

Define
\[
\lambda_k^+ := \frac{h t_k}{2\pi(1+t_k)^2 d^-_{A,k}d^+_{B,k}},
\qquad
\lambda_k^- := \frac{h t_k}{2\pi(1+t_k)^2 d^+_{A,k}d^-_{B,k}}.
\]
Then \eqref{eq:scaled-identity} becomes
\begin{align}
X_{K,h}
&=
\sum_{k=-K}^{K}\lambda_k^+\,
\Bigl(d^-_{A,k}(1+t_k)(A-it_k\I)^{-1}\Bigr)
\,C\,
\Bigl(d^+_{B,k}(1+t_k)(B+it_k\I)^{-1}\Bigr)
\nonumber\\
&\quad+
\sum_{k=-K}^{K}\lambda_k^-\,
\Bigl(d^+_{A,k}(1+t_k)(A+it_k\I)^{-1}\Bigr)
\,C\,
\Bigl(d^-_{B,k}(1+t_k)(B-it_k\I)^{-1}\Bigr).
\label{eq:scaled-identity-rho}
\end{align}
Moreover,
\begin{align*}
R_A R_B\sum_{k=-K}^{K}\bigl(\abs{\lambda_k^+}+\abs{\lambda_k^-}\bigr)
&=
\sum_{k=-K}^{K}\frac{h t_k}{2\pi(1+t_k)^2}
\Bigl(
\rho_{A,k}^{-}\rho_{B,k}^{+}
+
\rho_{A,k}^{+}\rho_{B,k}^{-}
\Bigr)\\
&=
\Theta_{\rho}^{\mathrm{syl}}.
\end{align*}

Let $\widehat T_k^\pm$ denote the products obtained by replacing the exact inverse blocks in \eqref{eq:scaled-identity-rho} by their QSVT approximants.
By \cref{lem:approx-product},
\[
\norm{\widehat T_k^\pm-T_k^\pm}
\le
\norm{C}\paren{\eps_A R_B + R_A\eps_B+\eps_A\eps_B}.
\]
Weighting by $\lambda_k^\pm$ and summing yields
\[
\norm{\widehat X_{K,h}^{\mathrm{reb}}-X_{K,h}}
\le
\Theta_{\rho}^{\mathrm{syl}}\,
\norm{C}\left(
\frac{\eps_A}{R_A}
+
\frac{\eps_B}{R_B}
+
\frac{\eps_A\eps_B}{R_A R_B}
\right).
\]
Each product term in \eqref{eq:scaled-identity-rho} is implemented by the product of a $(2R_A)$-block-encoding, the unit block-encoding of $C$, and a $(2R_B)$-block-encoding, so its normalization is at most $4R_A R_B$.
Weighted LCU therefore yields the normalization $4\Theta_{\rho}^{\mathrm{syl}}$ and ancilla count $a_A+a_B+a_C+O(\log K)$.
\end{proof}

\begin{corollary}[Uniform-profile specialization]\label{cor:plain-algorithm}
Suppose $r_A,r_B$ are available uniform upper bounds satisfying
\[
(1+t_k)\norm{(A\pm it_k\I)^{-1}}\le r_A,
\qquad
(1+t_k)\norm{(B\pm it_k\I)^{-1}}\le r_B
\]
for all active nodes and both signs.
Then the constant profile in \cref{def:sylvester-profile} and \cref{thm:profile-implementation} yield a unitary $U_{X_{K,h}}$ that is a
\begin{equation}\label{eq:output-beta-generic}
\paren{
4 r_A r_B\Lambda^{\mathrm{syl}}_{K,h},
\ a_A+a_B+a_C+O(\log K),
\ \Lambda^{\mathrm{syl}}_{K,h}\norm{C}\paren{\eps_A r_B + r_A \eps_B + \eps_A \eps_B}
}
\text{-block-encoding of }X_{K,h},
\end{equation}
and in particular
\begin{equation}\label{eq:XhatN-error}
\norm{\widehat{X}_{K,h} - X_{K,h}} \le \Lambda^{\mathrm{syl}}_{K,h}\norm{C}\paren{\eps_A r_B + r_A \eps_B + \eps_A \eps_B}.
\end{equation}
Its query complexity is
\begin{equation}\label{eq:qAqB-generic}
\qA = O\!\paren{r_A \log\frac{r_A}{\eps_A}},
\qquad
\qB = O\!\paren{r_B \log\frac{r_B}{\eps_B}},
\qquad
\qC = O(1).
\end{equation}
\end{corollary}

\begin{proof}
Use the constant profile $\rho_{A,k}^{\pm}\equiv r_A$ and $\rho_{B,k}^{\pm}\equiv r_B$ in \cref{thm:profile-implementation} and apply \eqref{eq:Theta-AB-plain}.
\end{proof}

\subsection{Uniform conditioning bounds}

The next lemmas give uniform shifted-inverse bounds for the FoV gap and strip-resolvent regimes.

\begin{lemma}[FoV gap resolvent bound]\label{lem:accretive-res}
If $H(T)\succeq \mu \I$ with $\mu>0$, then for all $t\in\R$,
\begin{equation}\label{eq:accretive-res}
\norm{(T+it\I)^{-1}} \le \frac{1}{\mu}.
\end{equation}
If in addition $\norm{T}\le 1$, then for $\abs{t}>1$,
\begin{equation}\label{eq:large-t-neumann}
\norm{(T\pm it\I)^{-1}} \le \frac{1}{\abs{t}-1}.
\end{equation}
\end{lemma}

\begin{proof}
For a unit vector $x$,
\[
\RePart\bigl(x^*(T+it\I)x\bigr) = x^*H(T)x \ge \mu,
\]
so $0$ is at distance at least $\mu$ from the numerical range of $T+it\I$, which implies \eqref{eq:accretive-res}.
The large-$\abs{t}$ bound is the usual Neumann-series estimate.
\end{proof}

\begin{lemma}[Uniform conditioning in regime $(\mathrm{R1})$]\label{lem:R1-conditioning}
Under \cref{ass:R1},
\begin{equation}\label{eq:R1-rA-rB}
r_A \le \frac{3}{\mu},
\qquad
r_B \le \frac{3}{\mu}.
\end{equation}
\end{lemma}

\begin{proof}
Fix $t>0$.
Then
\[
(1+t)\norm{(A\pm it\I)^{-1}}
\le \min\set{ \frac{1+t}{\mu},\ \frac{1+t}{t-1} }.
\]
If $t\le 2$, the first term is at most $3/\mu$.
If $t\ge 2$, the second term is at most $3$, hence at most $3/\mu$ because $\mu\le 1$.
Taking the maximum over $k$ and both signs gives \eqref{eq:R1-rA-rB}.
The statement for $B$ is identical.
\end{proof}

\begin{lemma}[Uniform conditioning in regime $(\mathrm{R2})$]\label{lem:R2-conditioning}
Assume \cref{ass:R2} and define
\[
\gamma_A := \sup_{z\in\Omega_a} \norm{(z\I-A)^{-1}},
\qquad
\gamma_B := \sup_{z\in\Omega_a} \norm{(z\I+B)^{-1}}.
\]
Then
\begin{equation}\label{eq:R2-rA-rB}
r_A \le 3\gamma_A,
\qquad
r_B \le 3\gamma_B.
\end{equation}
Moreover, $\gamma_A\le \gamma$ and $\gamma_B\le \gamma$, so in particular
\begin{equation}\label{eq:R2-rA-rB-gamma}
r_A \le 3\gamma,
\qquad
r_B \le 3\gamma.
\end{equation}
\end{lemma}

\begin{proof}
Because $\pm it\in\Omega_a$ whenever $t\le 2$ and $a\in(0,1)$,
\[
\norm{(A\pm it\I)^{-1}} \le \gamma_A,
\qquad
\norm{(B\mp it\I)^{-1}} \le \gamma_B
\qquad (0\le t\le 2).
\]
For $t\ge 2$, the Neumann estimate gives
\[
\norm{(A\pm it\I)^{-1}} \le \frac{1}{t-1},
\qquad
\norm{(B\pm it\I)^{-1}} \le \frac{1}{t-1}.
\]
Hence
\[
(1+t)\norm{(A\pm it\I)^{-1}}
\le \min\set{\frac{1+t}{1}\gamma_A,\frac{1+t}{t-1}},
\]
and similarly for $B$.
Since $\gamma_A\ge \norm{A^{-1}}\ge 1$ and $\gamma_B\ge \norm{B^{-1}}\ge 1$, the large-$t$ bound is already dominated by $3\gamma_A$ and $3\gamma_B$.
Taking the maxima proves \eqref{eq:R2-rA-rB}.
The inequalities $\gamma_A\le\gamma$ and $\gamma_B\le\gamma$ follow because the diagonal blocks of $(z\I-M)^{-1}$ are $(z\I-A)^{-1}$ and $(z\I+B)^{-1}$.
\end{proof}

\subsection{Single-family implementation for rational sums}
Several applications use a rational approximant that is a linear combination of
nodewise transformed inverses from a \emph{single} multiplexed family:
\[
Y_{K,h}=\sum_{j\in J} c_j R_j,
\qquad
R_j:=\sigma_j F_j^{-1},
\]
where each $F_j$ is an invertible node matrix, $\sigma_j>0$ is a scaling factor, and
$R_j$ is the corresponding scaled inverse block.
Equivalently, if
\[
F:=\sum_{j\in J}\ket{j}\!\bra{j}\otimes \frac{F_j}{\sigma_j},
\]
then
\[
F^{-1}=\sum_{j\in J}\ket{j}\!\bra{j}\otimes R_j.
\]
The square-root shifts use $R_j=\mathcal R_k^{\mathrm{sq}}=(1+t_k^2)(A+t_k^2\I)^{-1}$ in \cref{eq:sqrt-scaled-identity}; the geometric-mean shifts use $R_j=(1+t_k^2)(B+t_k^2A)^{-1}$ in the multiplexed form of \cref{eq:geom-mean-inv-approx}; and the Hamiltonian sign sum uses $R_j=(1+t_k)(H\pm it_k\I)^{-1}$ in \cref{eq:care-sign-sum}. The corresponding implementation theorem follows.

\begin{definition}[Single-family nodewise profile]\label{def:single-family-profile}
Let $J$ be a finite index set, and suppose an approximant $Y_{K,h}$ admits the representation
\[
Y_{K,h}=\sum_{j\in J} c_j R_j,
\qquad
R_j=\sigma_j F_j^{-1},
\]
where each $F_j$ is invertible.
Define the multiplexed direct sum
\[
F:=\sum_{j\in J}\ket{j}\!\bra{j}\otimes \frac{F_j}{\sigma_j},
\qquad
F^{-1}=\sum_{j\in J}\ket{j}\!\bra{j}\otimes R_j.
\]
A \emph{single-family nodewise profile} is a collection of numbers $\rho_j>0$ such that
\[
\norm{R_j}\le \rho_j
\qquad
\text{for all }j\in J.
\]
Associated to such a profile, define
\[
R_\rho:=\max_{j\in J}\rho_j,
\qquad
d_j:=\frac{R_\rho}{\rho_j},
\qquad
\Theta_\rho^{(1)}:=\sum_{j\in J}\abs{c_j}\rho_j.
\]
\end{definition}

In algorithmic use, the profile numbers $\rho_j$ are assumed to be classically available upper bounds, since they determine the controlled diagonal rebalancing rotations.  Exact nodewise norms are only one possible choice.

\begin{theorem}[Single-family implementation]\label{thm:single-family-algorithm}
Let the notation be as in \cref{def:single-family-profile}, and let $U_F\in \BE(1,a,0)$ be a block-encoding of $F$.
Define
\[
D_\rho^{-1}:=\sum_{j\in J}\frac{1}{d_j}\ket{j}\!\bra{j},
\qquad
F_\rho:= (D_\rho^{-1}\otimes \I)F
      = \sum_{j\in J}\ket{j}\!\bra{j}\otimes \frac{F_j}{\sigma_j d_j}.
\]
Then:
\begin{enumerate}[leftmargin=2em]
\item $F_\rho$ admits a block-encoding
\[
U_{F_\rho}\in \BE(1,a+O(1),0),
\]
and
\[
F_\rho^{-1}
=
\sum_{j\in J}\ket{j}\!\bra{j}\otimes d_j R_j,
\qquad
\norm{F_\rho^{-1}}\le R_\rho;
\]
\item for every $\eps_{\mathrm{inv}}\in(0,1]$, QSVT yields
\[
V_{F_\rho}\in \BE(2R_\rho,a+O(1),\eps_{\mathrm{inv}})
\]
using
\[
O\!\paren{R_\rho \log\frac{R_\rho}{\eps_{\mathrm{inv}}}}
\]
queries to $U_F$ and $U_F^\dagger$;
\item one has
\[
Y_{K,h}=\sum_{j\in J}\frac{c_j}{d_j}\,(d_jR_j),
\qquad
R_\rho\sum_{j\in J}\abs{\frac{c_j}{d_j}}=\Theta_\rho^{(1)},
\]
and there exists a unitary $U_{Y_{K,h}}^{(1,\rho)}$ that is a
\[
\paren{2\Theta_\rho^{(1)},\, a+O(\log|J|),\, \Theta_\rho^{(1)}\frac{\eps_{\mathrm{inv}}}{R_\rho}}
\text{-block-encoding of }Y_{K,h}.
\]
\end{enumerate}
\end{theorem}

\begin{proof}
Because $0<1/d_j\le 1$ for every $j$, \cref{lem:diag-contraction} gives an exact unit block-encoding of $D_\rho^{-1}$.
The product rule then yields the stated block-encoding of $F_\rho$.
The formula for $F_\rho^{-1}$ is immediate from block diagonality, and
\[
\norm{d_jR_j}\le d_j\rho_j=R_\rho
\qquad
\text{for all }j,
\]
so $\norm{F_\rho^{-1}}\le R_\rho$.

Applying \cref{prop:qsvt-inverse} to $F_\rho$ gives the inverse block-encoding and its query complexity.
The exact identity
\[
Y_{K,h}=\sum_{j\in J}c_jR_j
=
\sum_{j\in J}\frac{c_j}{d_j}\,(d_jR_j)
\]
is immediate, and
\[
R_\rho\sum_{j\in J}\abs{\frac{c_j}{d_j}}
=
\sum_{j\in J}\abs{c_j}\rho_j
=
\Theta_\rho^{(1)}.
\]
Each term in the LCU has normalization at most $2R_\rho$, so the weighted-LCU normalization is $2\Theta_\rho^{(1)}$.
The implementation error is bounded by
\[
\sum_{j\in J}\abs{\frac{c_j}{d_j}}\eps_{\mathrm{inv}}
=
\Theta_\rho^{(1)}\frac{\eps_{\mathrm{inv}}}{R_\rho}.
\]
\end{proof}
\section{Ordinary Sylvester equation}\label{sec:main-results}
The ordinary Sylvester results cover the FoV gap, banded-overlap, and strip-resolvent regimes shown in \fig{sylvester-regimes}. General nodewise estimates are given after these three results.

\begin{figure}[H]
\centering
\begin{tikzpicture}[
  >={Latex[length=2mm]},
  node distance=1.8cm and 2.0cm,
  box/.style={draw, rounded corners, align=center, inner sep=6pt, text width=3.5cm},
  arr/.style={->, thick}
]
\node[box] (r1) {$(\mathrm{R1})$\\FoV gap\\\cref{ass:R1}};
\node[box, right=of r1] (t1) {Theorem~\ref{thm:R1-main}\\uniform-profile FoV guarantee};
\draw[arr] (r1) -- (t1);

\node[box, below=1cm of r1] (r1b) {$(\mathrm{R1})$ + banded overlap\\assumption \eqref{eq:R1-banded-ass}};
\node[box, right=of r1b] (t1b) {Theorem~\ref{thm:R1-banded}\\refined FoV guarantee};
\draw[arr] (r1b) -- (t1b);

\node[box, below=1cm of r1b] (r2) {$(\mathrm{R2})$\\strip-resolvent regime\\\cref{ass:R2}};
\node[box, right=of r2] (t2) {Theorem~\ref{thm:R2-main}\\strip-resolvent guarantee};
\draw[arr] (r2) -- (t2);
\end{tikzpicture}
\caption{The three ordinary Sylvester regimes in \cref{sec:main-results}.}
\label{fig:sylvester-regimes}
\end{figure}

\subsection{An explicit strip-resolvent bound in the FoV gap regime}
\begin{lemma}[Explicit strip-resolvent bound in the FoV gap regime]\label{lem:R1-strip-bound}
Assume \cref{ass:R1} and fix $a\in(0,\mu)$.
Then
\begin{equation}\label{eq:R1-gammaM}
\gamma_{\Omega_a}(M)
\le
\frac{2}{\mu-a}+\frac{\norm{C}}{(\mu-a)^2}
\le
\frac{2}{\mu-a}+\frac{1}{(\mu-a)^2}.
\end{equation}
\end{lemma}

\begin{proof}
Let $z\in\Omega_a$, so $\abs{\RePart z}\le a$.
Then
\[
H(A-z\I)=H(A)-\RePart(z)\I\succeq (\mu-a)\I,
\qquad
H(B+z\I)=H(B)+\RePart(z)\I\succeq (\mu-a)\I.
\]
Hence
\[
\norm{(z\I-A)^{-1}}\le \frac{1}{\mu-a},
\qquad
\norm{(z\I+B)^{-1}}\le \frac{1}{\mu-a}.
\]
Using the block factorization
\[
(z\I-M)^{-1}=
\begin{bmatrix}
(z\I-A)^{-1} & (z\I-A)^{-1}C(z\I+B)^{-1} \\
0 & (z\I+B)^{-1}
\end{bmatrix}
\]
and the triangle inequality gives
\[
\norm{(z\I-M)^{-1}}
\le
\norm{(z\I-A)^{-1}} + \norm{(z\I+B)^{-1}} + \norm{(z\I-A)^{-1}}\,\norm{C}\,\norm{(z\I+B)^{-1}},
\]
which yields \eqref{eq:R1-gammaM}.
The last inequality uses $\norm{C}\le \norm{M}\le 1$.
\end{proof}

\subsection{Quantum algorithms for the ordinary Sylvester equation}

\begin{theorem}[Ordinary Sylvester in the FoV gap regime]\label{thm:R1-main}
Assume \cref{ass:R1}.
Then for every $\eps\in(0,1]$ there exist an ancilla count $a_X$ and a unitary $U_X$ such that
\begin{equation}\label{eq:R1-main-output}
\left\|X-\beta_X(\bra{0^{a_X}}\otimes \I)U_X(\ket{0^{a_X}}\otimes \I)\right\|<\eps.
\end{equation}
Moreover,
\begin{equation}\label{eq:R1-betaX}
\beta_X \le 36\Lambda^{\mathrm{syl}}_\star\frac{1}{\mu^2}<21\frac{1}{\mu^2},
\end{equation}
and the query complexities satisfy
\begin{equation}\label{eq:R1-query-main}
\qA = O\!\paren{\frac{1}{\mu}\log\frac{1}{\eps\mu}},
\qquad
\qB = O\!\paren{\frac{1}{\mu}\log\frac{1}{\eps\mu}},
\qquad
\qC = O(1).
\end{equation}
\end{theorem}

\begin{proof}
Choose
\[
a:=\frac{\mu}{2},
\qquad
\beta:=\frac12\arcsin a.
\]
Then $\beta=\Theta(\mu)$ for $\mu\in(0,1]$, and \cref{lem:R1-strip-bound} gives
\[
\gamma_{\Omega_a}(M)\le \frac{4}{\mu}+\frac{4}{\mu^2}\le \frac{8}{\mu^2}=: \gamma_{\mathrm{R1}}.
\]
Let $K$ satisfy \eqref{eq:Kchoice-sinc} with $\eps_{\mathrm{sgn}}=\eps/2$, $\gamma=\gamma_{\mathrm{R1}}$, and the present choices of $a$ and $\beta$, and set $h=\sqrt{2\pi\beta/K}$.
Then
\[
K=O\!\paren{\frac{1}{\mu}\log^2\frac{1}{\eps\mu}}.
\]
By \cref{lem:R1-conditioning}, the available uniform bounds
\[
r_A=r_B=\frac{3}{\mu}
\]
satisfy \eqref{eq:generic-uniform-profile}. Apply \cref{thm:generic-main} with inverse target precisions
\[
\eps_A = c\,\eps\mu,
\qquad
\eps_B = c\,\eps\mu,
\]
where $c>0$ is a sufficiently small absolute constant.
The deterministic term in \cref{thm:generic-main} is at most $\eps/2$ by construction, and the
implementation term is at most $\eps/2$ because $\Lambda^{\mathrm{syl}}_{K,h}\le \Lambda^{\mathrm{syl}}_\star$
and $r_A,r_B=O(\mu^{-1})$.
The normalization bound follows from
\[
4 r_A r_B\Lambda^{\mathrm{syl}}_{K,h}
\le
4\paren{\frac{3}{\mu}}^2\Lambda^{\mathrm{syl}}_\star.
\]
Finally,
\[
\log\frac{2r_A}{\eps_A}
=
O\!\paren{\log\frac{1}{\eps\mu}},
\qquad
\log\frac{2r_B}{\eps_B}
=
O\!\paren{\log\frac{1}{\eps\mu}},
\]
so \cref{eq:generic-query-main} yields \eqref{eq:R1-query-main}.
\end{proof}

\begin{theorem}[Ordinary Sylvester under strip-resolvent hypotheses]\label{thm:R2-main}
Assume \cref{ass:R2} and the half-plane separation hypotheses of \cref{thm:higham}.
Then for every $\eps\in(0,1]$ there exist an ancilla count $a_X$ and a unitary $U_X$ such that
\begin{equation}\label{eq:R2-main-output}
\left\|X-\beta_X(\bra{0^{a_X}}\otimes \I)U_X(\ket{0^{a_X}}\otimes \I)\right\|<\eps.
\end{equation}
Moreover,
\begin{equation}\label{eq:R2-betaX}
\beta_X \le 36\Lambda^{\mathrm{syl}}_\star \gamma_A\gamma_B\le 36\Lambda^{\mathrm{syl}}_\star \gamma^2
< 21 \gamma^2,
\end{equation}
and the query complexities satisfy
\begin{equation}\label{eq:R2-query-main}
\qA = O\!\paren{\gamma\log\frac{\gamma}{\eps}},
\end{equation}
\begin{equation}\label{eq:R2-query-main-B}
\qB = O\!\paren{\gamma\log\frac{\gamma}{\eps}},
\qquad
\qC = O(1).
\end{equation}
\end{theorem}

\begin{proof}
Choose
\[
\beta:=\frac12\arcsin a.
\]
Let $K$ satisfy \eqref{eq:Kchoice-sinc} with $\eps_{\mathrm{sgn}}=\eps/2$, $\gamma=\gamma$, and the present choices of $a$ and $\beta$, and set $h=\sqrt{2\pi\beta/K}$.
Because $\beta$ depends only on $a$, one has
\[
K=O\!\paren{\log^2\frac{\gamma}{\eps}},
\]
and the deterministic term in \cref{thm:generic-main} is at most $\eps/2$.
By \cref{lem:R2-conditioning}, the available uniform bounds
\[
r_A=r_B=3\gamma
\]
satisfy \eqref{eq:generic-uniform-profile}. Apply \cref{thm:generic-main} with inverse target precisions
\[
\eps_A = c\,\frac{\eps}{\gamma},
\qquad
\eps_B = c\,\frac{\eps}{\gamma},
\]
where $c>0$ is a sufficiently small absolute constant.
Then the implementation term is at most $\eps/2$, because $\Lambda^{\mathrm{syl}}_{K,h}\le \Lambda^{\mathrm{syl}}_\star$
and $r_A,r_B=O(\gamma)$.
Substituting $r_A,r_B\le 3\gamma$ into \cref{eq:generic-main-output} yields
\eqref{eq:R2-betaX}.
Moreover,
\[
\log\frac{2r_A}{\eps_A}
=
O\!\paren{\log\frac{\gamma}{\eps}},
\qquad
\log\frac{2r_B}{\eps_B}
=
O\!\paren{\log\frac{\gamma}{\eps}},
\]
so \cref{eq:generic-query-main} yields \eqref{eq:R2-query-main} and \eqref{eq:R2-query-main-B}.
\end{proof}

\begin{remark}[A direct augmented-resolvent implementation]
One may also implement the Sylvester sign block without using the off-diagonal
factorization
\[
[(zI-M)^{-1}]_{12}
=
(zI-A)^{-1}C(zI+B)^{-1}.
\]
Namely, since \(X=\frac12[\mathrm{sign}(M)]_{12}\), one can apply the single-family
implementation directly to the shifted augmented resolvents
\[
(1+t_k)(M\pm it_k I)^{-1}.
\]
With a nodewise profile
\[
(1+t_k)\|(M\pm it_kI)^{-1}\|\le \rho_{M,k}^{\pm},
\]
define
\[
\Theta_{\rho_M}^{\mathrm{dir}}
:=
\sum_{k=-K}^{K}
\frac{h t_k}{2\pi(1+t_k)}
\bigl(\rho_{M,k}^{-}+\rho_{M,k}^{+}\bigr).
\]
Then the direct implementation gives a block-encoding of \(X\) with normalization
\[
\beta_X^{\mathrm{dir}}=2\Theta_{\rho_M}^{\mathrm{dir}}.
\]
Under the uniform profile
\[
\rho_{M,k}^{\pm}\equiv r_M,
\qquad
r_M=\max_{k,\pm}(1+t_k)\|(M\pm it_kI)^{-1}\|,
\]
this becomes
\[
\beta_X^{\mathrm{dir}}
=
r_M\Lambda_{K,h}^{\mathrm{care}}
=
O\!\left(
r_M\log\frac{\widetilde C_{\beta,a,\gamma}}{\epsilon_{\mathrm{sign}}}
\right).
\]
In particular, under a full augmented strip-resolvent bound
\(\gamma\ge \gamma_{\Omega_a}(M)\), one has \(r_M\le 3\gamma\), and hence
\[
\beta_X^{\mathrm{dir}}
=
O\!\left(
\gamma
\log\frac{\widetilde C_{\beta,a,\gamma}}{\epsilon_{\mathrm{sign}}}
\right).
\]

The direct augmented-resolvent implementation can improve the coarse bound \(O(\gamma^2)\) obtained by
bounding both diagonal resolvent families by the same full augmented constant.
Its coefficient sum grows logarithmically. The factorized implementation has normalization
\(O(\gamma_A\gamma_B)\), or more generally \(4\Theta_\rho^{\mathrm{syl}}\), with a bounded
coefficient sum. The direct implementation has a smaller normalization when
\[
\gamma\log(\widetilde C_{\beta,a,\gamma}/\epsilon_{\mathrm{sign}})
\ll
\gamma_A\gamma_B,
\]
and it also applies when the available oracle encodes only \(M\).
\end{remark}

\begin{theorem}[Optimal shift-rotation from a FoV gap]\label{thm:fov-shift}
Let $A_0\in\C^{n\times n}$ and $B_0\in\C^{m\times m}$, and define
\begin{equation}\label{eq:delta-fov}
\delta := \dist\paren{W(A_0), -W(B_0)} > 0.
\end{equation}
Then there exist $\eta\in\C$ with $\abs{\eta}=1$ and $\omega\in\C$ such that, with
\[
A' := \eta(A_0-\omega\I),
\qquad
B' := \eta(B_0+\omega\I),
\]
one has
\begin{equation}\label{eq:fov-optimal-margin}
H(A') \succeq \frac{\delta}{2}\I,
\qquad
H(B') \succeq \frac{\delta}{2}\I.
\end{equation}
Moreover, $\delta/2$ is the largest common half-plane margin achievable by any shift-rotation.
\end{theorem}

\begin{proof}
Because $W(A_0)$ and $-W(B_0)$ are compact, convex, and disjoint, there exist closest points
$a_\star\in W(A_0)$ and $b_\star\in -W(B_0)$ with $\abs{a_\star-b_\star}=\delta$.
Let $v:=a_\star-b_\star$ and choose $\eta=e^{-i\arg(v)}$, so that $\eta v = \delta$ is positive real.
Set $\omega=(a_\star+b_\star)/2$.
By convexity, the common perpendicular bisector through $\omega$ supports both sets at the
closest points, hence
\[
\RePart\,\eta(a-\omega) \ge \frac{\delta}{2}\quad \forall a\in W(A_0),
\qquad
\RePart\,\eta(b-\omega) \le -\frac{\delta}{2}\quad \forall b\in -W(B_0).
\]
Replacing $b$ by $-z$ with $z\in W(B_0)$ gives
\[
\RePart\,\eta(z+\omega) \ge \frac{\delta}{2}\quad \forall z\in W(B_0).
\]
The characterization $\lambda_{\min}(H(T)) = \min\set{\RePart z \suchthat z\in W(T)}$ now yields
\eqref{eq:fov-optimal-margin}.
Optimality is immediate because shift-rotation acts by a rigid motion on the pair
$\bigl(W(A_0),-W(B_0)\bigr)$, so no such transformation can increase their Euclidean distance.
\end{proof}

\begin{corollary}[Ordinary Sylvester under a FoV gap]\label{cor:fov-main}
If the original pair $(A_0,B_0)$ satisfies \eqref{eq:delta-fov}, then after applying the optimal
shift-rotation from \cref{thm:fov-shift} and a common positive scaling $\lambda$ that enforces
$\norm{M}\le 1$, the conclusion of \cref{thm:R1-main} holds with
\[
\mu = \frac{\lambda\delta}{2}
\]
in the transformed coordinates.
\end{corollary}

\begin{proof}
Apply \cref{thm:fov-shift} and then invoke \cref{thm:R1-main}.
\end{proof}

\begin{theorem}[Banded-overlap refinement in the FoV gap regime]\label{thm:R1-banded}
Assume \cref{ass:R1}, and suppose in addition that
\begin{equation}\label{eq:R1-banded-ass}
W(A)\subset \set{z\in\C \suchthat \RePart z\ge \mu,\ \abs{\ImPart z}\le \tau},
\qquad
W(B)\subset \set{z\in\C \suchthat \RePart z\ge \mu,\ \abs{\ImPart z}\le \tau}
\end{equation}
for some $\tau\in[0,1]$.
Then for every $\eps\in(0,1]$ there exist an ancilla count $a_X$ and a unitary $U_X$ such that
\begin{equation}\label{eq:R1-band-output}
\left\|X-\beta_X(\bra{0^{a_X}}\otimes \I)U_X(\ket{0^{a_X}}\otimes \I)\right\|<\eps,
\end{equation}
with
\begin{equation}\label{eq:R1-band-beta}
\beta_X
\le
4\left(
\frac{1}{\mu}
+
\left(1+\frac{1}{\pi}\right)\frac{\tau}{\mu^2}
\right),
\end{equation}
and query complexities
\begin{equation}\label{eq:R1-band-queries}
\qA
=
O\!\paren{\frac{1}{\mu}\log\frac{\tau}{\eps\mu}},
\qquad
\qB
=
O\!\paren{\frac{1}{\mu}\log\frac{\tau}{\eps\mu}},
\qquad
\qC=O(1).
\end{equation}
\end{theorem}

\begin{proof}
Choose
\[
a:=\frac{\mu}{2},
\qquad
\beta:=\frac12\arcsin a.
\]
By \cref{lem:R1-strip-bound},
\[
\gamma_{\Omega_a}(M)\le \frac{4}{\mu}+\frac{4}{\mu^2}\le \frac{8}{\mu^2}=: \gamma_{\mathrm{R1}}.
\]
Let $K$ satisfy \eqref{eq:Kchoice-sinc} with $\eps_{\mathrm{sgn}}=\eps/2$, $\gamma=\gamma_{\mathrm{R1}}$, and the present choices of $a$ and $\beta$, and set $h=\sqrt{2\pi\beta/K}$.
Then
\[
K=O\!\paren{\frac{1}{\mu}\log^2\frac{1}{\eps\mu}}.
\]

Fix $t\ge 0$.
For $z=it$ and every $w\in W(A)$, assumption \eqref{eq:R1-banded-ass} gives
\[
\abs{z-w}^2 \ge \mu^2 + (t-\tau)_+^2.
\]
Hence
\[
\dist(it,W(A))\ge \sqrt{\mu^2+(t-\tau)_+^2},
\]
and the standard numerical-range resolvent bound yields
\[
\norm{(A-it\I)^{-1}}
\le
\frac{1}{\sqrt{\mu^2+(t-\tau)_+^2}}.
\]
The same bound holds for $(A+it\I)^{-1}$, $(B-it\I)^{-1}$, and $(B+it\I)^{-1}$ by applying the same argument to the points $\pm it$ and the sets $W(A)$ and $W(B)$.
Define the available nodewise profile
\[
\rho_{A,k}^{\pm}=\rho_{B,k}^{\pm}
:=
\frac{1+t_k}{\sqrt{\mu^2+(t_k-\tau)_+^2}}.
\]
This profile is admissible in the sense of \cref{def:sylvester-profile}.  Its overlap parameter satisfies
\begin{equation}\label{eq:R1-band-theta}
\Theta_{\rho}^{\mathrm{syl}}
\le
\frac{h}{\pi}\sum_{k=-K}^{K}\frac{t_k}{\mu^2+(t_k-\tau)_+^2}.
\end{equation}

Define
\[
q_{\mu,\tau}(x):=\frac{e^x}{\mu^2+(e^x-\tau)_+^2}.
\]
The function $q_{\mu,\tau}$ is nonnegative and unimodal on $\R$, so
\[
h\sum_{k\in\mathbb Z}q_{\mu,\tau}(kh)
\le
h\sup_{x\in\mathbb R}q_{\mu,\tau}(x)+\int_{-\infty}^{\infty}q_{\mu,\tau}(x)\,dx.
\]
Now
\[
\int_{-\infty}^{\infty}q_{\mu,\tau}(x)\,dx
=
\int_0^\infty \frac{dt}{\mu^2+(t-\tau)_+^2}
=
\frac{\tau}{\mu^2}+\frac{\pi}{2\mu}.
\]
Moreover, the maximum of $t\mapsto t/(\mu^2+(t-\tau)_+^2)$ on $[0,\infty)$ is attained at
$t_\star=\sqrt{\mu^2+\tau^2}$, and
\[
\sup_{x\in\R}q_{\mu,\tau}(x)
=
\frac{t_\star}{\mu^2+(t_\star-\tau)^2}
=
\frac{t_\star+\tau}{2\mu^2}
\le
\frac{1}{2\mu}+\frac{\tau}{\mu^2},
\]
because $t_\star\le \mu+\tau$.
Since the balanced choice satisfies $h<\pi$, we conclude from \eqref{eq:R1-band-theta} that
\begin{align*}
\Theta_{\rho}^{\mathrm{syl}}
&\le
\frac{1}{\pi}\left(
h\sup q_{\mu,\tau}
+
\int_{-\infty}^{\infty}q_{\mu,\tau}(x)\,dx
\right)\\
&\le
\left(\frac{1}{2\mu}+\frac{\tau}{\mu^2}\right)
+
\left(\frac{1}{2\mu}+\frac{\tau}{\pi\mu^2}\right)\\
&=
\frac{1}{\mu}
+
\left(1+\frac{1}{\pi}\right)\frac{\tau}{\mu^2}.
\end{align*}

By the same split as in \cref{lem:R1-conditioning}, this profile satisfies $R_A,R_B\le 3/\mu$.
Apply \cref{thm:generic-overlap} with the profile $\rho$ and choose
\[
\eps_A
=
c\,\frac{\eps}{1+\tau/\mu},
\qquad
\eps_B
=
c\,\frac{\eps}{1+\tau/\mu},
\]
where $c>0$ is a sufficiently small absolute constant.
The deterministic term is at most $\eps/2$ by construction, and the implementation term satisfies
\[
\Theta_{\rho}^{\mathrm{syl}}\norm{C}\left(
\frac{\eps_A}{R_A}
+
\frac{\eps_B}{R_B}
+
\frac{\eps_A\eps_B}{R_A R_B}
\right)
=
O\!\Bigl(
\paren{1+\frac{\tau}{\mu}}(\eps_A+\eps_B)+\eps_A\eps_B
\Bigr)
\le \frac{\eps}{2}.
\]
The block-encoding normalization is bounded by
\[
4\Theta_{\rho}^{\mathrm{syl}}
\le
4\left(
\frac{1}{\mu}
+
\left(1+\frac{1}{\pi}\right)\frac{\tau}{\mu^2}
\right),
\]
which is \eqref{eq:R1-band-beta}.
Finally, substituting $R_A,R_B=O(1/\mu)$ and the above choices of $\eps_A,\eps_B$ into \cref{eq:generic-overlap-queries} yields \eqref{eq:R1-band-queries}.
\end{proof}

\begin{corollary}[Diagonalizable strip bound]\label{cor:diag-strip}
Assume the hypotheses of \cref{thm:R2-main}, and suppose in addition that $M=SDS^{-1}$ is
diagonalizable and
\[
d := \dist\paren{\spec(M),\Omega_a} > 0.
\]
Then
\begin{equation}\label{eq:diag-strip}
\gamma \le \frac{\kappa(S)}{d}.
\end{equation}
Consequently, for every $\eps\in(0,1]$ there exist an ancilla count $a_X$ and a unitary $U_X$ such that
\[
\left\|X-\beta_X(\bra{0^{a_X}}\otimes \I)U_X(\ket{0^{a_X}}\otimes \I)\right\|<\eps,
\]
with
\[
\beta_X \le 36\Lambda^{\mathrm{syl}}_\star \frac{\kappa(S)^2}{d^2},
\]
and query complexities
\begin{align*}
\qA &= O\!\paren{\frac{\kappa(S)}{d}\log\frac{\kappa(S)}{\eps d}},\\
\qB &= O\!\paren{\frac{\kappa(S)}{d}\log\frac{\kappa(S)}{\eps d}},\\
\qC &= O(1).
\end{align*}
\end{corollary}

\begin{proof}
For $z\in\Omega_a$,
\[
(z\I-M)^{-1} = S(z\I-D)^{-1}S^{-1},
\]
so
\[
\norm{(z\I-M)^{-1}} \le \kappa(S)\,\norm{(z\I-D)^{-1}} \le \frac{\kappa(S)}{d}.
\]
This proves \eqref{eq:diag-strip}.
The stated normalization and query bounds follow by substituting
$\gamma\le \kappa(S)/d$ into \cref{thm:R2-main}.
\end{proof}

\subsection{General nodewise sign-embedding theorems}

The following nodewise and uniform-profile results imply \cref{thm:R1-main,thm:R2-main,thm:R1-banded} under their respective assumptions.

\begin{theorem}[Rebalanced sign-embedding theorem for the ordinary Sylvester equation]\label{thm:generic-overlap}
Assume the transformed matrices satisfy \eqref{eq:Mnorm} and the half-plane separation hypotheses
of \cref{thm:higham}.
Suppose there exist parameters
\[
a\in(0,1),\qquad \beta\in(0,\arcsin a),\qquad \gamma\ge \gamma_{\Omega_a}(M),
\]
and choose any $K\ge 1$, any $h>0$, inverse precisions $\eps_A,\eps_B\in(0,1]$, and a classically available Sylvester nodewise inverse profile $\rho$ in the sense of \cref{def:sylvester-profile}.
Let $R_A$, $R_B$, and $\Theta_{\rho}^{\mathrm{syl}}$ be defined by \eqref{eq:syl-profile-R}--\eqref{eq:Theta-syl-profile}.
Then there exists a unitary $U_X^{\mathrm{reb}}$ that is a
\begin{equation}\label{eq:generic-overlap-output}
\paren{
4\Theta_{\rho}^{\mathrm{syl}},
\ a_A+a_B+a_C+O(\log K),
\ \eps_{\mathrm{det}}
+
\Theta_{\rho}^{\mathrm{syl}}\norm{C}\left(
\frac{\eps_A}{R_A}
+
\frac{\eps_B}{R_B}
+
\frac{\eps_A\eps_B}{R_A R_B}
\right)
}
\text{-block-encoding of }X,
\end{equation}
where
\begin{equation}\label{eq:eps-det-overlap}
\eps_{\mathrm{det}} := \half\eps_{\mathrm{sgn}}^{\mathrm{sinc}}(K,h;\beta,a,\gamma).
\end{equation}
Its query complexity is
\begin{equation}\label{eq:generic-overlap-queries}
\qA = O\!\paren{R_A \log\frac{R_A}{\eps_A}},
\qquad
\qB = O\!\paren{R_B \log\frac{R_B}{\eps_B}},
\qquad
\qC = O(1).
\end{equation}
In particular, if
\begin{equation}\label{eq:sufficient-choice-generic-overlap}
\eps_{\mathrm{sgn}}^{\mathrm{sinc}}(K,h;\beta,a,\gamma) \le \eps
\qquad\text{and}\qquad
\Theta_{\rho}^{\mathrm{syl}}\norm{C}\left(
\frac{\eps_A}{R_A}
+
\frac{\eps_B}{R_B}
+
\frac{\eps_A\eps_B}{R_A R_B}
\right)
\le \frac{\eps}{2},
\end{equation}
then $U_X^{\mathrm{reb}}$ is a $(4\Theta_{\rho}^{\mathrm{syl}},\,a_A+a_B+a_C+O(\log K),\,\eps)$-block-encoding of $X$.
\end{theorem}

\begin{proof}
Combine \cref{cor:XN-error,thm:profile-implementation}.
The deterministic quadrature error contributes at most $\eps_{\mathrm{det}}$, and the implementation error is exactly the one stated in \cref{eq:profile-be}.
The query bounds are part of \cref{thm:profile-implementation}.
\end{proof}

\begin{theorem}[Uniform-profile sign-embedding specialization]\label{thm:generic-main}
Assume the transformed matrices satisfy \eqref{eq:Mnorm} and the half-plane separation hypotheses
of \cref{thm:higham}.
Suppose there exist parameters
\[
a\in(0,1),\qquad \beta\in(0,\arcsin a),\qquad \gamma\ge \gamma_{\Omega_a}(M),
\]
and choose any $K\ge 1$, any $h>0$, and inverse precisions $\eps_A,\eps_B\in(0,1]$.
Suppose $r_A,r_B$ are available uniform upper bounds satisfying
\begin{equation}\label{eq:generic-uniform-profile}
(1+t_k)\norm{(A\pm it_k\I)^{-1}}\le r_A,
\qquad
(1+t_k)\norm{(B\pm it_k\I)^{-1}}\le r_B
\end{equation}
for all active nodes and both signs.
Then there exists a unitary $U_X$ that is a
\begin{equation}\label{eq:generic-main-output}
\paren{4 r_A r_B\Lambda^{\mathrm{syl}}_{K,h},\ a_A+a_B+a_C+O(\log K),\ \eps_{\mathrm{det}} + \Lambda^{\mathrm{syl}}_{K,h}\Delta_{\mathrm{prod}}}\text{-block-encoding of }X,
\end{equation}
where
\[
\eps_{\mathrm{det}} := \half\eps_{\mathrm{sgn}}^{\mathrm{sinc}}(K,h;\beta,a,\gamma)
\]
and
\begin{equation}\label{eq:delta-prod}
\Delta_{\mathrm{prod}} := \norm{C}\paren{\eps_A r_B + r_A \eps_B + \eps_A \eps_B}.
\end{equation}
Its query complexity is
\begin{equation}\label{eq:generic-query-main}
\qA = O\!\paren{r_A \log\frac{r_A}{\eps_A}},
\qquad
\qB = O\!\paren{r_B \log\frac{r_B}{\eps_B}},
\qquad
\qC = O(1).
\end{equation}
In particular, if
\begin{equation}\label{eq:sufficient-choice-generic}
\eps_{\mathrm{sgn}}^{\mathrm{sinc}}(K,h;\beta,a,\gamma) \le \eps
\qquad\text{and}\qquad
\Lambda^{\mathrm{syl}}_{K,h}\Delta_{\mathrm{prod}} \le \frac{\eps}{2},
\end{equation}
then $U_X$ is a $(4 r_A r_B\Lambda^{\mathrm{syl}}_{K,h},\,a_A+a_B+a_C+O(\log K),\,\eps)$-block-encoding of $X$.
\end{theorem}

\begin{proof}
Use the constant Sylvester profile $\rho_{A,k}^{\pm}\equiv r_A$ and $\rho_{B,k}^{\pm}\equiv r_B$ in \cref{thm:generic-overlap}.  Then
\[
\Theta_{\rho}^{\mathrm{syl}}\le r_A r_B\Lambda^{\mathrm{syl}}_{K,h},
\]
and
\[
\Theta_{\rho}^{\mathrm{syl}}\norm{C}\left(
\frac{\eps_A}{r_A}
+
\frac{\eps_B}{r_B}
+
\frac{\eps_A\eps_B}{r_A r_B}
\right)
\le
\Lambda^{\mathrm{syl}}_{K,h}\Delta_{\mathrm{prod}}.
\]
This gives the stated normalization, error, and query bounds.
\end{proof}

\section{Generalized Sylvester and generalized Lyapunov equations}\label{sec:generalized}
The generalized equation reduces to an ordinary Sylvester equation. Applying the sign representation, log-sinc approximation, and rebalanced shifted-inverse algorithm to the reduced coefficients gives the result below; the generalized Lyapunov equation is then obtained as a specialization.

\subsection{Generalized Sylvester equations}
Consider the nonsingular generalized Sylvester equation
\begin{equation}\label{eq:gen-sylvester}
AXD + EXB = C,
\end{equation}
where $E$ and $D$ are invertible. Set
\begin{equation}\label{eq:reduced-coeffs}
\widetilde A := E^{-1}A,
\qquad
\widetilde B := BD^{-1},
\qquad
\widetilde C := E^{-1}CD^{-1}.
\end{equation}
Then \eqref{eq:gen-sylvester} is equivalent to the ordinary Sylvester equation
\begin{equation}\label{eq:reduced-gen-sylvester}
\widetilde A X + X\widetilde B = \widetilde C,
\end{equation}
with reduced augmented matrix
\begin{equation}\label{eq:gen-M}
\widetilde M := \begin{bmatrix} \widetilde A & \widetilde C \\ 0 & -\widetilde B \end{bmatrix}.
\end{equation}

As in the ordinary case, the solution is unchanged by generalized shift-rotation and pair scaling:
for any nonzero $\eta\in\C$, any $\omega\in\C$, and any $\lambda_E,\lambda_D>0$,
\[
A \mapsto \eta(A-\omega E),\qquad
B \mapsto \eta(B+\omega D),\qquad
C \mapsto \eta C,
\]
\[
A \mapsto \lambda_E A,\quad E \mapsto \lambda_E E,\qquad
B \mapsto \lambda_D B,\quad D \mapsto \lambda_D D,\qquad
C \mapsto \lambda_E\lambda_D C,
\]
leave the solution $X$ unchanged. We apply these transformations and retain the notation $A,B,C,E,D$ for the transformed data.

\begin{assumption}[Unit-normalized block-encoding access for generalized data]\label{ass:gen-be}
We are given exact unit block-encodings
\[
\begin{aligned}
U_A&\in \BE(1,a_A,0), & U_B&\in \BE(1,a_B,0), & U_C&\in \BE(1,a_C,0),\\
U_E&\in \BE(1,a_E,0), & U_D&\in \BE(1,a_D,0).
\end{aligned}
\]
Controlled versions and adjoints are available at the same asymptotic query cost.
In the generalized Lyapunov specialization, an exact unit block-encoding $U_Q\in \BE(1,a_Q,0)$ of $Q$ is given.
\end{assumption}

Define the shifted pencil norms
\begin{equation}\label{eq:gen-RAB}
R_{A,E}:=\max_{k,\pm}(1+t_k)\norm{(A\pm it_k E)^{-1}},
\qquad
R_{B,D}:=\max_{k,\pm}(1+t_k)\norm{(B\pm it_k D)^{-1}},
\end{equation}
and define the generalized overlap sum
\begin{equation}\label{eq:gen-Theta}
\Theta_{A,E;B,D}^{\mathrm{gen}}
:=
\sum_{k=-K}^{K}\frac{h t_k}{2\pi}
\Bigl(
\norm{(A-it_kE)^{-1}}\norm{(B+it_kD)^{-1}}
+
\norm{(A+it_kE)^{-1}}\norm{(B-it_kD)^{-1}}
\Bigr).
\end{equation}
The plain bound
\begin{equation}\label{eq:gen-Theta-plain}
\Theta_{A,E;B,D}^{\mathrm{gen}}
\le
R_{A,E}R_{B,D}\Lambda^{\mathrm{syl}}_{K,h}
\end{equation}
is immediate.

\begin{theorem}[Generalized Sylvester equation via reduced sign embedding]\label{thm:gen-main}
Assume \cref{ass:gen-be}. Consider the transformed generalized Sylvester equation \eqref{eq:gen-sylvester} and its reduction \eqref{eq:reduced-gen-sylvester}. Assume
\[
\spec(\widetilde A)\subset\set{z\in\C:\RePart z>0},
\qquad
\spec(-\widetilde B)\subset\set{z\in\C:\RePart z<0},
\qquad
\norm{\widetilde M}\le 1,
\]
and suppose that for some $a\in(0,1)$ and $\beta\in(0,\arcsin a)$,
\begin{equation}\label{eq:gen-gamma}
\gamma_{\mathrm{red}} \ge \gamma_{\Omega_a}(\widetilde M).
\end{equation}
Choose any $K\ge 1$, any $h>0$, and inverse precisions $\eps_A,\eps_B\in(0,1]$.
Define
\begin{equation}\label{eq:gen-eps-det}
\eps_{\mathrm{det}}^{(\mathrm{gen})}
:=
\half \eps_{\mathrm{sgn}}^{\mathrm{sinc}}(K,h;\beta,a,\gamma_{\mathrm{red}}).
\end{equation}
Then there exists a unitary $U_X^{(\mathrm{gen})}$ that is a

\begin{equation}\label{eq:gen-main-output}
\begin{aligned}
\Bigl(
4\Theta_{A,E;B,D}^{\mathrm{gen}},
\ a_A+a_E+a_B+a_D+a_C+O(\log K),
\ \eps_{\mathrm{det}}^{(\mathrm{gen})}
+
\Theta_{A,E;B,D}^{\mathrm{gen}}\norm{C}\left(
\frac{\eps_A}{R_{A,E}}
+
\frac{\eps_B}{R_{B,D}}
+
\frac{\eps_A\eps_B}{R_{A,E}R_{B,D}}
\right)
\Bigr) \\
\text{-block-encoding of }X.
\end{aligned}
\end{equation}
Its query complexity is
\begin{equation}\label{eq:gen-query-rho}
O\!\paren{R_{A,E}\log\frac{R_{A,E}}{\eps_A}}
\text{ queries each to }U_A,U_E,
\end{equation}
\begin{equation}\label{eq:gen-query-rho-B}
O\!\paren{R_{B,D}\log\frac{R_{B,D}}{\eps_B}}
\text{ queries each to }U_B,U_D,
\end{equation}
plus $O(1)$ queries to $U_C$.

If, in addition,
\begin{equation}\label{eq:gen-weighted-accretive}
\begin{aligned}
E&\succ 0,\qquad D\succ 0,\\
\norm{E^{-1/2}AE^{-1/2}} &\le 1,
&\qquad \norm{D^{-1/2}BD^{-1/2}} &\le 1,\\
H(E^{-1/2}AE^{-1/2}) &\succeq \mu\I,
&\qquad H(D^{-1/2}BD^{-1/2}) &\succeq \mu\I,
\end{aligned}
\end{equation}
then under the balanced choice $h=\sqrt{2\pi\beta/K}$ with \eqref{eq:Kchoice-sinc} enforced using $\gamma=\gamma_{\mathrm{red}}$, one obtains
\begin{equation}\label{eq:gen-R1-specialization}
\beta_X \le 36\Lambda^{\mathrm{syl}}_\star\frac{\norm{E^{-1}}\norm{D^{-1}}}{\mu^2}
<21\frac{\norm{E^{-1}}\norm{D^{-1}}}{\mu^2}
\end{equation}
and
\begin{equation}\label{eq:gen-R1-queries}
O\!\paren{\frac{\norm{E^{-1}}}{\mu}\log\frac{\norm{E^{-1}}\norm{D^{-1}}}{\eps\mu}}
\text{ queries each to }U_A,U_E,
\end{equation}
\begin{equation}\label{eq:gen-R1-queries-B}
O\!\paren{\frac{\norm{D^{-1}}}{\mu}\log\frac{\norm{E^{-1}}\norm{D^{-1}}}{\eps\mu}}
\text{ queries each to }U_B,U_D.
\end{equation}
\end{theorem}

\begin{remark}
Replacing the exact quantities in \eqref{eq:gen-RAB} and \eqref{eq:gen-Theta} by explicit upper bounds enlarges the final normalization and query bounds without changing the proof.
\end{remark}

When $E\succ0$ and $D\succ0$, the relevant separation condition is a weighted FoV gap:
\[
W_E(A) := W(E^{-1/2}AE^{-1/2}),
\qquad
W_D(B) := W(D^{-1/2}BD^{-1/2}).
\]
A positive gap $\dist(W_E(A),-W_D(B))>0$ can be converted, by the ordinary field-of-values shift-rotation argument applied to $E^{-1/2}AE^{-1/2}$ and $D^{-1/2}BD^{-1/2}$, into the normalized weighted accretivity hypothesis \eqref{eq:gen-weighted-accretive} after a subsequent pair scaling.

\begin{lemma}[Weighted field-of-values gap conditioning]\label{lem:gen-conditioning-R1}
Assume \eqref{eq:gen-weighted-accretive}. Then
\begin{equation}\label{eq:gen-r-bounds-R1}
(1+t)\norm{(A\pm itE)^{-1}} \le \frac{3\norm{E^{-1}}}{\mu},
\qquad
(1+t)\norm{(B\pm itD)^{-1}} \le \frac{3\norm{D^{-1}}}{\mu}
\qquad (t\ge 0).
\end{equation}
\end{lemma}

\begin{proof}
Write $A = E^{1/2}\widehat A E^{1/2}$ with $\widehat A := E^{-1/2}AE^{-1/2}$.
Then
\[
A \pm itE = E^{1/2}(\widehat A \pm it\I)E^{1/2},
\qquad
(A \pm itE)^{-1} = E^{-1/2}(\widehat A \pm it\I)^{-1}E^{-1/2},
\]
so
\[
(1+t)\norm{(A\pm itE)^{-1}}
\le
\norm{E^{-1}}\,(1+t)\norm{(\widehat A \pm it\I)^{-1}}.
\]
Applying \cref{lem:accretive-res} and the case split from \cref{lem:R1-conditioning} gives the first bound. The proof for $(B,D)$ is identical.
\end{proof}

\begin{proof}[Proof of \cref{thm:gen-main}]
Applying \cref{cor:XN-error} to the reduced Sylvester equation \eqref{eq:reduced-gen-sylvester} gives the deterministic error bound
\[
\norm{X-X_{K,h}^{(\mathrm{gen})}}
\le
\half \eps_{\mathrm{sgn}}^{\mathrm{sinc}}(K,h;\beta,a,\gamma_{\mathrm{red}})
=
\eps_{\mathrm{det}}^{(\mathrm{gen})},
\]
where
\begin{equation}\label{eq:gen-XKh-unscaled}
X_{K,h}^{(\mathrm{gen})}
:=
\frac{h}{2\pi}\sum_{k=-K}^{K} t_k \paren{
(A-it_kE)^{-1} C (B+it_kD)^{-1} + (A+it_kE)^{-1} C (B-it_kD)^{-1}
}.
\end{equation}
The reduced resolvent identity
\begin{equation}\label{eq:gen-reduced-res-identity}
(\widetilde A\mp it\I)^{-1}\widetilde C(\widetilde B\pm it\I)^{-1}
=
(A\mp itE)^{-1}C(B\pm itD)^{-1}
\end{equation}
is immediate from \eqref{eq:reduced-coeffs}.

Rewrite \eqref{eq:gen-XKh-unscaled} as
\begin{align}
X_{K,h}^{(\mathrm{gen})}
&=
\sum_{k=-K}^{K}\frac{h t_k}{2\pi(1+t_k)^2}
\Big(
((1+t_k)(A-it_kE)^{-1})\,C\,((1+t_k)(B+it_kD)^{-1})
\nonumber\\
&\hspace{8em}+
((1+t_k)(A+it_kE)^{-1})\,C\,((1+t_k)(B-it_kD)^{-1})
\Big).
\label{eq:gen-scaled-identity}
\end{align}

Define the rebalancing factors
\[
d^\pm_{A,E,k}
:=
\frac{R_{A,E}}{(1+t_k)\norm{(A\pm it_kE)^{-1}}},
\qquad
d^\pm_{B,D,k}
:=
\frac{R_{B,D}}{(1+t_k)\norm{(B\pm it_kD)^{-1}}}.
\]
Introduce the rebalanced direct sums
\[
A^{\pm}_{E,\mathrm{reb}}
:=
\sum_{k=-K}^{K}\ket{k}\!\bra{k}\otimes
\frac{A\pm it_k E}{(1+t_k)d^\pm_{A,E,k}},
\qquad
B^{\pm}_{D,\mathrm{reb}}
:=
\sum_{k=-K}^{K}\ket{k}\!\bra{k}\otimes
\frac{B\pm it_k D}{(1+t_k)d^\pm_{B,D,k}}.
\]
For each $k$, the coefficients of $A$ and $E$ in the $k$th block of $A^{\pm}_{E,\mathrm{reb}}$ have absolute values summing to $1/d^\pm_{A,E,k}\le 1$, and similarly for $B^{\pm}_{D,\mathrm{reb}}$.
Hence \cref{lem:mux-two-matrix} gives unit-normalized block-encodings of both direct sums.

Their inverse blocks are
\[
(A^{\pm}_{E,\mathrm{reb}})^{-1}
=
\sum_{k=-K}^{K}\ket{k}\!\bra{k}\otimes
d^\pm_{A,E,k}(1+t_k)(A\pm it_kE)^{-1},
\]
\[
(B^{\pm}_{D,\mathrm{reb}})^{-1}
=
\sum_{k=-K}^{K}\ket{k}\!\bra{k}\otimes
d^\pm_{B,D,k}(1+t_k)(B\pm it_kD)^{-1},
\]
so
\[
\norm{(A^{\pm}_{E,\mathrm{reb}})^{-1}}\le R_{A,E},
\qquad
\norm{(B^{\pm}_{D,\mathrm{reb}})^{-1}}\le R_{B,D}.
\]
Applying \cref{prop:qsvt-inverse} therefore yields inverse block-encodings with query costs \eqref{eq:gen-query-rho}--\eqref{eq:gen-query-rho-B}.

Set
\[
\lambda_k^+ := \frac{h t_k}{2\pi(1+t_k)^2 d^-_{A,E,k}d^+_{B,D,k}},
\qquad
\lambda_k^- := \frac{h t_k}{2\pi(1+t_k)^2 d^+_{A,E,k}d^-_{B,D,k}}.
\]
Then \eqref{eq:gen-scaled-identity} becomes
\begin{align*}
X_{K,h}^{(\mathrm{gen})}
&=
\sum_{k=-K}^{K}\lambda_k^+\,
\Bigl(d^-_{A,E,k}(1+t_k)(A-it_kE)^{-1}\Bigr)
\,C\,
\Bigl(d^+_{B,D,k}(1+t_k)(B+it_kD)^{-1}\Bigr)
\\
&\quad+
\sum_{k=-K}^{K}\lambda_k^-\,
\Bigl(d^+_{A,E,k}(1+t_k)(A+it_kE)^{-1}\Bigr)
\,C\,
\Bigl(d^-_{B,D,k}(1+t_k)(B-it_kD)^{-1}\Bigr).
\end{align*}
Moreover,
\begin{align*}
R_{A,E}R_{B,D}\sum_{k=-K}^{K}\bigl(\abs{\lambda_k^+}+\abs{\lambda_k^-}\bigr)
&=
\sum_{k=-K}^{K}\frac{h t_k}{2\pi}
\Bigl(
\norm{(A-it_kE)^{-1}}\norm{(B+it_kD)^{-1}}
\\ 
&\quad +
\norm{(A+it_kE)^{-1}}\norm{(B-it_kD)^{-1}}
\Bigr)\\
&=
\Theta_{A,E;B,D}^{\mathrm{gen}}
\end{align*}

The product-error estimate in \cref{thm:profile-implementation} therefore yields a
\[
\Bigl(
4\Theta_{A,E;B,D}^{\mathrm{gen}},
\ a_A+a_E+a_B+a_D+a_C+O(\log K),
\ \Theta_{A,E;B,D}^{\mathrm{gen}}\norm{C}\left(
\frac{\eps_A}{R_{A,E}}
+
\frac{\eps_B}{R_{B,D}}
+
\frac{\eps_A\eps_B}{R_{A,E}R_{B,D}}
\right)
\Bigr)
\]
block-encoding of $X_{K,h}^{(\mathrm{gen})}$.
Combining this with the deterministic quadrature error yields \eqref{eq:gen-main-output}.

Under \eqref{eq:gen-weighted-accretive}, \cref{lem:gen-conditioning-R1} gives
\[
R_{A,E}\le \frac{3\norm{E^{-1}}}{\mu},
\qquad
R_{B,D}\le \frac{3\norm{D^{-1}}}{\mu}.
\]
Hence \eqref{eq:gen-Theta-plain} implies
\[
\Theta_{A,E;B,D}^{\mathrm{gen}}
\le
\frac{9\norm{E^{-1}}\norm{D^{-1}}}{\mu^2}\Lambda^{\mathrm{syl}}_{K,h}
\le
\frac{9\norm{E^{-1}}\norm{D^{-1}}}{\mu^2}\Lambda^{\mathrm{syl}}_\star
\]
under the balanced choice.
Choosing
\[
\eps_A=\Theta\!\paren{\frac{\eps\mu}{\norm{D^{-1}}}},
\qquad
\eps_B=\Theta\!\paren{\frac{\eps\mu}{\norm{E^{-1}}}}
\]
forces the implementation term to be at most $\eps/2$, and substituting
$R_{A,E}=\Theta(\norm{E^{-1}}/\mu)$ and $R_{B,D}=\Theta(\norm{D^{-1}}/\mu)$ into
\eqref{eq:gen-query-rho}--\eqref{eq:gen-query-rho-B} proves
\eqref{eq:gen-R1-specialization}--\eqref{eq:gen-R1-queries-B}.
\end{proof}

\begin{remark}[Generalized strip-resolvent regime]
Without weighted field-of-values information, use the reduced strip constants
\[
\widehat\gamma_{A,E} := \sup_{z\in\Omega_a} \norm{(zE-A)^{-1}E}
= \sup_{z\in\Omega_a} \norm{(z\I-\widetilde A)^{-1}},
\]
\[
\widehat\gamma_{B,D} := \sup_{z\in\Omega_a} \norm{D(zD+B)^{-1}}
= \sup_{z\in\Omega_a} \norm{(z\I+\widetilde B)^{-1}}.
\]
Since $\norm{\widetilde A},\norm{\widetilde B}\le 1$ whenever $\norm{\widetilde M}\le 1$, the argument in \cref{lem:R2-conditioning} gives the plain bounds
\[
(1+t_k)\norm{(A\pm it_kE)^{-1}} \le 3\norm{E^{-1}}\widehat\gamma_{A,E},
\qquad
(1+t_k)\norm{(B\pm it_kD)^{-1}} \le 3\norm{D^{-1}}\widehat\gamma_{B,D}.
\]
These bounds give the plain-profile choices in \cref{thm:gen-main}; sharper nodewise pencil bounds can be used through a nonuniform generalized profile.
\end{remark}

\subsection{Generalized Lyapunov equations}
The generalized Lyapunov equation
\begin{equation}\label{eq:gen-lyap}
A^* X E + E^* X A = -Q
\end{equation}
is the specialization of \eqref{eq:gen-sylvester} obtained by taking the left pair $(A^*,E^*)$, the right pair $(A,E)$, and $C=-Q$. Equivalently, if
\[
\widetilde A := AE^{-1},
\qquad
\widetilde Q := E^{-*} Q E^{-1},
\]
then \eqref{eq:gen-lyap} is equivalent to
\[
\widetilde A^* X + X\widetilde A = -\widetilde Q.
\]
\begin{corollary}[Generalized Lyapunov equation via reduced sign embedding]\label{cor:lyap-main}
Assume \cref{ass:gen-be}, and consider \eqref{eq:gen-lyap} in transformed coordinates. Suppose the reduced augmented matrix of the transformed problem satisfies $\norm{\widetilde M}\le 1$ and admits a strip-resolvent bound $\gamma_{\mathrm{red}}\ge \gamma_{\Omega_a}(\widetilde M)$ for some $a\in(0,1)$ and $\beta\in(0,\arcsin a)$. Then \cref{thm:gen-main} yields a block-encoding of the solution once the corresponding quantities
\[
R_{A^*,E^*},\qquad R_{A,E},\qquad \Theta_{A^*,E^*;A,E}^{\mathrm{gen}}
\]
are bounded.

If, in addition,
\[
E\succ 0,\qquad
\norm{E^{-1/2}AE^{-1/2}} \le 1,
\qquad
H(E^{-1/2}AE^{-1/2}) \succeq \mu\I,
\]
and $(K,h)$ are chosen in balanced form so that \eqref{eq:Kchoice-sinc} holds with $\eps_{\mathrm{sgn}}=\eps$ and $\gamma=\gamma_{\mathrm{red}}$, then one obtains a block-encoding of $X$ with normalization
\begin{equation}\label{eq:lyap-beta}
\beta_X \le 36\Lambda^{\mathrm{syl}}_\star\frac{\norm{E^{-1}}^2}{\mu^2}
<21\frac{\norm{E^{-1}}^2}{\mu^2}
\end{equation}
and query complexity
\begin{equation}\label{eq:lyap-queries}
O\!\paren{\frac{\norm{E^{-1}}}{\mu}\log\frac{\norm{E^{-1}}}{\eps\mu}}
\text{ queries each to }U_A,U_E,
\qquad
O(1)\text{ queries to }U_Q.
\end{equation}
Sharper pointwise bounds for the weighted pencil shifts $(A^*\pm itE^*)^{-1}$ and $(A\pm itE)^{-1}$ decrease the overlap sum in \cref{thm:gen-main}.
\end{corollary}

\begin{proof}
Apply \cref{thm:gen-main} with the left pair $(A^*,E^*)$, the right pair $(A,E)$, and $C=-Q$.
In the weighted field-of-values regime, $E\succ 0$ implies $E^*=E$, so \cref{lem:gen-conditioning-R1} gives
\[
R_{A^*,E^*}\le \frac{3\norm{E^{-1}}}{\mu},
\qquad
R_{A,E}\le \frac{3\norm{E^{-1}}}{\mu},
\]
and therefore
\[
\Theta_{A^*,E^*;A,E}^{\mathrm{gen}}
\le
\frac{9\norm{E^{-1}}^2}{\mu^2}\Lambda^{\mathrm{syl}}_{K,h}
\le
\frac{9\norm{E^{-1}}^2}{\mu^2}\Lambda^{\mathrm{syl}}_\star.
\]
Substitution into the generalized Sylvester bounds gives the displayed normalization and query estimates.
\end{proof}

\section{Matrix functions: matrix roots and geometric means}\label{sec:sqrt}
Matrix-sign embeddings also apply when a matrix function occurs in an off-diagonal block of the sign of an augmented matrix. We consider principal square roots, inverse square roots, and matrix geometric means. The FoV and accretivity assumptions used below allow non-normal inputs whenever the relevant principal branches are defined.

We use the normalized coordinates of \cref{sec:model}.
For square roots, given an original matrix $A_0$, choose a scalar $\lambda_A\in(0,1]$ so that $A:=\lambda_A A_0$ satisfies $\norm{A}\le 1$.
Then
\[
A_0^{-1/2}=\lambda_A^{1/2}A^{-1/2},
\qquad
A_0^{1/2}=\lambda_A^{-1/2}A^{1/2},
\]
so scalar rescaling recovers the outputs for $A_0$.
The scalar case $A=\mu$ has $\|A^{-1/2}\|=\mu^{-1/2}$, which shows that an $O(\mu^{-1/2})$ normalization is optimal up to constants.

Here and below we define the principal geometric mean as
\begin{equation}\label{eq:geom-mean-def}
A\# B := A(A^{-1}B)^{1/2},
\end{equation}
whenever $\spec(A^{-1}B)\cap(-\infty,0]=\varnothing$.

For geometric means, given original matrices $A_0$ and $B_0$, choose a common scalar $\lambda_\#\in(0,1]$ so that $A:=\lambda_\# A_0$ and $B:=\lambda_\# B_0$ satisfy $\norm{A},\norm{B}\le 1$.
Then
\[
A\# B=\lambda_\#(A_0\# B_0),
\qquad
(A\# B)^{-1}=\lambda_\#^{-1}(A_0\# B_0)^{-1},
\]
so the normalized result determines the outputs for $A_0$ and $B_0$.  The factor $O((\mu_A\mu_B)^{-1/2})$ has the optimal order for $(A\# B)^{-1}$.  The direct geometric-mean output uses the normalization of the inverse-geometric-mean construction. 

\begin{theorem}[FoV-based quantum algorithms for $A^{-1/2}$ and $A^{1/2}$]\label{thm:sqrt-accretive}
Assume that $U_A$ is a $(1,a_A,0)$-block-encoding of $A$, that $\norm{A}\le 1$, and that
\[
H(A)\succeq \mu \I
\qquad
\text{for some }\mu\in(0,1].
\]
Choose any $a\in(0,\sqrt{\mu})$, any $\beta\in(0,\arcsin a)$, and choose $K$ so that
\begin{equation}\label{eq:sqrt-accretive-K}
K\ge
\left\lceil
\frac{1}{2\pi\beta}
\max\!\left\{
(\log 2)^2,
\log^2\!\left(1+\frac{2\wt C_{\beta,a,4/(\mu-a^2)}}{\eps}\right)
\right\}
\right\rceil.
\end{equation}
Set $h=\sqrt{2\pi\beta/K}$.
Then:
\begin{enumerate}[leftmargin=2em]
\item there exists a
\begin{equation}\label{eq:acc-minushalf}
\paren{\frac{4}{\sqrt{\mu}}, a_A+O(\log K), \eps}\text{-block-encoding of }A^{-1/2}
\end{equation}
using
\begin{equation}\label{eq:acc-minushalf-queries}
O\!\paren{\frac{1}{\mu}\log\frac{1}{\eps\mu}}
\end{equation}
queries to $U_A$ and $U_A^\dagger$;
\item there exists a
\begin{equation}\label{eq:acc-plushalf}
\paren{\frac{4}{\sqrt{\mu}}, 2a_A+O(\log K), \eps}\text{-block-encoding of }A^{1/2}
\end{equation}
with the same asymptotic query count and one additional use of $U_A$.
\end{enumerate}
\end{theorem}

The square-root construction uses the sign embedding $K(A)$; the geometric-mean construction uses the corresponding single-family shifted inverses.

\subsection{The square-root sign embedding}
\begin{theorem}[Principal square-root sign embedding]\label{thm:sqrt-sign}
Let $A$ be invertible and suppose that the principal square root $A^{1/2}$ is defined, equivalently
\[
\spec(A)\cap(-\infty,0]=\varnothing.
\]
Define the augmented matrix
\begin{equation}\label{eq:KofA}
K(A) := \begin{bmatrix} 0 & A \\ \I & 0 \end{bmatrix}.
\end{equation}
Then $\spec(K(A))\cap i\R = \varnothing$ and
\begin{equation}\label{eq:signK}
\signf(K(A)) =
\begin{bmatrix}
0 & A^{1/2} \\
A^{-1/2} & 0
\end{bmatrix}.
\end{equation}
\end{theorem}

\begin{proof}
We have
\[
K(A)^2 = \diag(A,A).
\]
If $\lambda\in\spec(A)$, then the eigenvalues of $K(A)$ are $\pm \sqrt{\lambda}$, where $\sqrt{\lambda}$ denotes the principal square root.
Because $\lambda\notin(-\infty,0]$, that principal square root has strictly positive real part, so neither $\sqrt{\lambda}$ nor $-\sqrt{\lambda}$ lies on the imaginary axis.
Hence $\spec(K(A))\cap i\R=\varnothing$ and $\signf(K(A))$ is well defined.

Since the principal square root of $A$ exists and $A$ is invertible, so does the principal inverse square root $A^{-1/2}$, and
\[
(K(A)^2)^{-1/2} = \diag(A^{-1/2},A^{-1/2}).
\]
The identity $\signf(M)=M(M^2)^{-1/2}$ therefore yields
\[
\signf(K(A))
= K(A)(K(A)^2)^{-1/2}
=
\begin{bmatrix}
0 & A \\
\I & 0
\end{bmatrix}
\begin{bmatrix}
A^{-1/2} & 0 \\
0 & A^{-1/2}
\end{bmatrix}
=
\begin{bmatrix}
0 & A^{1/2} \\
A^{-1/2} & 0
\end{bmatrix}.
\]
\end{proof}

\subsection{Log-sinc rational formulas for inverse and ordinary square roots}

Applying the general sign integral to $K(A)$ gives
\[
\signf(K(A))
=
\frac{2}{\pi}\int_0^{\infty} K(A)\bigl(K(A)^2+t^2\I\bigr)^{-1}\,dt.
\]
Since $K(A)^2=\diag(A,A)$, the entire problem reduces to the positive-shift family $A+t^2\I$.
\begin{lemma}[Resolvents of the square-root embedding]\label{lem:sqrt-resolvent}
For every $t\in\R$ such that $A+t^2\I$ is invertible,
\begin{equation}\label{eq:sqrt-resolvent}
(K(A)\mp it\I)^{-1}
=
(K(A)\pm it\I)\diag\!\paren{(A+t^2\I)^{-1},(A+t^2\I)^{-1}}
\end{equation}
and hence
\begin{equation}\label{eq:sqrt-resolvent-blocks}
(K(A)\mp it\I)^{-1}
=
\begin{bmatrix}
\pm it (A+t^2\I)^{-1} & A(A+t^2\I)^{-1} \\
(A+t^2\I)^{-1} & \pm it (A+t^2\I)^{-1}
\end{bmatrix}.
\end{equation}
\end{lemma}

\begin{proof}
Since $K(A)^2=\diag(A,A)$, we have
\[
(K(A)\mp it\I)(K(A)\pm it\I) = K(A)^2+t^2\I = \diag(A+t^2\I,A+t^2\I),
\]
which immediately gives \eqref{eq:sqrt-resolvent}.
Expanding the block product yields \eqref{eq:sqrt-resolvent-blocks}.
\end{proof}

\begin{corollary}[Symmetrized resolvent blocks]\label{cor:sqrt-blocks}
For every $t\in\R$ such that $A+t^2\I$ is invertible,
\begin{align}
\bracks{(K(A)-it\I)^{-1} + (K(A)+it\I)^{-1}}_{21} &= 2(A+t^2\I)^{-1}, \label{eq:21block}\\
\bracks{(K(A)-it\I)^{-1} + (K(A)+it\I)^{-1}}_{12} &= 2A(A+t^2\I)^{-1}. \label{eq:12block}
\end{align}
\end{corollary}

\begin{proof}
Add the two block formulas from \cref{lem:sqrt-resolvent}; the diagonal terms cancel and the off-diagonal terms add.
\end{proof}

Define the log-sinc approximants
\begin{align}
A^{-1/2}_{K,h} &:= \frac{2h}{\pi}\sum_{k=-K}^{K} t_k(A+t_k^2\I)^{-1}, \label{eq:ANminushalf}\\
A^{1/2}_{K,h} &:= \frac{2h}{\pi}\sum_{k=-K}^{K} t_k A(A+t_k^2\I)^{-1}. \label{eq:ANplushalf}
\end{align}

\begin{corollary}[Deterministic square-root approximation]\label{cor:sqrt-det}
Assume that $K(A)$ satisfies the hypotheses of \cref{thm:sign-approx} with parameters $(a,\beta,\gamma_K)$.
Then
\begin{equation}\label{eq:sqrt-det-error}
\norm{A^{-1/2} - A^{-1/2}_{K,h}} \le \eps_{\mathrm{sgn}}^{\mathrm{sinc}}(K,h;\beta,a,\gamma_K),
\qquad
\norm{A^{1/2} - A^{1/2}_{K,h}} \le \eps_{\mathrm{sgn}}^{\mathrm{sinc}}(K,h;\beta,a,\gamma_K).
\end{equation}
Under the balanced choice \eqref{eq:balanced-h}, both errors are at most $\wt C_{\beta,a,\gamma_K}e^{-\sqrt{2\pi\beta K}}$.
\end{corollary}

\begin{proof}
By \cref{thm:sqrt-sign}, the $(2,1)$-block of $\signf(K(A))$ is $A^{-1/2}$ and the $(1,2)$-block is $A^{1/2}$.
By \cref{cor:sqrt-blocks} and the definition \eqref{eq:SKh}, the corresponding blocks of $S_{K,h}(K(A))$ are exactly $A^{-1/2}_{K,h}$ and $A^{1/2}_{K,h}$.
Taking the relevant block of $\signf(K(A))-S_{K,h}(K(A))$ and using \cref{prop:sign-approx-explicit} proves the claim.
\end{proof}

\subsection{Scaled multiplexing and block-encoding normalizations}
Define
\begin{equation}\label{eq:sqrt-sk}
\mathcal R^{\mathrm{sq}}_k := (1+t_k^2)(A+t_k^2\I)^{-1},
\qquad
\nu_k := \frac{2h t_k}{\pi(1+t_k^2)}.
\end{equation}
Then
\begin{equation}\label{eq:sqrt-scaled-identity}
A^{-1/2}_{K,h} = \sum_{k=-K}^{K}\nu_k\mathcal R^{\mathrm{sq}}_k,
\qquad
A^{1/2}_{K,h} = A A^{-1/2}_{K,h}.
\end{equation}
Set
\begin{equation}\label{eq:Lambda-sq}
\Lambda^{\mathrm{sq}}_{K,h} := \sum_{k=-K}^{K}\abs{\nu_k}.
\end{equation}

\begin{lemma}[Total LCU weight for the square-root approximant]\label{lem:sqrt-weight}
For every $K\ge 1$ and $h>0$,
\begin{equation}\label{eq:sqrt-weight-bound}
\Lambda^{\mathrm{sq}}_{K,h}
\le
1+\frac{h}{\pi}.
\end{equation}
In particular, if $h\le \pi$, then
\begin{equation}\label{eq:sqrt-weight-balanced}
\Lambda^{\mathrm{sq}}_{K,h}\le 2.
\end{equation}
\end{lemma}

\begin{proof}
Since $t_k=e^{kh}$,
\[
\nu_k = \frac{h}{\pi}\operatorname{sech}(kh).
\]
The function $x\mapsto \operatorname{sech}(x)$ is positive, even, decreasing on $[0,\infty)$, and satisfies $\operatorname{sech}(0)=1$ and $\int_0^{\infty}\operatorname{sech}(x)\,dx=\pi/2$.
Thus
\[
h\sum_{k=-K}^{K}\operatorname{sech}(kh)
\le
h + 2\int_0^{\infty}\operatorname{sech}(x)\,dx
=
h+\pi,
\]
which proves \eqref{eq:sqrt-weight-bound}.
\end{proof}

Let
\begin{equation}\label{eq:Tsq}
T_{\mathrm{sq}} := \sum_{k=-K}^{K} \ket{k}\!\bra{k}\otimes \frac{A+t_k^2\I}{1+t_k^2}.
\end{equation}
Then
\begin{equation}\label{eq:Tsqinv}
T_{\mathrm{sq}}^{-1}
=
\sum_{k=-K}^{K} \ket{k}\!\bra{k}\otimes \mathcal R^{\mathrm{sq}}_k.
\end{equation}

\begin{lemma}[Block-encoding of the multiplexed positive shifts]\label{lem:sqrt-shift-be}
Let $U_A$ be a $(1,a_A,0)$-block-encoding of $A$.
Then there exists a block-encoding
\[
U_{T_{\mathrm{sq}}} \in \BE(1, a_A+O(1), 0)
\]
of $T_{\mathrm{sq}}$.
Each use of $U_{T_{\mathrm{sq}}}$ requires $O(1)$ queries to $U_A$.
\end{lemma}

\begin{proof}
For each $k$,
\[
\frac{A+t_k^2\I}{1+t_k^2}
=
\frac{1}{1+t_k^2}A
+
\frac{t_k^2}{1+t_k^2}\I,
\]
and the absolute values of the two coefficients sum to one.
Thus \cref{lem:mux-one-matrix} gives the claimed direct-sum block-encoding of $T_{\mathrm{sq}}$.
\end{proof}

\begin{definition}[Square-root nodewise profile]\label{def:sqrt-nodewise}
For each node $t_k=e^{kh}$, define the exact scaled positive-shift norm
\begin{equation}\label{eq:sqrt-nodewise}
m_k^{\mathrm{sq}} := (1+t_k^2)\norm{(A+t_k^2\I)^{-1}}.
\end{equation}
A \emph{square-root nodewise profile} is a collection of numbers $\rho_k^{\mathrm{sq}}>0$ such that 
\[
m_k^{\mathrm{sq}}\le \rho_k^{\mathrm{sq}}
\qquad
\text{for all }-K\le k\le K.
\]
Associated to such a profile, define
\begin{equation}\label{eq:sqrt-profile-params}
R_{\mathrm{sq},\rho}:=\max_{-K\le k\le K}\rho_k^{\mathrm{sq}},
\qquad
\Theta_\rho^{\mathrm{sq}}:=\sum_{k=-K}^{K}\abs{\nu_k}\rho_k^{\mathrm{sq}}.
\end{equation}
\end{definition}

\begin{remark}
The profile notation accommodates different nodewise bounds in later applications.  Substituting \eqref{eq:sqrt-profile-mu} gives the FoV-based corollary, while a nonuniform profile records the gain from nodewise rebalancing.
\end{remark}

\begin{theorem}[Generic sign embedding theorem for square roots]\label{thm:sqrt-generic}
Assume that $U_A$ is a $(1,a_A,0)$-block-encoding of $A$, that $\norm{A}\le 1$, and that $K(A)$ satisfies the hypotheses of \cref{thm:sign-approx} with parameters $(a,\beta,\gamma_K)$.
Let $K\ge 1$, $h>0$, let $\rho$ be a square-root nodewise profile in the sense of \cref{def:sqrt-nodewise}, and let $\eps_{\mathrm{inv}}\in(0,1]$.
Then:
\begin{enumerate}[leftmargin=2em]
\item there exists a
\begin{equation}\label{eq:sqrt-AN-be}
\paren{2\Theta_\rho^{\mathrm{sq}},\ a_A+O(\log K),\ \eps_{\mathrm{sgn}}^{\mathrm{sinc}}(K,h;\beta,a,\gamma_K)+\Theta_\rho^{\mathrm{sq}}\frac{\eps_{\mathrm{inv}}}{R_{\mathrm{sq},\rho}}}\text{-block-encoding of }A^{-1/2};
\end{equation}
\item multiplying by $U_A$ yields a
\begin{equation}\label{eq:sqrtplus-final-be}
\paren{2\Theta_\rho^{\mathrm{sq}},\ 2a_A+O(\log K),\ \eps_{\mathrm{sgn}}^{\mathrm{sinc}}(K,h;\beta,a,\gamma_K)+\Theta_\rho^{\mathrm{sq}}\frac{\eps_{\mathrm{inv}}}{R_{\mathrm{sq},\rho}}}\text{-block-encoding of }A^{1/2};
\end{equation}
\item both constructions use
\begin{equation}\label{eq:sqrt-generic-queries}
O\!\paren{R_{\mathrm{sq},\rho} \log\frac{R_{\mathrm{sq},\rho}}{\eps_{\mathrm{inv}}}}
\end{equation}
queries to $U_A$ and $U_A^\dagger$, up to one additional use of $U_A$ for the $A^{1/2}$ output.
\end{enumerate}
If one chooses the plain profile $\rho_k^{\mathrm{sq}}\equiv r_{\mathrm{sq}}$, then \eqref{eq:sqrt-AN-be} reduces to the plain normalization $2r_{\mathrm{sq}}\Lambda^{\mathrm{sq}}_{K,h}$.
\end{theorem}

\begin{proof}
Apply \cref{thm:single-family-algorithm} to the data
\[
J=\{-K,-K+1,\ldots,K\},
\qquad
F_k=A+t_k^2\I,
\qquad
\sigma_k=1+t_k^2,
\qquad
c_k=\nu_k.
\]
The multiplexed direct sum is exactly $T_{\mathrm{sq}}$, and the single-family blocks are exactly the scaled resolvents $\mathcal R_k^{\mathrm{sq}}$.
Therefore \cref{thm:single-family-algorithm} yields a block-encoding of $A^{-1/2}_{K,h}$ with normalization $2\Theta_\rho^{\mathrm{sq}}$ and implementation error $\Theta_\rho^{\mathrm{sq}}\eps_{\mathrm{inv}}/R_{\mathrm{sq},\rho}$.
Combining this with the deterministic quadrature error from \cref{cor:sqrt-det} proves \eqref{eq:sqrt-AN-be}.
Multiplying by the unit block-encoding of $A$ gives \eqref{eq:sqrtplus-final-be}; because $\norm{A}\le 1$, the same final error bound applies to $A^{1/2}$.
The query complexity is exactly the one from \cref{thm:single-family-algorithm} specialized to $T_{\mathrm{sq}}$.
\end{proof}

\subsection{FoV gap inputs}
\begin{lemma}[FoV-based positive-shift bounds]\label{lem:sqrt-conditioning}
Assume $\norm{A}\le 1$ and $H(A)\succeq \mu \I$ for some $\mu\in(0,1]$.
Then for every $t\in\R$,
\begin{equation}\label{eq:sqrt-conditioning}
\norm{\frac{A+t^2\I}{1+t^2}} \le 1,
\qquad
(1+t^2)\norm{(A+t^2\I)^{-1}} \le \frac{1+t^2}{\mu+t^2}\le \frac{1}{\mu}.
\end{equation}
Consequently, the profile
\begin{equation}\label{eq:sqrt-profile-mu}
\rho_k^{\mathrm{sq}} := \frac{1+t_k^2}{\mu+t_k^2}
\end{equation}
is valid and satisfies
\begin{equation}\label{eq:sqrt-profile-R}
R_{\mathrm{sq},\rho} = \frac{1}{\mu}.
\end{equation}
\end{lemma}

\begin{proof}
The norm bound is immediate:
\[
\norm{\frac{A+t^2\I}{1+t^2}} \le \frac{\norm{A}+t^2}{1+t^2}\le 1.
\]
For the inverse, the numerical range of $A+t^2\I$ is contained in the half-plane $\RePart z\ge \mu+t^2$, so
\[
\norm{(A+t^2\I)^{-1}} \le \frac{1}{\mu+t^2}.
\]
Multiplying by $1+t^2$ proves \eqref{eq:sqrt-conditioning}, and \eqref{eq:sqrt-profile-R} follows because the maximum is attained at $t_k=0$.
\end{proof}

\begin{lemma}[FoV-profile overlap bound for the square-root shifts]\label{lem:sqrt-overlap}
Let $\rho^{\mathrm{sq}}$ be the FoV-based profile \eqref{eq:sqrt-profile-mu}.
Then
\begin{equation}\label{eq:sqrt-overlap-bound}
\Theta_\rho^{\mathrm{sq}}
=
\frac{2h}{\pi}\sum_{k=-K}^{K}\frac{t_k}{\mu+t_k^2}
\le
\frac{1+h/\pi}{\sqrt{\mu}}.
\end{equation}
In particular, if $h\le \pi$, then
\begin{equation}\label{eq:sqrt-overlap-balanced}
\Theta_\rho^{\mathrm{sq}}\le \frac{2}{\sqrt{\mu}}.
\end{equation}
\end{lemma}

\begin{proof}
Set
\[
q_\mu(x):=\frac{e^x}{\mu+e^{2x}}.
\]
Then $q_\mu$ is nonnegative and unimodal on $\R$, since
\[
q_\mu'(x)=\frac{e^x(\mu-e^{2x})}{(\mu+e^{2x})^2}.
\]
Hence
\[
h\sum_{k\in\mathbb Z} q_\mu(kh)
\le
h\sup_{x\in\R}q_\mu(x)
+
\int_{-\infty}^{\infty}q_\mu(x)\,dx.
\]
The maximum is attained at $e^x=\sqrt{\mu}$ and equals $1/(2\sqrt{\mu})$.
Moreover,
\[
\int_{-\infty}^{\infty}q_\mu(x)\,dx
=
\int_0^\infty \frac{dt}{\mu+t^2}
=
\frac{\pi}{2\sqrt{\mu}}.
\]
Because the summand in \eqref{eq:sqrt-overlap-bound} is exactly $(2/\pi)q_\mu(kh)$, we obtain
\[
\Theta_\rho^{\mathrm{sq}}
\le
\frac{2}{\pi}\left(\frac{h}{2\sqrt{\mu}}+\frac{\pi}{2\sqrt{\mu}}\right)
=
\frac{1+h/\pi}{\sqrt{\mu}},
\]
which proves the claim.
\end{proof}

\begin{lemma}[Half-plane separation and strip resolvent bound for $K(A)$]\label{lem:sqrt-strip}
Assume $\norm{A}\le 1$ and $H(A)\succeq \mu \I$ for some $\mu\in(0,1]$.
Then $\norm{K(A)}=1$ and
\begin{equation}\label{eq:K-spectrum}
\spec(K(A)) \subset \set{z\in\C:\RePart(z)\ge \sqrt{\mu}} \cup \set{z\in\C:\RePart(z)\le -\sqrt{\mu}}.
\end{equation}
Moreover, for every $a\in(0,\sqrt{\mu})$,
\begin{equation}\label{eq:K-strip}
\gamma_{\Omega_a}(K(A))
\le \frac{2(1+a)}{\mu-a^2}
\le \frac{4}{\mu-a^2}.
\end{equation}
\end{lemma}

\begin{proof}
Since
\[
K(A)^*K(A) = \diag(\I,A^*A),
\]
we have $\norm{K(A)}^2 = \max\{1,\norm{A}^2\}=1$.
If $\lambda\in\spec(A)$, then $\RePart(\lambda)\ge \mu$ because $H(A)\succeq \mu\I$.
For the principal square root,
\[
\RePart(\sqrt{\lambda}) = \sqrt{\frac{|\lambda|+\RePart(\lambda)}{2}} \ge \sqrt{\RePart(\lambda)} \ge \sqrt{\mu}.
\]
Hence the eigenvalues $\pm\sqrt{\lambda}$ of $K(A)$ satisfy \eqref{eq:K-spectrum}.

Now fix $z=x+iy\in\Omega_a$, so $|x|\le a$.
Using
\[
(z\I-K(A))^{-1} = (z\I+K(A))\diag((z^2\I-A)^{-1},(z^2\I-A)^{-1}),
\]
we get
\[
\norm{(z\I-K(A))^{-1}} \le (\abs{z}+1)\norm{(z^2\I-A)^{-1}}.
\]
Furthermore,
\[
H(A-z^2\I) = H(A) - \RePart(z^2)\I = H(A) - (x^2-y^2)\I \succeq (\mu-a^2+y^2)\I,
\]
so
\[
\norm{(z^2\I-A)^{-1}} \le \frac{1}{\mu-a^2+y^2}.
\]
Since $\abs{z}\le a+\abs{y}$, it follows that
\[
\norm{(z\I-K(A))^{-1}}
\le \frac{1+a+\abs{y}}{\mu-a^2+y^2}.
\]
Setting $b:=\mu-a^2\in(0,1]$, we have $\abs{y}/(b+y^2)\le 1/(2\sqrt{b})\le 1/b$, and therefore
\[
\norm{(z\I-K(A))^{-1}} \le \frac{1+a}{b}+\frac{1}{b} \le \frac{2(1+a)}{b}.
\]
Taking the supremum over $z\in\Omega_a$ proves \eqref{eq:K-strip}.
\end{proof}

\subsection{Proof of the main square-root theorem}

\begin{proof}[Proof of \cref{thm:sqrt-accretive}]
By \cref{lem:sqrt-strip}, we have $\norm{K(A)}=1$ and $\gamma_{\Omega_a}(K(A))\le 4/(\mu-a^2)$.
Equation~\eqref{eq:sqrt-accretive-K} therefore implies
\[
\eps_{\mathrm{sgn}}^{\mathrm{sinc}}(K,h;\beta,a,\gamma_K)\le \frac{\eps}{2}
\qquad\text{with}\qquad
\gamma_K:=\frac{4}{\mu-a^2}.
\]
By \cref{lem:sqrt-conditioning}, the FoV-based profile \eqref{eq:sqrt-profile-mu} satisfies $R_{\mathrm{sq},\rho}=1/\mu$, and by \cref{lem:sqrt-overlap}, $\Theta_\rho^{\mathrm{sq}}\le 2/\sqrt{\mu}$ under the balanced choice.
Choose
\[
\eps_{\mathrm{inv}}:=\frac{\eps\sqrt{\mu}}{4}.
\]
Then
\[
\Theta_\rho^{\mathrm{sq}}\frac{\eps_{\mathrm{inv}}}{R_{\mathrm{sq},\rho}}
\le
\frac{2}{\sqrt{\mu}}\cdot \frac{\eps\sqrt{\mu}}{4}\cdot \mu
\le
\frac{\eps}{2},
\]
because $\mu\le 1$.
Applying \cref{thm:sqrt-generic} therefore yields total error at most $\eps$.
The normalization bounds are
\[
2\Theta_\rho^{\mathrm{sq}}\le \frac{4}{\sqrt{\mu}},
\qquad
2\Theta_\rho^{\mathrm{sq}}\le \frac{4}{\sqrt{\mu}}.
\]
The query estimate follows from
\[
O\!\paren{R_{\mathrm{sq},\rho}\log\frac{R_{\mathrm{sq},\rho}}{\eps_{\mathrm{inv}}}}
=
O\!\paren{\frac{1}{\mu}\log\frac{1}{\eps\mu}},
\]
which proves the query estimate.
\end{proof}

\begin{remark}
The profile \eqref{eq:sqrt-profile-mu} reduces the block-encoding normalization from $O(\mu^{-1})$ to $O(\mu^{-1/2})$.  The scalar inverse-square-root problem shows that this order is optimal.  Since the assumption involves only the FoV, the result also covers non-normal and non-diagonalizable matrices.
\end{remark}

\subsection{Extension to geometric mean}
The geometric-mean construction transports the square-root sign embedding to the original coordinates.  Even when $\norm{A}\le 1$, $\norm{B}\le 1$, and $H(A),H(B)\succ0$, the unscaled embedding
\[
G(A,B)=
\begin{bmatrix}
0&B\\
A^{-1}&0
\end{bmatrix}
\]
can have norm greater than one.  The algorithm therefore includes an explicit positive sign-scaling parameter.  Positive scalar multiplication leaves the half-plane sign, its projector, and the extracted geometric mean unchanged.

Assume $A$ and $B$ are invertible and satisfy
\[
\spec(A^{-1}B)\cap(-\infty,0]=\emptyset.
\]
Define the geometric-mean sign embedding
\begin{equation}\label{eq:GAB}
G(A,B):=
\begin{bmatrix}
0 & B\\
A^{-1} & 0
\end{bmatrix}.
\end{equation}
Fix a known scaling parameter $\chi_{\#}>0$ and set
\begin{equation}\label{eq:GAB-scaled}
M_{\#}:=\chi_{\#}G(A,B).
\end{equation}
For the log-sinc nodes $t_k=e^{kh}$, define the induced geometric-mean nodes
\begin{equation}\label{eq:geom-scaled-nodes}
s_k:=\frac{t_k}{\chi_{\#}}.
\end{equation}
Then define
\begin{align}
(A\# B)^{-1}_{K,h}
&:= \frac{2h}{\pi}\sum_{k=-K}^{K} s_k(B+s_k^2A)^{-1},
\label{eq:geom-mean-inv-approx}\\
(A\# B)_{K,h}
&:= B(A\# B)^{-1}_{K,h}A.
\label{eq:geom-mean-approx}
\end{align}

For these nodes, define the exact scaled inverse norms
\[
m_k^{\#}:=(1+s_k^2)\norm{(B+s_k^2A)^{-1}}.
\]
A \emph{geometric-mean nodewise profile} is a collection of numbers $\rho_k^\#>0$ such that
\[
m_k^\#\le \rho_k^\#
\qquad
\text{for all }-K\le k\le K.
\]
Associated to such a profile, define
\[
R_{\#,\rho}:=\max_{-K\le k\le K}\rho_k^\#,
\qquad
\Theta_\rho^{\#}:=\sum_{k=-K}^{K}\abs{\nu_k^\#}\rho_k^\#,
\qquad
\nu_k^\#:=\frac{2h s_k}{\pi(1+s_k^2)}.
\]

\begin{proposition}[Geometric-mean sign embedding]\label{prop:geom-mean-sign}
Set $C:=A^{-1}B$ and $S:=\diag(A,\I)$. Then
\[
G(A,B)=S K(C) S^{-1},
\qquad
\signf(G(A,B))
=
\begin{bmatrix}
0 & A\# B\\
(A\# B)^{-1} & 0
\end{bmatrix}.
\]
Moreover, for every $u\in\R$ such that $B+u^2A$ is invertible,
\begin{align}
\bigl[(G(A,B)-iu\I)^{-1}+(G(A,B)+iu\I)^{-1}\bigr]_{21}
&= 2(B+u^2A)^{-1},
\label{eq:geom-mean-21}\\
\bigl[(G(A,B)-iu\I)^{-1}+(G(A,B)+iu\I)^{-1}\bigr]_{12}
&= 2\,B(B+u^2A)^{-1}A.
\label{eq:geom-mean-12}
\end{align}
Consequently, if $M_{\#}=\chi_{\#}G(A,B)$ satisfies the hypotheses of \cref{thm:sign-approx} with parameters $(a,\beta,\gamma_{\#})$, then
\begin{equation}\label{eq:geom-mean-det}
\norm{(A\# B)^{-1}-(A\# B)^{-1}_{K,h}}
\le
\eps_{\mathrm{sgn}}^{\mathrm{sinc}}(K,h;\beta,a,\gamma_{\#}),
\end{equation}
and
\begin{equation}\label{eq:geom-mean-det-2}
\norm{A\# B-(A\# B)_{K,h}}
\le
\eps_{\mathrm{sgn}}^{\mathrm{sinc}}(K,h;\beta,a,\gamma_{\#}).
\end{equation}
\end{proposition}

\begin{proof}
Because $C=A^{-1}B$, one has $G(A,B)=SK(C)S^{-1}$ with $S=\diag(A,\I)$.
By \cref{thm:sqrt-sign},
\[
\signf(K(C))=
\begin{bmatrix}
0 & C^{1/2}\\
C^{-1/2} & 0
\end{bmatrix},
\]
and therefore
\[
\signf(G(A,B))
= S\,\signf(K(C))\,S^{-1}
=
\begin{bmatrix}
0 & A C^{1/2}\\
C^{-1/2}A^{-1} & 0
\end{bmatrix}
=
\begin{bmatrix}
0 & A\# B\\
(A\# B)^{-1} & 0
\end{bmatrix}.
\]

Likewise,
\[
(z\I-G(A,B))^{-1}=S(z\I-K(C))^{-1}S^{-1}.
\]
Applying \cref{cor:sqrt-blocks} to $K(C)$ and using
\[
C+u^2\I=A^{-1}(B+u^2A)
\]
gives \eqref{eq:geom-mean-21} and \eqref{eq:geom-mean-12}.

For the positive scaling $\chi_{\#}$,
\[
M_{\#}=S\bigl(\chi_{\#}K(C)\bigr)S^{-1},
\]
we have, for $s_k=t_k/\chi_{\#}$,
\[
(\chi_{\#}K(C)\mp it_k\I)^{-1}
=
\chi_{\#}^{-1}(K(C)\mp is_k\I)^{-1}.
\]
Thus the $(2,1)$ and $(1,2)$ blocks of
\[
(M_{\#}-it_k\I)^{-1}+(M_{\#}+it_k\I)^{-1}
\]
are respectively
\[
2\chi_{\#}^{-1}(B+s_k^2A)^{-1},
\qquad
2\chi_{\#}^{-1}B(B+s_k^2A)^{-1}A.
\]
Multiplying by the log-sinc factor $h t_k/\pi=h\chi_{\#}s_k/\pi$ gives exactly the two approximants \eqref{eq:geom-mean-inv-approx} and \eqref{eq:geom-mean-approx}.  Taking the corresponding blocks of
\[
\signf(M_{\#})-S_{K,h}(M_{\#})
\]
and using $\signf(M_{\#})=\signf(G(A,B))$ proves \eqref{eq:geom-mean-det} and \eqref{eq:geom-mean-det-2} by \cref{prop:sign-approx-explicit}.
\end{proof}

\begin{theorem}[Geometric mean from the scaled sign embedding]\label{thm:geom-mean}
Assume the standing spectral condition
\[
\spec(A^{-1}B)\cap(-\infty,0]=\emptyset,
\]
and use the normalized coordinates specified above, so that
\[
\norm{A}\le 1,
\qquad
\norm{B}\le 1,
\]
and assume block-encodings
\[
U_A\in \BE(1,a_A,0),
\qquad
U_B\in \BE(1,a_B,0)
\]
are given. Assume also that a known scaling $\chi_{\#}>0$ and parameters $(a,\beta,\gamma_{\#})$ are available such that
\begin{equation}\label{eq:geom-sign-bound}
\norm{M_{\#}}\le 1,
\qquad
\gamma_{\#}\ge \gamma_{\Omega_a}(M_{\#}),
\qquad
M_{\#}:=\chi_{\#}G(A,B).
\end{equation}

Then for every $K\ge 1$, every $h>0$, every geometric-mean nodewise profile $\rho^\#$,
and every inverse target precision $\eps_{\mathrm{inv}}\in(0,1]$, there exists a
\begin{equation}\label{eq:geom-mean-inv-be}
\paren{
2\Theta_\rho^{\#},
\ a_A+a_B+O(\log K),
\ \eps_{\mathrm{sgn}}^{\mathrm{sinc}}(K,h;\beta,a,\gamma_{\#})
+\Theta_\rho^{\#}\frac{\eps_{\mathrm{inv}}}{R_{\#,\rho}}
}
\text{-block-encoding of }(A\# B)^{-1},
\end{equation}
using
\begin{equation}\label{eq:geom-mean-inv-q}
O\!\paren{R_{\#,\rho}\log\frac{R_{\#,\rho}}{\eps_{\mathrm{inv}}}}
\end{equation}
queries each to $U_A$, $U_B$, and their adjoints. Moreover, there exists a
\begin{equation}\label{eq:geom-mean-be}
\paren{
2\Theta_\rho^{\#},
\ 2a_A+2a_B+O(\log K),
\ \eps_{\mathrm{sgn}}^{\mathrm{sinc}}(K,h;\beta,a,\gamma_{\#})
+\Theta_\rho^{\#}\frac{\eps_{\mathrm{inv}}}{R_{\#,\rho}}
}
\text{-block-encoding of }A\# B,
\end{equation}
with the same asymptotic query complexity and one additional use each of $U_B$ and $U_A$.
\end{theorem}

\begin{proof}
By \cref{prop:geom-mean-sign}, the deterministic quadrature error is exactly the one stated in
\eqref{eq:geom-mean-det}--\eqref{eq:geom-mean-det-2}. Define
\[
T_{\#}:=
\sum_{k=-K}^{K}\ket{k}\!\bra{k}\otimes \frac{B+s_k^2A}{1+s_k^2}.
\]
For each $k$, the coefficients of $B$ and $A$ in the $k$th block have absolute values summing to one, so \cref{lem:mux-two-matrix} gives a unit block-encoding of $T_{\#}$ using $O(1)$ queries each to $U_A$ and $U_B$.

Apply \cref{thm:single-family-algorithm} with
\[
J=\{-K,-K+1,\ldots,K\},
\qquad
F_k=B+s_k^2A,
\qquad
\sigma_k=1+s_k^2,
\qquad
c_k=\nu_k^\#.
\]
Then the single-family target is exactly the inverse geometric-mean approximant
\[
(A\# B)^{-1}_{K,h}
=
\frac{2h}{\pi}\sum_{k=-K}^{K} s_k(B+s_k^2A)^{-1}.
\]
Thus \cref{thm:single-family-algorithm}, together with the deterministic error from
\cref{prop:geom-mean-sign}, yields \eqref{eq:geom-mean-inv-be} and
\eqref{eq:geom-mean-inv-q}. Multiplying on the left by $U_B$ and on the right by $U_A$ gives
\eqref{eq:geom-mean-be}; because $\norm{A},\norm{B}\le 1$ in normalized coordinates, the same final error bound applies to $A\# B$.
\end{proof}

\begin{theorem}[FoV/profile specialization for $A\# B$ and $(A\# B)^{-1}$]\label{thm:geom-mean-accretive}
Assume the hypotheses of \cref{thm:geom-mean}, including the principal-branch condition and the scaled strip-resolvent bound \eqref{eq:geom-sign-bound}. Suppose in addition that
\[
H(A)\succeq \mu_A \I,
\qquad
H(B)\succeq \mu_B \I,
\qquad
\mu_A,\mu_B\in(0,1],
\]
and set $\mu_{\#}:=\min\{\mu_A,\mu_B\}$. Then for every $\eps\in(0,1]$ there exists a
\begin{equation}\label{eq:geom-mean-fov-inv-be}
\paren{
\frac{4}{\sqrt{\mu_A\mu_B}},
\ a_A+a_B+O(\log K_\eps),
\ \eps
}
\text{-block-encoding of }(A\# B)^{-1},
\end{equation}
using
\begin{equation}\label{eq:geom-mean-fov-inv-q}
O\!\paren{
\frac{1}{\mu_{\#}}
\log\max\!\left\{
\frac{1}{\mu_{\#}},
\frac{1}{\eps\mu_A\mu_B}
\right\}
}\leq O\!\paren{
\frac{1}{\min\{\mu_A,\mu_B\}}
\log(\frac{1}{\eps \min\{\mu_A,\mu_B\}})
}
\end{equation}
queries each to $U_A$, $U_B$, and their adjoints. Moreover, there exists a
\begin{equation}\label{eq:geom-mean-fov-be}
\paren{
\frac{4}{\sqrt{\mu_A\mu_B}},
\ 2a_A+2a_B+O(\log K_\eps),
\ \eps
}
\text{-block-encoding of }A\# B,
\end{equation}
with the same asymptotic query complexity and one additional use each of $U_B$ and $U_A$.
\end{theorem}

\begin{proof}
Choose the balanced log-sinc step
\[
h=\sqrt{\frac{2\pi\beta}{K}},
\]
and choose $K=K_\eps$ large enough that
\[
h\le \pi
\qquad\text{and}\qquad
\eps_{\mathrm{sgn}}^{\mathrm{sinc}}(K,h;\beta,a,\gamma_{\#})\le \frac{\eps}{2}.
\]
Such a choice exists by \cref{thm:sign-approx}; for example,
$K_\eps=O\!\paren{\beta^{-1}\log^2(\wt C_{\beta,a,\gamma_{\#}}/\eps)}$ suffices up to absolute constants.  The scaled strip-resolvent bound determines how $\chi_{\#}$ and $\gamma_{\#}$ affect the quadrature size.  The inverse-stage query bound below is independent of both parameters.

For every $s\ge 0$, the FoV hypotheses give
\[
H(B+s^2A)=H(B)+s^2H(A)\succeq (\mu_B+\mu_A s^2)\I,
\]
and hence
\[
\norm{(B+s^2A)^{-1}} \le \frac{1}{\mu_B+\mu_A s^2}.
\]
Therefore the nodewise choice
\begin{equation}\label{eq:geom-mean-profile}
\rho_k^\# := \frac{1+s_k^2}{\mu_B+\mu_A s_k^2}
\end{equation}
is a valid geometric-mean profile. Its maximum is attained at $s=0$ or in the limit
$s\to\infty$, so
\begin{equation}\label{eq:geom-mean-profile-R}
R_{\#,\rho} = \frac{1}{\mu_{\#}}.
\end{equation}

For the overlap parameter, apply the shifted unimodal-sum argument from \cref{lem:sqrt-overlap} to
\[
q_{\mu_A,\mu_B}(x):=\frac{e^x}{\mu_B+\mu_A e^{2x}}.
\]
Since $s_k=e^{kh-\log\chi_{\#}}$, the summand is sampled on a shifted logarithmic lattice.  For any shift $s_0\in\R$ and any nonnegative unimodal integrable function $q$,
\[
h\sum_{k\in\mathbb Z} q(kh+s_0)
\le
h\sup_{x\in\mathbb R}q(x)+\int_{-\infty}^{\infty}q(x)\,dx.
\]
The maximum of $q_{\mu_A,\mu_B}$ is $1/(2\sqrt{\mu_A\mu_B})$, and
\[
\int_{-\infty}^{\infty} q_{\mu_A,\mu_B}(x)\,dx
=
\int_0^\infty \frac{ds}{\mu_B+\mu_A s^2}
=
\frac{\pi}{2\sqrt{\mu_A\mu_B}}.
\]
Since
\[
\Theta_\rho^{\#}
=
\frac{2h}{\pi}\sum_{k=-K}^{K}\frac{s_k}{\mu_B+\mu_A s_k^2},
\]
we obtain
\begin{equation}\label{eq:geom-mean-profile-Theta}
\Theta_\rho^{\#}
\le \frac{1+h/\pi}{\sqrt{\mu_A\mu_B}}
\le \frac{2}{\sqrt{\mu_A\mu_B}},
\qquad h\le\pi.
\end{equation}

Choose
\begin{equation}\label{eq:geom-mean-epsinv-choice}
\eps_{\mathrm{inv}}
:=
\min\!\left\{
1,\,
\frac{\eps\sqrt{\mu_A\mu_B}}{4\mu_{\#}}
\right\}.
\end{equation}
Then the implementation term in \cref{thm:geom-mean} satisfies
\[
\Theta_\rho^{\#}\frac{\eps_{\mathrm{inv}}}{R_{\#,\rho}}
\le
\frac{2}{\sqrt{\mu_A\mu_B}}
\cdot
\frac{\eps\sqrt{\mu_A\mu_B}}{4\mu_{\#}}
\cdot
\mu_{\#}
=
\frac{\eps}{2}.
\]
The deterministic term is at most $\eps/2$ by the choice of $K$ and $h$.
Applying \cref{thm:geom-mean} therefore gives the two $\eps$-accurate block-encodings.
The normalization bound follows from
\[
2\Theta_\rho^{\#}\le \frac{4}{\sqrt{\mu_A\mu_B}}.
\]
Finally,
\[
R_{\#,\rho}\log\frac{2R_{\#,\rho}}{\eps_{\mathrm{inv}}}
=
\frac{1}{\mu_{\#}}
\log\!\left(
\frac{2}{\mu_{\#}}
\max\!\left\{1,\frac{4\mu_{\#}}{\eps\sqrt{\mu_A\mu_B}}\right\}
\right)
=
\frac{1}{\mu_{\#}}
\log\max\!\left\{
\frac{2}{\mu_{\#}},
\frac{8}{\eps\sqrt{\mu_A\mu_B}}
\right\},
\]
which proves \eqref{eq:geom-mean-fov-inv-q}. The block-encoding of $A\#B$ uses this inverse stage followed by multiplication by $B$ on the left and $A$ on the right.
\end{proof}

\begin{remark}
The FoV assumptions in \cref{thm:geom-mean-accretive} allow non-normal matrices.  Sign approximation of $G(A,B)$ additionally requires the principal-branch condition and the scaled strip-resolvent bound \eqref{eq:geom-sign-bound}.  Under these assumptions, the normalization $O((\mu_A\mu_B)^{-1/2})$ has the optimal order for the inverse geometric mean.  The block-encoding of $A\#B$ uses the normalization of the inverse-geometric-mean construction.
\end{remark}

\section{Continuous-time algebraic Riccati equation}\label{sec:care}
The preceding matrix functions are obtained from off-diagonal blocks of a matrix sign. For the continuous-time algebraic Riccati equation (CARE), the stabilizing solution is extracted from a spectral projector of the Hamiltonian sign.

The continuous-time algebraic Riccati equation is written as
\begin{equation}\label{eq:care}
A^*X + XA - XGX + Q = 0,
\qquad
Q=Q^*,
\quad
G=G^*.
\end{equation}

Classically, the stabilizing solution is recovered from the half-plane matrix sign of the Hamiltonian matrix
\begin{equation}\label{eq:care-H}
H := \begin{bmatrix} A & -G \\ -Q & -A^* \end{bmatrix},
\end{equation}
a construction used in the matrix-sign methods of Roberts and Byers and in the treatments of Higham and Lancaster--Rodman \cite{Roberts1980,Byers1987,Higham2008,LancasterRodman1995}. The analysis below applies this Hamiltonian projector to the standard CARE.

\begin{notation}[CARE projector data]\label{not:care-projector-data}
Whenever $\spec(H)\cap i\R=\varnothing$, write
\[
\Pi_- := \frac{\I-\signf(H)}{2}
=
\begin{bmatrix}
\Pi_{11} & \Pi_{12} \\
\Pi_{21} & \Pi_{22}
\end{bmatrix}.
\]
If $\Pi_{11}$ is invertible, define
\[
X := \Pi_{21}\Pi_{11}^{-1},
\qquad
p_H := \norm{\Pi_-},
\qquad
\sigma := \sigma_{\min}(\Pi_{11}).
\]
\end{notation}

\begin{theorem}[CARE via Hamiltonian sign projectors]\label{thm:care-main}
Assume $\spec(H)\cap i\R=\varnothing$ and that $\Pi_{11}$ is invertible, so that the quantities in \cref{not:care-projector-data} are well defined. Suppose a sign stage provides a $(\beta_{\mathrm{sign}},a_H+O(\log K),\eps_{\mathrm{sign}})$-block-encoding of $\signf(H)$ with $\eps_{\mathrm{sign}}<2\sigma$, and let
\[
\beta_\Pi := \frac{1+\beta_{\mathrm{sign}}}{2}.
\]
Then for every inverse target precision $\eps_{\Pi^{-1}}\in(0,1]$ there exists a
\begin{equation}\label{eq:care-final-be}
\paren{\beta_X,\, 2a_H+O(\log K),\, \eps_X}\text{-block-encoding of }X,
\end{equation}
where
\begin{equation}\label{eq:care-final-beta}
\beta_X \le \frac{1+\beta_{\mathrm{sign}}}{\sigma-\eps_{\mathrm{sign}}/2}
\end{equation}
and
\begin{equation}\label{eq:care-final-error}
\eps_X \le
\frac{\eps_{\mathrm{sign}}}{2}\,
\frac{1+\norm{X}}{\sigma-\eps_{\mathrm{sign}}/2}
+
\paren{p_H+\frac{\eps_{\mathrm{sign}}}{2}}\eps_{\Pi^{-1}}.
\end{equation}
The projector-inversion stage uses
\begin{equation}\label{eq:care-pi11-queries}
O\!\paren{
\frac{\beta_\Pi}{2\sigma-\eps_{\mathrm{sign}}}
\log\frac{1}{(2\sigma-\eps_{\mathrm{sign}})\eps_{\Pi^{-1}}}
}
\end{equation}
queries to the projector block-encoding.
If the sign stage is instantiated by \cref{thm:care-sign}, set
\begin{equation}\label{eq:care-Qsign}
Q_{\mathrm{sign}}
:=
O\!\paren{
R_H\log\frac{R_H}{\eps_H}
}.
\end{equation}
Then the total number of queries to the Hamiltonian block-encoding oracle $U_H$ and its adjoint is
\begin{equation}\label{eq:care-total-UH-queries}
Q_{\mathrm{total}}
=
O\!\left(
Q_{\mathrm{sign}}
\left[
1+
\frac{\beta_\Pi}{2\sigma-\eps_{\mathrm{sign}}}
\log\frac{1}{(2\sigma-\eps_{\mathrm{sign}})\eps_{\Pi^{-1}}}
\right]
\right).
\end{equation}
If $H=V\Lambda V^{-1}$ is diagonalizable, then $p_H\le \kappa(V)$. If, moreover, $Q\succeq 0$, $G\succeq 0$, $(A,G^{1/2})$ is stabilizable, and $(A,Q^{1/2})$ is detectable, then $X$ is the unique stabilizing Hermitian solution of \eqref{eq:care}.
\end{theorem}

The construction consists of the stable-projector quotient, a Hamiltonian sign approximation, and inversion of the leading projector block.

\subsection{Hamiltonian sign-projector formulation}
Let
\[
J := \begin{bmatrix} 0 & \I \\ -\I & 0 \end{bmatrix}.
\]

\begin{proposition}[Hamiltonian symmetry and spectral pairing]\label{prop:care-ham}
For the Hamiltonian matrix $H$ in \eqref{eq:care-H},
\[
H^*J = -JH.
\]
Equivalently, $H$ and $-H^*$ are similar, so
\[
\spec(H) = -\overline{\spec(H)}.
\]
In particular, if $\spec(H)\cap i\R=\varnothing$, then $H$ has exactly $n$ eigenvalues in each open half-plane, counted with algebraic multiplicity.
\end{proposition}

\begin{proof}
A direct block computation gives
\[
H^* =
\begin{bmatrix}
A^* & -Q \\
-G & -A
\end{bmatrix},
\qquad
H^*J =
\begin{bmatrix}
Q & A^* \\
A & -G
\end{bmatrix}
=-JH.
\]
Since $J^{-1}=-J$, we have $H = -J^{-1}H^*J$, so $H$ is similar to $-H^*$.
Taking spectra gives $\spec(H)=\spec(-H^*)=-\overline{\spec(H)}$.
If no eigenvalue lies on the imaginary axis, this symmetry pairs eigenvalues across the imaginary axis and therefore forces exactly $n$ of the $2n$ eigenvalues into each open half-plane.
\end{proof}

\begin{lemma}[Graph extraction from a rank-$n$ projector]\label{lem:care-graph}
Let $\Pi\in\C^{2n\times 2n}$ satisfy $\Pi^2=\Pi$ and $\rank(\Pi)=n$.
Write
\[
\Pi=
\begin{bmatrix}
\Pi_{11} & \Pi_{12} \\
\Pi_{21} & \Pi_{22}
\end{bmatrix}
\]
with $n\times n$ blocks.
If $\Pi_{11}$ is invertible, then
\[
\range(\Pi)=\range
\begin{bmatrix}
\I \\
\Pi_{21}\Pi_{11}^{-1}
\end{bmatrix}.
\]
\end{lemma}

\begin{proof}
The first $n$ columns of $\Pi$ are
\[
\begin{bmatrix}
\Pi_{11} \\
\Pi_{21}
\end{bmatrix}.
\]
Because $\Pi_{11}$ is invertible, these columns are linearly independent. They therefore span an $n$-dimensional subspace contained in $\range(\Pi)$, and since $\rank(\Pi)=n$, they already span the whole range.
Right-multiplying by the invertible matrix $\Pi_{11}^{-1}$ does not change the range, so
\[
\range(\Pi)=\range
\begin{bmatrix}
\Pi_{11} \\
\Pi_{21}
\end{bmatrix}
=
\range
\begin{bmatrix}
\I \\
\Pi_{21}\Pi_{11}^{-1}
\end{bmatrix}.
\]
\end{proof}

\begin{theorem}[CARE solution from the left spectral projector]\label{thm:care-projector}
Assume $\spec(H)\cap i\R=\varnothing$, and let
\[
\Pi_- := \frac{\I-\signf(H)}{2}
=
\begin{bmatrix}
\Pi_{11} & \Pi_{12} \\
\Pi_{21} & \Pi_{22}
\end{bmatrix}
\]
be the spectral projector onto the stable invariant subspace of $H$.
Suppose $\Pi_{11}$ is invertible, and define
\begin{equation}\label{eq:care-projector-ratio}
X := \Pi_{21}\Pi_{11}^{-1}.
\end{equation}
Then $X$ solves the CARE \eqref{eq:care}, and the closed-loop matrix $A-GX$ is Hurwitz.
If, in addition, $Q\succeq 0$, $G\succeq 0$, $(A,G^{1/2})$ is stabilizable, and $(A,Q^{1/2})$ is detectable, then $X$ is the unique stabilizing Hermitian solution.
\end{theorem}

\begin{proof}
By \cref{prop:care-ham}, the stable subspace has dimension $n$.
By \cref{lem:care-graph},
\[
\range(\Pi_-)=\range
\begin{bmatrix}
\I \\
X
\end{bmatrix}.
\]
Because $\Pi_-$ is a spectral projector of $H$, its range is $H$-invariant.
Hence there exists $R_X\in\C^{n\times n}$ such that
\[
H
\begin{bmatrix}
\I \\
X
\end{bmatrix}
=
\begin{bmatrix}
\I \\
X
\end{bmatrix}
R_X.
\]
Expanding with the block form \eqref{eq:care-H} gives
\[
\begin{bmatrix}
A-GX \\
-Q-A^*X
\end{bmatrix}
=
\begin{bmatrix}
R_X \\
XR_X
\end{bmatrix}.
\]
Therefore $R_X=A-GX$ and
\[
-Q-A^*X = X(A-GX) = XA-XGX,
\]
which rearranges to \eqref{eq:care}.
Moreover,
\[
\spec(A-GX)=\spec(R_X)\subset\{z\in\C:\RePart z<0\},
\]
because $R_X$ is the restriction of $H$ to its stable invariant subspace.
Hence $A-GX$ is Hurwitz.
Under the additional control-theoretic assumptions, the stabilizing Hermitian solution is unique, so the extracted $X$ coincides with it.
\end{proof}

\subsection{Hamiltonian sign approximation with the log-sinc rule}
We assume block-encoding access to the Hamiltonian.

Because the half-plane matrix sign is invariant under positive rescaling,
\[
\signf(\lambda H)=\signf(H)
\qquad
(\lambda>0),
\]
the spectral projector $\Pi_-=(\I-\signf(H))/2$ and the extracted Riccati solution are unchanged by replacing $H$ with $\lambda H$.
Hence we may and do work in normalized Hamiltonian coordinates.

\begin{assumption}[Unit-normalized Hamiltonian block-encoding]\label{ass:care-be}
We are given a $(1,a_H,0)$-block-encoding $U_H$ of a Hamiltonian matrix $H$ satisfying $\norm{H}\le 1$.
\end{assumption}

\begin{remark}[Building $H$ from $A,G,Q$]
If unit block-encodings of normalized $A$, $G$, and $Q$ are given separately, then standard direct-sum, adjoint, and sign-flip constructions, followed by one final normalization of the Hamiltonian, yield a unit block-encoding of $H$.
We keep $U_H$ abstract because the Riccati chapter only needs the Hamiltonian block-encoding itself.
\end{remark}

For the log-sinc nodes $t_k=e^{kh}$ from \eqref{eq:logsinc-nodes}, define the sign-augmented multiplexed Hamiltonian
\begin{equation}\label{eq:care-direct-sums}
K_H
:=
\sum_{\sigma\in\{+,-\}}\sum_{k=-K}^{K}\ket{\sigma,k}\!\bra{\sigma,k}\otimes \frac{H+\sigma it_k\I}{1+t_k}.
\end{equation}
Its inverse blocks are
\begin{equation}\label{eq:care-direct-sum-inverse}
K_H^{-1}
=
\sum_{\sigma\in\{+,-\}}\sum_{k=-K}^{K}\ket{\sigma,k}\!\bra{\sigma,k}\otimes (1+t_k)(H+\sigma it_k\I)^{-1}.
\end{equation}

\begin{lemma}[Block-encodings and plain conditioning of the multiplexed Hamiltonian shifts]\label{lem:care-mux}
Assume \cref{ass:care-be}.
Then there exists a block-encoding
\[
U_{K_H}\in \BE(1,a_H+O(1),0)
\]
of $K_H$, and each use of $U_{K_H}$ requires $O(1)$ queries to $U_H$.
If
\[
\gamma_H := \sup_{z\in\Omega_a}\norm{(z\I-H)^{-1}} < \infty
\]
for some $a\in(0,1)$, then the plain quantity
\begin{equation}\label{eq:care-rH}
r_H := \max_{k,\pm}(1+t_k)\norm{(H \pm it_k\I)^{-1}}
\end{equation}
satisfies
\begin{equation}\label{eq:care-rH-bound}
r_H \le 3\gamma_H.
\end{equation}
\end{lemma}

\begin{proof}
For fixed $k$ and a sign,
\[
\frac{H \pm it_k\I}{1+t_k}
=
\frac{1}{1+t_k}\,H
+
\frac{\pm i t_k}{1+t_k}\,\I,
\]
and the absolute values of the two coefficients sum to one.
Hence \cref{lem:mux-one-matrix} gives the claimed block-encoding of \eqref{eq:care-direct-sums}.

For the plain inverse norm, repeat the argument in \cref{lem:R2-conditioning}.
If $t_k\le 2$, then $\pm it_k\in\Omega_a$, so
\[
\norm{(H \pm it_k\I)^{-1}} \le \gamma_H
\]
and hence $(1+t_k)\norm{(H \pm it_k\I)^{-1}}\le 3\gamma_H$.
If $t_k\ge 2$, then $\norm{H}\le 1$ gives the Neumann bound
\[
\norm{(H \pm it_k\I)^{-1}} \le \frac{1}{t_k-1},
\]
so
\[
(1+t_k)\norm{(H \pm it_k\I)^{-1}}
\le
\frac{1+t_k}{t_k-1}
\le 3.
\]
Since $0\in\Omega_a$ and $\norm{H}\le 1$, we also have $\gamma_H\ge \norm{H^{-1}}\ge 1$, so the large-$t_k$ bound is already dominated by $3\gamma_H$.
Taking the maximum proves \eqref{eq:care-rH-bound}.
\end{proof}

Define
\begin{equation}\label{eq:care-weights}
w_k := \frac{h t_k}{\pi(1+t_k)},
\qquad
\Lambda^{\mathrm{care}}_{K,h} := \sum_{k=-K}^{K} 2w_k.
\end{equation}
Then
\begin{equation}\label{eq:care-sign-sum}
S_{K,h}(H) = \sum_{k=-K}^{K} w_k
\paren{
(1+t_k)(H-it_k\I)^{-1}
+
(1+t_k)(H+it_k\I)^{-1}
}.
\end{equation}

\begin{lemma}[Plain coefficient growth for the Hamiltonian sign sum]\label{lem:care-lambda}
For every $K\ge 1$ and $h>0$,
\begin{equation}\label{eq:care-lambda-bound}
\Lambda^{\mathrm{care}}_{K,h} = \frac{2h}{\pi}\left(K+\half\right).
\end{equation}
Under the balanced choice \eqref{eq:balanced-h}, this becomes
\begin{equation}\label{eq:care-lambda-balanced}
\Lambda^{\mathrm{care}}_{K,h} = \frac{2}{\pi}\left(\sqrt{2\pi\beta K}+\frac{h}{2}\right).
\end{equation}
\end{lemma}

\begin{proof}
Since $t_{-k}=e^{-kh}=1/t_k$,
\[
\frac{t_k}{1+t_k}+\frac{t_{-k}}{1+t_{-k}}=1
\qquad (k\ge 1),
\]
and $t_0=1$, so $t_0/(1+t_0)=1/2$.
Therefore
\[
\sum_{k=-K}^{K}\frac{t_k}{1+t_k}
=
\frac12 + \sum_{k=1}^{K}1
=
K+\half,
\]
which proves \eqref{eq:care-lambda-bound}.
Equation~\eqref{eq:care-lambda-balanced} is just $Kh=\sqrt{2\pi\beta K}$.
\end{proof}

\begin{definition}[Hamiltonian nodewise profile]\label{def:care-nodewise}
For each node $t_k=e^{kh}$, define the exact scaled Hamiltonian inverse norms 
\[
m^\pm_{H,k}:=(1+t_k)\norm{(H\pm it_k\I)^{-1}}.
\]
A \emph{Hamiltonian nodewise profile} is a collection of numbers $\rho^\pm_{H,k}>0$ such that
\[
m^\pm_{H,k}\le \rho^\pm_{H,k}
\qquad
\text{for all }k.
\]
Associated to such a profile, define
\[
R_H:=\max_{k,\pm}\rho^\pm_{H,k},
\qquad
d^\pm_{H,k}:=\frac{R_H}{\rho^\pm_{H,k}},
\qquad
\Theta_{\rho_H}^{\mathrm{care}}
:=
\sum_{k=-K}^{K}w_k\bigl(\rho^-_{H,k}+\rho^+_{H,k}\bigr).
\]
\end{definition}

\begin{remark}
The plain choice $\rho^\pm_{H,k}\equiv r_H$ gives the default Hamiltonian sign bound.  Nonuniform profiles give smaller bounds when nodewise information is available.
\end{remark}

\begin{theorem}[Hamiltonian sign approximation based on the log-sinc rule]\label{thm:care-sign}
Assume \cref{ass:care-be}, and let $a\in(0,1)$, $\beta\in(0,\arcsin a)$, and $\gamma_H\ge \gamma_{\Omega_a}(H)$.
Choose $K\ge 1$, $h>0$, a Hamiltonian nodewise profile $\rho_H$ in the sense of \cref{def:care-nodewise}, and an inverse precision $\eps_H\in(0,1]$.
Then there exists a
\begin{equation}\label{eq:care-sign-be}
\paren{\beta_{\mathrm{sign}},\, a_H+O(\log K),\, \eps_{\mathrm{sgn}}^{\mathrm{sinc}}(K,h;\beta,a,\gamma_H)+\Theta_{\rho_H}^{\mathrm{care}}\frac{\eps_H}{R_H}}
\text{-block-encoding of }\signf(H),
\end{equation}
where
\begin{equation}\label{eq:care-sign-beta}
\beta_{\mathrm{sign}} = 2\Theta_{\rho_H}^{\mathrm{care}}.
\end{equation}
The corresponding query complexity is
\begin{equation}\label{eq:care-sign-queries}
O\!\paren{R_H\log\frac{R_H}{\eps_H}}
\end{equation}
queries to $U_H$ and $U_H^\dagger$.
If one chooses the plain profile $\rho^\pm_{H,k}\equiv r_H$, then
\[
\Theta_{\rho_H}^{\mathrm{care}}=r_H\Lambda^{\mathrm{care}}_{K,h},
\]
so \eqref{eq:care-sign-be} reduces to the original plain normalization $2r_H\Lambda^{\mathrm{care}}_{K,h}$.
\end{theorem}

\begin{proof}
By \cref{prop:sign-approx-explicit}, the deterministic quadrature error is at most $\eps_{\mathrm{sgn}}^{\mathrm{sinc}}(K,h;\beta,a,\gamma_H)$.
Apply \cref{thm:single-family-algorithm} to the index set
\[
J=\{+,-\}\times\{-K,-K+1,\ldots,K\},
\]
with coefficients $c_{(\sigma,k)}=w_k$, matrices $F_{(\sigma,k)}=H+\sigma it_k\I$, and scales $\sigma_{(\sigma,k)}=1+t_k$.
The multiplexed operator is exactly $K_H$, so \cref{lem:care-mux} provides the required direct-sum block-encoding.
The resulting single-family target is the finite sign sum \eqref{eq:care-sign-sum}. Hence \cref{thm:single-family-algorithm} yields a block-encoding of $S_{K,h}(H)$ with normalization $2\Theta_{\rho_H}^{\mathrm{care}}$ and implementation error $\Theta_{\rho_H}^{\mathrm{care}}\eps_H/R_H$.
Adding the deterministic quadrature error proves \eqref{eq:care-sign-be}.
\end{proof}

\begin{corollary}[Diagonalizable Hamiltonians: explicit strip-based quadrature control]\label{cor:care-diag}
Assume \cref{ass:care-be}, and suppose that $H=V\Lambda V^{-1}$ is diagonalizable with
\begin{equation}\label{eq:care-diag-gap}
\min_j \abs{\RePart \lambda_j} \ge \delta > 0.
\end{equation}
Set
\[
a_H^\star := \frac{\delta}{2}\in(0,1).
\]
Then
\begin{equation}\label{eq:care-gamma-diag}
\gamma_H \le \frac{2\kappa(V)}{\delta},
\qquad
r_H \le \frac{6\kappa(V)}{\delta}.
\end{equation}
For any $\beta\in(0,\arcsin a_H^\star)$, equation~\eqref{eq:Kchoice-sinc} gives the deterministic quadrature size $K$ with
\[
a=a_H^\star,
\qquad
\gamma = \frac{2\kappa(V)}{\delta}.
\]
Moreover, the plain profile $\rho^\pm_{H,k}\equiv r_H$ is always admissible, so \cref{thm:care-sign} recovers the explicit normalization
\[
\beta_{\mathrm{sign}} \le 2r_H\Lambda^{\mathrm{care}}_{K,h}
\]
with $r_H$ bounded by \eqref{eq:care-gamma-diag}.
No further annulus-bound estimate is needed.
\end{corollary}

\begin{proof}
The spectrum of $H$ lies outside the strip $\Omega_{a_H^\star}$.
For $z\in\Omega_{a_H^\star}$,
\[
(z\I-H)^{-1}
=
V(z\I-\Lambda)^{-1}V^{-1},
\]
so
\[
\norm{(z\I-H)^{-1}}
\le
\kappa(V)\max_j\frac{1}{\abs{z-\lambda_j}}
\le
\frac{\kappa(V)}{\delta-a_H^\star}
=
\frac{2\kappa(V)}{\delta},
\]
which proves \eqref{eq:care-gamma-diag} for $\gamma_H$.
The bound on $r_H$ then follows from \cref{lem:care-mux}.
\end{proof}

\begin{remark}[Why the CARE sign stage has a logarithmic LCU weight factor]
\label{rem:care-logweight}
In the Hamiltonian sign stage, the overlap parameter is
\[
\Theta_{\rho_H}^{\mathrm{care}}
:=
\sum_{k=-K}^{K}w_k\bigl(\rho^-_{H,k}+\rho^+_{H,k}\bigr),
\qquad
\beta_{\mathrm{sign}}=2\Theta_{\rho_H}^{\mathrm{care}}.
\]
Thus the plain profile \(\rho^\pm_{H,k}\equiv r_H\) gives
\[
\Theta_{\rho_H}^{\mathrm{care}}
=
r_H\Lambda^{\mathrm{care}}_{K,h},
\qquad
\beta_{\mathrm{sign}}
=
2r_H\Lambda^{\mathrm{care}}_{K,h}.
\]

The reason a logarithmic coefficient growth remains in this CARE sign stage is that
\cref{thm:care-sign} implements the full Hamiltonian sign by a single scaled inverse family
on each branch,
\[
(1+t_k)(H\pm it_k\I)^{-1},
\]
with coefficients
\[
w_k=\frac{h t_k}{\pi(1+t_k)}.
\]
For large positive logarithmic nodes \(t_k\), one has \(w_k\sim h/\pi\), so the positive
tail contributes \(O(Kh)\); its coefficient sum is unbounded as $K$ grows.

The algebraic identity
\[
(H-it\I)^{-1}+(H+it\I)^{-1}
=
2(H-it\I)^{-1}H(H+it\I)^{-1}
\]
gives a bounded coefficient sum after scaling both inverse factors by \(1+t\).
Its two-inverse product has sign-stage normalization of order \(R_H^2\).
The single-family implementation uses one inverse family and has normalization linear in
\(R_H\), together with a logarithmic weight factor.

Indeed, by \cref{lem:care-lambda}, under the balanced rule
\[
\Lambda^{\mathrm{care}}_{K,h}
=
\frac{2}{\pi}\left(Kh+\frac{h}{2}\right).
\]
Together with the log-sinc choice of \(K\), this gives
\[
\Lambda^{\mathrm{care}}_{K,h}
=
O\!\left(
\log\frac{\widetilde C_{\beta,a,\gamma_H}}{\eps_{\mathrm{sign}}}
\right)
\]
for fixed \(a\) and \(\beta\). The sign-stage query bound in \cref{thm:care-main}
therefore contains this logarithmic factor.
\end{remark}

\subsection{Projector inversion and the final CARE block-encoding}
Let
\[
\Pi_- = \frac{\I-\signf(H)}{2}
=
\begin{bmatrix}
\Pi_{11} & \Pi_{12} \\
\Pi_{21} & \Pi_{22}
\end{bmatrix},
\qquad
p_H := \norm{\Pi_-},
\qquad
\sigma := \sigma_{\min}(\Pi_{11}) > 0.
\]
Let $\widetilde S$ be the encoded matrix produced by \cref{thm:care-sign}, so that
\[
\norm{\widetilde S-\signf(H)} \le \eps_{\mathrm{sign}}.
\]
Define
\begin{equation}\label{eq:care-projector-approx}
\widetilde \Pi_- := \frac{\I-\widetilde S}{2}
=
\begin{bmatrix}
\widetilde \Pi_{11} & \widetilde \Pi_{12} \\
\widetilde \Pi_{21} & \widetilde \Pi_{22}
\end{bmatrix}.
\end{equation}
Then
\[
\norm{\widetilde \Pi_- - \Pi_-} \le \frac{\eps_{\mathrm{sign}}}{2}.
\]

\begin{lemma}[Invertibility of the perturbed leading block]\label{lem:care-pi11}
If $\eps_{\mathrm{sign}}<2\sigma$, then $\widetilde \Pi_{11}$ is invertible and
\begin{equation}\label{eq:care-pi11-inv}
\norm{\widetilde \Pi_{11}^{-1}} \le \frac{1}{\sigma-\eps_{\mathrm{sign}}/2}.
\end{equation}
\end{lemma}

\begin{proof}
Because
\[
\norm{\widetilde \Pi_{11}-\Pi_{11}} \le \frac{\eps_{\mathrm{sign}}}{2},
\]
Weyl's inequality gives
\[
\sigma_{\min}(\widetilde \Pi_{11})
\ge
\sigma_{\min}(\Pi_{11}) - \norm{\widetilde \Pi_{11}-\Pi_{11}}
\ge
\sigma-\frac{\eps_{\mathrm{sign}}}{2} > 0.
\]
Hence $\widetilde \Pi_{11}$ is invertible and \eqref{eq:care-pi11-inv} follows.
\end{proof}

\begin{proof}[Proof of \cref{thm:care-main}]
The projector block-encoding $\widetilde \Pi_-$ in \eqref{eq:care-projector-approx} has normalization $\beta_\Pi$, because it is obtained from the $(\beta_{\mathrm{sign}},\cdot,\eps_{\mathrm{sign}})$-block-encoding of $\widetilde S$ and the exact identity block-encoding by a two-term LCU.

By \cref{lem:care-pi11}, $\widetilde \Pi_{11}$ is invertible and
\[
\norm{\widetilde \Pi_{11}^{-1}} \le \frac{1}{\sigma-\eps_{\mathrm{sign}}/2}.
\]
The same unitary that block-encodes $\widetilde \Pi_{11}$ with normalization $\beta_\Pi$ is a unit block-encoding of
\[
T_\Pi:=\frac{\widetilde \Pi_{11}}{\beta_\Pi}.
\]
Apply \cref{prop:qsvt-inverse} to $T_\Pi$ with target error $\min\{1,\beta_\Pi\eps_{\Pi^{-1}}\}$.
Because
\[
T_\Pi^{-1}=\beta_\Pi \widetilde \Pi_{11}^{-1},
\qquad
\norm{T_\Pi^{-1}}=\beta_\Pi\norm{\widetilde \Pi_{11}^{-1}},
\]
the resulting QSVT circuit is a
\[
\paren{\frac{2\beta_\Pi}{\sigma-\eps_{\mathrm{sign}}/2},\, a_H+O(\log K),\, \min\{1,\beta_\Pi\eps_{\Pi^{-1}}\}}
\text{-block-encoding of }\beta_\Pi \widetilde \Pi_{11}^{-1}.
\]
Reinterpreting the same circuit as a block-encoding of $\widetilde \Pi_{11}^{-1}$ divides both the normalization and the error by $\beta_\Pi$, so the resulting error is at most $\eps_{\Pi^{-1}}$, and we obtain a
\[
\paren{\frac{2}{\sigma-\eps_{\mathrm{sign}}/2},\, a_H+O(\log K),\, \eps_{\Pi^{-1}}}
\text{-block-encoding of }\widetilde \Pi_{11}^{-1},
\]
with query cost \eqref{eq:care-pi11-queries} in calls to the projector block-encoding.
Multiplying it by the $(\beta_\Pi,\cdot,\eps_{\mathrm{sign}}/2)$-block-encoding of $\widetilde \Pi_{21}$ gives \eqref{eq:care-final-beta} and the claimed ancilla count.

For the operator error, define
\[
\widehat X := \widetilde \Pi_{21}\widetilde \Pi_{11}^{-1},
\qquad
\widetilde X := \widetilde \Pi_{21}Y.
\]
Then
\[
\widetilde X - X = (\widetilde X-\widehat X) + (\widehat X-X).
\]
The exact identity
\[
\widehat X - X
=
\paren{(\widetilde \Pi_{21}-\Pi_{21}) - X(\widetilde \Pi_{11}-\Pi_{11})}\widetilde \Pi_{11}^{-1}
\]
implies
\[
\norm{\widehat X-X}
\le
\paren{\norm{\widetilde \Pi_{21}-\Pi_{21}} + \norm{X}\norm{\widetilde \Pi_{11}-\Pi_{11}}}
\norm{\widetilde \Pi_{11}^{-1}}
\le
\frac{\eps_{\mathrm{sign}}}{2}\,
\frac{1+\norm{X}}{\sigma-\eps_{\mathrm{sign}}/2}.
\]
For the inverse-implementation term,
\[
\norm{\widetilde X-\widehat X}
=
\norm{\widetilde \Pi_{21}(Y-\widetilde \Pi_{11}^{-1})}
\le
\norm{\widetilde \Pi_{21}}\eps_{\Pi^{-1}}
\le
\paren{p_H+\frac{\eps_{\mathrm{sign}}}{2}}\eps_{\Pi^{-1}}.
\]
Adding the two bounds proves \eqref{eq:care-final-error}.
If $H=V\Lambda V^{-1}$ is diagonalizable, then
\[
\Pi_- = V D_- V^{-1}
\]
for a diagonal projector $D_-$, so $p_H\le \kappa(V)$.
\end{proof}

\begin{proposition}[CARE residual bound]\label{prop:care-residual}
Let $X$ be the exact solution from \cref{thm:care-projector}, and let $\widetilde X$ be any approximation.
Then
\begin{equation}\label{eq:care-residual}
\norm{A^*\widetilde X + \widetilde X A - \widetilde X G\widetilde X + Q}
\le
\paren{2\norm{A} + 2\norm{G}\norm{X}}\norm{\widetilde X-X}
+
\norm{G}\norm{\widetilde X-X}^2.
\end{equation}
\end{proposition}

\begin{proof}
Set $\Delta:=\widetilde X-X$.
Since $X$ solves \eqref{eq:care},
\[
A^*\widetilde X + \widetilde X A - \widetilde X G\widetilde X + Q
=
A^*\Delta + \Delta A - XG\Delta - \Delta GX - \Delta G\Delta.
\]
Taking norms and using submultiplicativity proves \eqref{eq:care-residual}.
\end{proof}

\section{Discussion}\label{sec:discussion}
The algorithms studied here share three steps: a matrix-sign representation of the target, a rational approximation of the sign, and a shifted-inverse implementation with nodewise normalization bounds. Ordinary and generalized Sylvester equations use the bilinear inverse sums in \cref{thm:generic-overlap,thm:gen-main}. Square roots, geometric means, and the Hamiltonian sign stage for CARE use the single-family construction in \cref{thm:single-family-algorithm}; CARE then requires an additional inversion of the leading projector block.

The operator-output access model affects how the result can be used. A state-output algorithm prepares $\lvert X\rangle\rangle/\norm{X}_{\mathrm F}$. A $(\beta_X,a_X,\eps)$-block-encoding gives coherent access to $X/\beta_X$ as a linear map, which can be applied to selected input states and composed with other block-encoded operators. For an entry $X_{jk}$, the two representations yield scales $X_{jk}/\norm{X}_{\mathrm F}$ and $X_{jk}/\beta_X$, respectively. Somma \emph{et al.} prove a one-query versus $\Omega(\sqrt{N})$ oracle separation for a fixed-entry problem \cite{SommaLowBerryBabbush2025}. The normalization bounds derived above therefore determine the cost of subsequent uses of the encoded operator.

Contour-integral algorithms approximate a matrix function by a weighted sum of its resolvents. In the present algorithms, each target is represented first as a block of a half-plane sign or its spectral projector. For the Sylvester equation, the shifted resolvent factorization gives the constant-weight identity in \cref{lem:constant-weight} and the overlap sum $\Theta_{\rho}^{\mathrm{syl}}$. The root and geometric-mean embeddings give positive-shift inverse families. The CARE embedding gives a Hamiltonian projector followed by a separate inversion of its leading block, with sign-stage normalization $\Theta_{\rho_H}^{\mathrm{care}}$. Thus the normalization analysis follows from the shifted inverse families associated with each sign representation.

The quantitative estimates require FoV gaps, strip-resolvent bounds, projector invertibility parameters, or nodewise shifted-inverse profiles. This paper does not provide general procedures for estimating these quantities. When only uniform FoV or strip bounds are available, the plain-profile estimates apply but may give a larger normalization than the nodewise estimates. The geometric-mean results additionally assume a principal branch and a scaled strip-resolvent bound for the sign embedding. The CARE results require Hamiltonian strip control and a lower bound on the smallest singular value of the leading stable-projector block. Applying these results therefore requires problem-specific estimates of the relevant quantities.

Open problems include efficient estimation of strip-resolvent bounds and nodewise shifted-inverse profiles for non-normal pencils. Weighted-pencil analogues of the banded-overlap bound may improve the generalized Sylvester and Lyapunov normalizations. For geometric means, asymmetric positive-shift estimates could account more sharply for different accretivity margins. For CARE, Hamiltonian or symplectic structure may reduce the cost of the sign stage or improve the projector-block conditioning estimate. Adaptive rational approximations and problem-dependent preconditioners may also reduce the shifted-inverse costs.

Across these problems, the sign representation reduces the target to one or two families of shifted inverses. Their conditioning determines the query cost, and the corresponding overlap sums determine the block-encoding normalization.

\section*{Acknowledgments}
JPL acknowledges support from Quantum Science and Technology--National Science and Technology Major Project under Grant No.~2024ZD0300500, Excellent Young Scientists Fund Program, start-up funding from Tsinghua University and Beijing Institute of Mathematical Sciences and Applications.

YQW acknowledges the assistance of AIM~\cite{AIM2025}, an AI research agent. The authors independently checked all AI-assisted material and take responsibility for the manuscript.

\bibliographystyle{unsrt}
\bibliography{references}

\appendix
\section{Minimal block-encoding calculus used in the proofs}\label{app:be-calculus}
This appendix collects the block-encoding rules used in the proofs.

\subsection{Basic rules}

\begin{lemma}[Product rule]\label{lem:product-rule}
If $U_X$ is an $(\nu_X,a_X,0)$-block-encoding of $X$ and $U_Y$ is an $(\nu_Y,a_Y,0)$-block-encoding of $Y$, then $(U_X\otimes \I)(U_Y\otimes \I)$ can be arranged as a $(\nu_X\nu_Y, a_X+a_Y,0)$-block-encoding of $XY$.
\end{lemma}

\begin{lemma}[Two-term LCU rule]\label{lem:sum-rule}
Let $U_X$ be a $(\nu_X,a_X,0)$-block-encoding of $X$ and $U_Y$ a  $(\nu_Y,a_Y,0)$-block-encoding of $Y$.
If
\[
\lambda \ge \abs{c_X}\nu_X + \abs{c_Y}\nu_Y,
\]
then a standard two-term LCU construction yields a $(\lambda, a_X+a_Y+O(1),0)$-block-encoding of $c_X X + c_Y Y$.
\end{lemma}

\begin{lemma}[Error for the product of approximate operators]\label{lem:approx-product}
If $\widehat X$ and $\widehat Y$ satisfy
\[
\norm{\widehat X-X}\le \delta_X,
\qquad
\norm{\widehat Y-Y}\le \delta_Y,
\]
then
\[
\norm{\widehat X\widehat Y - XY} \le \delta_X\norm{Y} + \norm{X}\delta_Y + \delta_X\delta_Y.
\]
\end{lemma}

\begin{proof}
Expand
\[
\widehat X\widehat Y - XY = (\widehat X-X)Y + X(\widehat Y-Y) + (\widehat X-X)(\widehat Y-Y)
\]
and use submultiplicativity.
\end{proof}

\subsection{Standard multiplexed gadgets}

\begin{lemma}[Diagonal contractions on the index register]\label{lem:diag-contraction}
Let $J$ be a finite index set, and let
\[
D=\sum_{j\in J}\delta_j\ket{j}\!\bra{j}
\]
with $0\le \delta_j\le 1$ for all $j$.
Then $D$ admits an exact $(1,O(1),0)$-block-encoding.
\end{lemma}

\begin{proof}
Prepare one ancilla qubit and, conditioned on the index state $\ket{j}$, apply the one-qubit rotation that maps
\[
\ket{0}\longmapsto \delta_j\ket{0}+\sqrt{1-\delta_j^2}\ket{1}.
\]
The top-left block of the resulting controlled unitary is exactly $D$.
\end{proof}

\begin{lemma}[Multiplexed one-matrix shift gadget]\label{lem:mux-one-matrix}
Let $U_T\in \BE(1,a_T,0)$ be a block-encoding of $T$.
For each $j$ in a finite index set $J$, let $\alpha_j,\beta_j\in\C$ satisfy
\[
\abs{\alpha_j}+\abs{\beta_j}\le 1.
\]
Then the direct sum
\[
M:=\sum_{j\in J}\ket{j}\!\bra{j}\otimes\paren{\alpha_j T+\beta_j \I}
\]
admits a $(1,a_T+O(1),0)$-block-encoding using $O(1)$ queries to $U_T$.
\end{lemma}

\begin{proof}
For each fixed $j$, the two-term LCU rule yields a unit-normalized block-encoding of $\alpha_j T+\beta_j\I$ because the coefficient sum is at most one.
Controlling this construction on the index register yields the stated direct-sum block-encoding.
\end{proof}

\begin{lemma}[Multiplexed two-matrix shift gadget]\label{lem:mux-two-matrix}
Let $U_T\in \BE(1,a_T,0)$ and $U_S\in \BE(1,a_S,0)$ be block-encodings of $T$ and $S$.
For each $j$ in a finite index set $J$, let $\alpha_j,\beta_j\in\C$ satisfy
\[
\abs{\alpha_j}+\abs{\beta_j}\le 1.
\]
Then the direct sum
\[
M:=\sum_{j\in J}\ket{j}\!\bra{j}\otimes\paren{\alpha_j T+\beta_j S}
\]
admits a $(1,a_T+a_S+O(1),0)$-block-encoding using $O(1)$ queries each to $U_T$ and $U_S$.
\end{lemma}

\begin{proof}
For each fixed $j$, the two-term LCU rule gives a unit-normalized block-encoding of $\alpha_j T+\beta_j S$ because the coefficient sum is at most one.
Controlling that block-encoding on the index register yields the desired multiplexed direct sum.
\end{proof}

\section{Interpretation of the FoV gap in transformed coordinates}\label{app:mu-conditioning}
In the main text, the FoV gap $\mu$ refers to the \emph{transformed} pair $(A,B)$ obtained by an optimal shift-rotation and a common positive scaling with $\norm{M}\le 1$.  Its relevant comparison is therefore the conditioning of the Sylvester operator in the transformed coordinates.  We prove that $1/\mu$ bounds the Euclidean condition number of the vectorized Sylvester operator.  The bound can be loose when the off-diagonal block $C$ controls the final normalization or when the FoV gap is smaller than the singular gap of the vectorized operator.

\subsection{Transformed coordinates and invariances}
Let
\[
K_{A,B}:=\I_m\otimes A + B^\top\otimes \I_n
\]
denote the vectorized Sylvester operator, so that
\[
\operatorname{vec}(AX+XB)=K_{A,B}\operatorname{vec}(X).
\]
For the original problem data $(A_0,B_0,C_0)$, define
\[
K_0:=\I_m\otimes A_0 + B_0^\top\otimes \I_n,
\qquad
\delta:=\dist\paren{W(A_0),-W(B_0)}.
\]
Let $(\eta,\omega)$ be the optimal shift-rotation from \cref{thm:fov-shift}, and define
\[
\widehat A:=\eta(A_0-\omega\I),
\qquad
\widehat B:=\eta(B_0+\omega\I),
\qquad
\widehat C:=\eta C_0,
\]
\[
\widehat M:=
\begin{bmatrix}
\widehat A & \widehat C\\
0 & -\widehat B
\end{bmatrix},
\qquad
\widehat K:=\I_m\otimes \widehat A + \widehat B^\top\otimes \I_n.
\]
The maximal common scaling under the constraint $\norm{M}\le 1$ is
\begin{equation}\label{eq:mu-maximal-scaling}
\lambda:=\frac{1}{\norm{\widehat M}}.
\end{equation}
Set
\[
A:=\lambda\widehat A,
\qquad
B:=\lambda\widehat B,
\qquad
C:=\lambda\widehat C,
\qquad
M:=\lambda\widehat M.
\]
Then $\norm{M}=1$ and, by \cref{thm:fov-shift},
\begin{equation}\label{eq:mu-delta-Mhat}
H(A)\succeq \mu\I,
\qquad
H(B)\succeq \mu\I,
\qquad
\mu=\frac{\lambda\delta}{2}=\frac{\delta}{2\norm{\widehat M}}.
\end{equation}

\begin{proposition}[Invariance of the vectorized Sylvester operator under preprocessing]\label{prop:mu-kappa-transformed}
With the notation above,
\begin{equation}\label{eq:Khat-K0}
\widehat K=\eta K_0,
\qquad
K_{A,B}=\lambda\eta K_0.
\end{equation}
Consequently,
\begin{equation}\label{eq:kappa-invariant}
\kappa_2(K_{A,B})=\kappa_2(\widehat K)=\kappa_2(K_0),
\end{equation}
and
\begin{equation}\label{eq:Knorm-scaled}
\norm{K_{A,B}}_2=\lambda\norm{K_0}_2=\frac{\norm{K_0}_2}{\norm{\widehat M}}.
\end{equation}
Thus $1/\mu$ is the explicit size-over-gap ratio
\begin{equation}\label{eq:mu-size-gap}
\frac{1}{\mu}=\frac{2\norm{\widehat M}}{\delta}
\end{equation}
in the optimally shift-rotated coordinates. The preprocessing leaves $\kappa_2(K_0)$ invariant.
\end{proposition}

\begin{proof}
Using the cancellation of the shift in the Sylvester operator,
\begin{align*}
\widehat K
&=
\I_m\otimes \eta(A_0-\omega\I)+\eta(B_0+\omega\I)^\top\otimes \I_n\\
&=
\eta\paren{\I_m\otimes A_0 + B_0^\top\otimes \I_n}
=
\eta K_0.
\end{align*}
Multiplying by the common scaling $\lambda$ gives $K_{A,B}=\lambda\widehat K=\lambda\eta K_0$, which proves \eqref{eq:Khat-K0}. Since multiplication by a nonzero scalar does not change the $2$-norm condition number, \eqref{eq:kappa-invariant} follows immediately. Equation~\eqref{eq:Knorm-scaled} is just the norm identity coming from $K_{A,B}=\lambda\eta K_0$. The formula \eqref{eq:mu-size-gap} is exactly \eqref{eq:mu-delta-Mhat} with the maximal scaling \eqref{eq:mu-maximal-scaling}.
\end{proof}

\subsection{Relation with $\sepop(A,-B)$}

\begin{proposition}[FoV gap versus Sylvester separation]\label{prop:mu-sep}
Assume
\[
H(A)\succeq \mu \I,
\qquad
H(B)\succeq \mu \I
\]
for some $\mu>0$.
Then the Sylvester solution operator satisfies
\begin{equation}\label{eq:mu-sep-lower}
\norm{S_{A,B}^{-1}} \le \frac{1}{2\mu},
\qquad
\sepop(A,-B)\ge 2\mu.
\end{equation}
If, in addition, $\norm{A}\le 1$ and $\norm{B}\le 1$, then
\begin{equation}\label{eq:mu-sep-upper}
\sepop(A,-B)\le \norm{A}+\norm{B}\le 2.
\end{equation}
Hence, in the normalized FoV regime used in the main text,
\begin{equation}\label{eq:mu-sep-window}
2\mu \le \sepop(A,-B)\le 2.
\end{equation}
Moreover, in the transformed-coordinate setup above one has
\begin{equation}\label{eq:sep-original-delta}
\sepop(A_0,-B_0)\ge \delta.
\end{equation}
\end{proposition}

\begin{proof}
Fix $C$ and define
\[
X := \int_0^\infty e^{-tA} C e^{-tB}\,dt.
\]
Because $H(A)\succeq \mu\I$ and $H(B)\succeq \mu\I$, one has
\[
\norm{e^{-tA}}\le e^{-\mu t},
\qquad
\norm{e^{-tB}}\le e^{-\mu t}
\qquad (t\ge 0),
\]
so the integral converges absolutely in operator norm and
\[
\norm{X}
\le
\int_0^\infty \norm{e^{-tA}}\,\norm{C}\,\norm{e^{-tB}}\,dt
\le
\norm{C}\int_0^\infty e^{-2\mu t}\,dt
=
\frac{\norm{C}}{2\mu}.
\]
Now set
\[
F(t):=e^{-tA}Ce^{-tB}.
\]
Then
\[
F'(t)=-Ae^{-tA}Ce^{-tB}-e^{-tA}Ce^{-tB}B,
\]
and therefore
\[
AX+XB
=
-\int_0^\infty F'(t)\,dt
=
F(0)-\lim_{T\to\infty}F(T)
=
C,
\]
because $\norm{F(T)}\le e^{-2\mu T}\norm{C}\to 0$.
Thus $X=S_{A,B}^{-1}(C)$ and
\[
\norm{S_{A,B}^{-1}(C)}\le \frac{\norm{C}}{2\mu}
\qquad
\text{for every }C,
\]
which proves $\norm{S_{A,B}^{-1}}\le 1/(2\mu)$ and hence \eqref{eq:mu-sep-lower}.

If $\norm{X}=1$, then
\[
\norm{AX+XB}\le \norm{A}\norm{X}+\norm{X}\norm{B}=\norm{A}+\norm{B}.
\]
Taking the infimum over all such $X$ gives \eqref{eq:mu-sep-upper}. This proves \eqref{eq:mu-sep-window}.

In the transformed-coordinate setup, the Sylvester operator itself transforms as
\[
S_{A,B}=\lambda\eta S_{A_0,B_0},
\]
so $\sepop(A,-B)=\lambda\sepop(A_0,-B_0)$. Combining this with \eqref{eq:mu-sep-lower} and \eqref{eq:mu-delta-Mhat} gives
\[
\lambda\sepop(A_0,-B_0)=\sepop(A,-B)\ge 2\mu=\lambda\delta,
\]
which is exactly \eqref{eq:sep-original-delta}.
\end{proof}

\begin{remark}
Equation~\eqref{eq:sep-original-delta} shows that the FoV gap $\delta$ is a lower bound on the classical Sylvester separation of the original problem.  The parameter $\mu$ is the corresponding value after maximal admissible scaling.
\end{remark}

\subsection{Relation with the vectorized Kronecker operator}

Let $A\in\C^{n\times n}$ and $B\in\C^{m\times m}$, and define
\begin{equation}\label{eq:kronecker-sylvester}
K_{A,B}:=\I_m\otimes A + B^\top\otimes \I_n.
\end{equation}
Then
\[
\operatorname{vec}(AX+XB)=K_{A,B}\operatorname{vec}(X).
\]
The alternative vectorization convention gives $A\otimes \I_m+\I_n\otimes B^\top$.  Transposition and a perfect-shuffle permutation preserve the singular values and the $2$-norm condition number.

\begin{proposition}[Universal FoV bound on Kronecker-sum conditioning]\label{prop:mu-kappa-kronecker}
Assume
\[
H(A)\succeq \mu \I,
\qquad
H(B)\succeq \mu \I
\]
for some $\mu>0$, and define
\[
\kappa_2(K_{A,B}):=\norm{K_{A,B}}_2\,\norm{K_{A,B}^{-1}}_2.
\]
Then
\begin{equation}\label{eq:kron-sigmamin}
\sigma_{\min}(K_{A,B})\ge 2\mu.
\end{equation}
Consequently,
\begin{equation}\label{eq:kron-kappa-bound}
\kappa_2(K_{A,B})\le \frac{\norm{K_{A,B}}_2}{2\mu}\le \frac{\norm{A}+\norm{B}}{2\mu}.
\end{equation}
In particular, under the normalized coordinates $\norm{A},\norm{B}\le 1$,
\begin{equation}\label{eq:kron-kappa-normalized}
\kappa_2(K_{A,B})\le \frac{1}{\mu}.
\end{equation}
Hence, in the transformed-coordinate setup of \cref{prop:mu-kappa-transformed},
\begin{equation}\label{eq:kappa-original-upper}
\kappa_2(K_0)=\kappa_2(K_{A,B})\le \frac{1}{\mu}.
\end{equation}
\end{proposition}

\begin{proof}
The Hermitian part of $K_{A,B}$ is
\[
H(K_{A,B})
=
\I_m\otimes H(A)+H(B)^\top\otimes \I_n
\succeq
2\mu \I_{mn}.
\]
Hence for every unit vector $u\in\C^{mn}$,
\[
2\mu
\le
u^*H(K_{A,B})u
=
\RePart\!\bigl(u^*K_{A,B}u\bigr)
\le
\abs{u^*K_{A,B}u}
\le
\norm{K_{A,B}u}_2.
\]
Taking the infimum over all unit vectors proves \eqref{eq:kron-sigmamin}.
Moreover,
\[
\norm{K_{A,B}}_2
\le
\norm{\I_m\otimes A}_2+\norm{B^\top\otimes \I_n}_2
=
\norm{A}_2+\norm{B}_2,
\]
so
\[
\kappa_2(K_{A,B})
=
\frac{\norm{K_{A,B}}_2}{\sigma_{\min}(K_{A,B})}
\le
\frac{\norm{A}+\norm{B}}{2\mu},
\]
which proves \eqref{eq:kron-kappa-bound}. The normalized bound \eqref{eq:kron-kappa-normalized} is immediate, and \eqref{eq:kappa-original-upper} then follows from \eqref{eq:kappa-invariant}.
\end{proof}

\begin{proposition}[A criterion for $\kappa_2(K)$ to be of order $1/\mu$]\label{prop:mu-kappa-same-order}
Adopt the transformed-coordinate setup of \cref{prop:mu-kappa-transformed}. Suppose the optimally shift-rotated pair additionally satisfies
\begin{equation}\label{eq:mu-kappa-compare-ass}
\sigma_{\min}(\widehat K)\le c_{\delta}\,\delta,
\qquad
\norm{\widehat K}_2\ge c_M\,\norm{\widehat M}
\end{equation}
for some constants $c_{\delta},c_M>0$.
Then
\begin{equation}\label{eq:mu-kappa-same-order}
\frac{c_M}{2c_{\delta}}\frac{1}{\mu}
\le
\kappa_2(K_{A,B})
=
\kappa_2(K_0)
\le
\frac{1}{\mu}.
\end{equation}
Hence $\kappa_2(K_{A,B})=\Theta(1/\mu)$ whenever \eqref{eq:mu-kappa-compare-ass} holds with absolute constants.
\end{proposition}

\begin{proof}
The upper bound in \eqref{eq:mu-kappa-same-order} is exactly \eqref{eq:kappa-original-upper}. For the lower bound, \eqref{eq:kappa-invariant} gives
\[
\kappa_2(K_{A,B})=\kappa_2(\widehat K)=\frac{\norm{\widehat K}_2}{\sigma_{\min}(\widehat K)}.
\]
Using \eqref{eq:mu-kappa-compare-ass} and then \eqref{eq:mu-size-gap},
\[
\kappa_2(K_{A,B})
\ge
\frac{c_M\norm{\widehat M}}{c_{\delta}\delta}
=
\frac{c_M}{2c_{\delta}}\frac{1}{\mu},
\]
which proves the claim.
\end{proof}

\begin{remark}\label{rem:mu-kappa-interpretation}
The gap between $1/\mu$ and $\kappa_2(K)$ in \cref{prop:mu-kappa-same-order} has two possible sources.  If the off-diagonal block $C$ controls the normalization $\norm{M}=1$, then $\norm{\widehat M}\gg \norm{\widehat K}$ and $\mu$ is small relative to the scale of the Sylvester operator.  A loose FoV lower bound, $\delta\ll\sigma_{\min}(\widehat K)$, has the same effect.  In the absence of these effects, $1/\mu$ and $\kappa_2(K)$ have the same order.

\end{remark}

\begin{corollary}[Hermitian benchmark after the optimal shift]\label{cor:mu-kappa-hermitian}
Assume $A_0$ and $B_0$ are Hermitian and
\[
\delta:=\dist\paren{W(A_0),-W(B_0)}>0.
\]
Apply the optimal shift from \cref{thm:fov-shift} and then the maximal common scaling \eqref{eq:mu-maximal-scaling}, producing transformed matrices $A$ and $B$ and FoV gap $\mu$.
Then $A$ and $B$ are Hermitian positive definite, and one can choose the optimal shift so that
\begin{equation}\label{eq:hermitian-min-eigs}
\lambda_{\min}(A)=\lambda_{\min}(B)=\mu.
\end{equation}
Consequently,
\begin{equation}\label{eq:hermitian-kron-sigmamin}
\sigma_{\min}(K_{A,B})=2\mu,
\qquad
\kappa_2(K_{A,B})=\frac{\norm{K_{A,B}}_2}{2\mu}.
\end{equation}
In particular, if $\norm{K_{A,B}}_2=\Theta(1)$ in normalized coordinates, then
\begin{equation}\label{eq:hermitian-kappa-theta}
\kappa_2(K_{A,B})=\Theta\!\paren{\frac{1}{\mu}}.
\end{equation}
\end{corollary}

\begin{proof}
In the Hermitian case the optimal shift-rotation can be chosen with $\eta=\pm 1$, so the transformed matrices remain Hermitian. Since $W(A_0)$ and $-W(B_0)$ are intervals on a line, the midpoint construction from \cref{thm:fov-shift} places the two closest endpoints symmetrically at distance $\delta/2$ from the origin. After the maximal common scaling this gives \eqref{eq:hermitian-min-eigs}.

Now $K_{A,B}=\I_m\otimes A + B^\top\otimes \I_n$ is Hermitian positive definite, and because $B^\top$ is unitarily similar to $B$, the eigenvalues of $K_{A,B}$ are exactly $\alpha_i+\beta_j$, where $\alpha_i\in\spec(A)$ and $\beta_j\in\spec(B)$. Therefore
\[
\sigma_{\min}(K_{A,B})
=
\lambda_{\min}(K_{A,B})
=
\lambda_{\min}(A)+\lambda_{\min}(B)
=
2\mu,
\]
which proves \eqref{eq:hermitian-kron-sigmamin}. The formula for $\kappa_2(K_{A,B})$ is immediate, and \eqref{eq:hermitian-kappa-theta} follows whenever $\norm{K_{A,B}}_2=\Theta(1)$.
\end{proof}

\begin{remark}[Relevance to the main theorems]
The quantity $1/\mu$ is an explicit upper bound on the invariant Kronecker condition number.  Under the conditions of \cref{prop:mu-kappa-same-order}, it also has the same order as that condition number.  For Hermitian inputs, the $O(1/\mu)$ regime in \cref{thm:R1-banded} therefore matches the vectorized conditioning scale up to constants whenever the normalized Sylvester operator has $O(1)$ norm.
\end{remark}

\section{Comparison tables}\label{app:comparison}

This appendix compares representative quantum algorithms for matrix equations and matrix
functions.  Each table uses the same five columns:
\emph{work}, \emph{input/access model}, \emph{assumptions}, \emph{output}, and
\emph{query complexity}.  The notation ``BE'' means block-encoding; in the output column,
\((\alpha,\eps)\)-BE suppresses ancilla counts.  The notation \(\widetilde O(\cdot)\)
suppresses polylogarithmic factors in the precision, dimension, and stated conditioning
parameters.

The present paper obtains deterministic operator-output block-encodings from matrix-sign
embeddings.  The other methods listed below use vectorized linear systems, generic
contour-integral formulas, mixed-state output, or positive-definite geometric-mean reductions.

\newpage

\paragraph{Sylvester-type equations.}
The closest dedicated quantum Sylvester baseline is Somma--Low--Berry--Babbush
\cite{SommaLowBerryBabbush2025}.  Their algorithm uses the vectorized Kronecker operator.
Our result extracts the matrix solution from a sign block and applies this construction to
generalized Sylvester and Lyapunov equations.

\begin{table}[H]
\centering
\caption{Quantum algorithms for ordinary and generalized Sylvester equations}
\label{tab:sylvester-compare}
\scriptsize
\setlength{\tabcolsep}{2.2pt}
\renewcommand{\arraystretch}{1.22}
\begin{tabularx}{\textwidth}{@{}
>{\RaggedRight\arraybackslash}p{1.75cm}
>{\RaggedRight\arraybackslash}p{2.55cm}
>{\RaggedRight\arraybackslash}p{3.25cm}
>{\RaggedRight\arraybackslash}p{3.05cm}
>{\RaggedRight\arraybackslash}X
@{}}
\toprule
\textbf{Work} &
\textbf{Input/access model} &
\textbf{Assumptions} &
\textbf{Output} &
\textbf{Query complexity} \\
\midrule
Somma \emph{et al.}~\cite{SommaLowBerryBabbush2025}
&
BEs of \(A,B,C\); access sufficient to implement the vectorized operator
\(Q_{\rm Syl}=A\otimes I+I\otimes B^{\top}\).
&
Ordinary Sylvester equation \(AX+XB=C\).  Representative theorem assumes
diagonalizability of the relevant matrices and condition number
\(\kappa=\|Q_{\rm Syl}^{-1}\|\).
&
\((x,\eps)\)-BE of \(X\), with \(x\ge \kappa\alpha\) in the main block-encoding
formulation.
&
Representative regimes:
\[
O\!\left(\kappa\log\frac{\kappa N}{\epsilon}\right)
\]
for normal $Q$;
\[
O\!\left(\kappa\,\operatorname{polylog}\frac{\kappa\alpha}{\epsilon}\right)
\]
for positive Hermitian part;
\[
O\!\left(\sqrt{\kappa}\log\frac{\kappa N}{\epsilon}\right)
\]
in the positive-matrix regime with additional square-root access.
\\
\addlinespace

This paper
&
Unit BEs of transformed \(A,B,C\).  Generalized case additionally uses BEs of
\(E,D\).
&
FoV gap
\[
H(A),H(B)\succeq \mu I
\]
or strip-resolvent bound \(\gamma\).  Generalized case assumes invertible
\(E,D\) and a corresponding weighted FoV or pencil-strip bound.
&
Ordinary:
\[
(O(\mu^{-2}),\eps)\text{-BE of }X,
\]
with banded refinement \(O(\mu^{-1}+\tau\mu^{-2})\).  Generalized:
\[
O(\|E^{-1}\|\|D^{-1}\|\mu^{-2})
\]
normalization.
&
Ordinary FoV:
\[
\widetilde O(\mu^{-1})
\]
queries to \(U_A,U_B\), and \(O(1)\) to \(U_C\).  Strip regime:
\[
\widetilde O(\gamma).
\]
Generalized FoV:
\[
\widetilde O(\|E^{-1}\|/\mu),\quad
\widetilde O(\|D^{-1}\|/\mu).
\]
\\
\bottomrule
\end{tabularx}
\end{table}

\newpage

\paragraph{Lyapunov equations.}
The Lyapunov literature uses several output models: pure states for vectorized solutions,
problem-specific block-encodings, and mixed states.  Our result obtains a deterministic
operator-output block-encoding by specializing the generalized Sylvester sign embedding.

\begin{table}[H]
\centering
\caption{Quantum algorithms for Lyapunov equations}
\label{tab:lyapunov-compare}
\scriptsize
\setlength{\tabcolsep}{2.2pt}
\renewcommand{\arraystretch}{1.22}
\begin{tabularx}{\textwidth}{@{}
>{\RaggedRight\arraybackslash}p{1.75cm}
>{\RaggedRight\arraybackslash}p{2.55cm}
>{\RaggedRight\arraybackslash}p{3.25cm}
>{\RaggedRight\arraybackslash}p{3.05cm}
>{\RaggedRight\arraybackslash}X
@{}}
\toprule
\textbf{Work} &
\textbf{Input/access model} &
\textbf{Assumptions} &
\textbf{Output} &
\textbf{Query complexity} \\
\midrule
Sun--Zhang~\cite{SunZhang2017}
&
Sparse oracle for the vectorized Lyapunov operator and state preparation for the
right-hand side.
&
Vectorized continuous-time Lyapunov equation.  Standard comparison setting uses
sparse, negative-definite \(A\).
&
Pure state
\[
|\widetilde x\rangle\approx |\operatorname{vec}(X)\rangle .
\]
&
HHL-type complexity
\[
\widetilde O\!\left(
\frac{\log(N)s^2\kappa^2}{\eps}
\right).
\]
\\
\addlinespace

Clayton \emph{et al.}~\cite{Clayton2024}
&
BEs of \(A\) and \(\Omega\).
&
\[
AX+XA^\top+\Omega=0,
\]
with \(A\) Hurwitz and \(\Omega\succ0\).
&
\((\widetilde O(\alpha^2\eta),\eps)\)-BE of \(X^\ast\).
&
One query to \(O_\Omega\) and
\[
\widetilde O\!\left(\alpha^2\sqrt{\frac{\eta}{\eps}}\right)
\]
queries to controlled \(O_A\) and its inverse.
\\
\addlinespace

Benedetti--Rosmanis--Rosenkranz~\cite{Benedetti2025}
&
Access to a block-encoded dynamical/Kraus operator and state preparation for the
PSD forcing term.
&
Continuous- and discrete-time Lyapunov settings, with normality/commutation
assumptions in the main protocols.
&
Mixed state close to
\[
X/\operatorname{tr}(X).
\]
&
Expected stopping-time bounds; continuous-time comparison scale
\[
\widetilde O\!\left(\frac{\kappa^3}{\eps}\right).
\]
\\
\addlinespace

This paper
&
BEs of \(A,E,Q\).
&
Generalized Lyapunov equation
\[
A^*XE+E^*XA=-Q.
\]
FoV specialization assumes \(E\succ0\) and
\[
H(E^{-1/2}AE^{-1/2})\succeq \mu I.
\]
&
\[
\left(O(\|E^{-1}\|^2\mu^{-2}),\eps\right)\text{-BE of }X.
\]
&
\[
\widetilde O(\|E^{-1}\|/\mu)
\]
queries to \(U_A,U_E\), and \(O(1)\) queries to \(U_Q\).
\\
\bottomrule
\end{tabularx}
\end{table}

\newpage

\paragraph{Square roots.}
Generic contour methods apply to broad holomorphic matrix functions.  The present square-root
algorithm uses the dedicated sign embedding
\[
K(A)=\begin{bmatrix}0&A\\ I&0\end{bmatrix},
\]
so both \(A^{-1/2}\) and \(A^{1/2}\) are obtained from a single positive-shift inverse family.

\begin{table}[H]
\centering
\caption{Quantum algorithms for inverse square roots and principal square roots}
\label{tab:sqrt-compare}
\scriptsize
\setlength{\tabcolsep}{2.2pt}
\renewcommand{\arraystretch}{1.22}
\begin{tabularx}{\textwidth}{@{}
>{\RaggedRight\arraybackslash}p{1.55cm}
>{\RaggedRight\arraybackslash}p{1.95cm}
>{\RaggedRight\arraybackslash}p{4.35cm}
>{\RaggedRight\arraybackslash}p{4.55cm}
>{\RaggedRight\arraybackslash}X
@{}}
\toprule
\textbf{Work} &
\textbf{Input/access model} &
\textbf{Assumptions} &
\textbf{Output} &
\textbf{Query complexity} \\
\midrule
Takahira \emph{et al.}~\cite{TakahiraOhashiSogabeUsuda2022}
&
BE of \(A\).
&
Holomorphic functional calculus:
\[
f(A)=\frac{1}{2\pi i}\int_\Gamma f(z)(zI-A)^{-1}\,dz.
\]
Requires shifted matrices to be invertible and have a known bound $\beta_0$
&
\[
(4r\beta_0\sum_k|w_k|,\eps)
\]
-BE of a quadrature approximation \(f_M(A)\).
&
\[
O\!\left(
(y_{\max}+z_{\max}\alpha)\beta_0
\log\frac{r\beta_0\mu_w}{\eps}
\right)
\]
uses of \(U_A\)
\\
\addlinespace

This paper
&
Unit BE of \(A\).
&
Principal square root exists.  Explicit FoV theorem assumes
\[
\|A\|\le1,\qquad H(A)\succeq \mu I.
\]
&
\[
(4/\sqrt{\mu},\eps)\text{-BE of }A^{-1/2},
\]
and the same normalization for \(A^{1/2}\).
&
FoV regime:
\[
O\!\left(
\frac1\mu\log\frac1{\eps\mu}
\right)
\]
queries to \(U_A,U_A^\dagger\).
Generic profile form:
\[
\widetilde O(R_{\mathrm{sq},\rho}).
\]
\\
\bottomrule
\end{tabularx}
\end{table}

\newpage

\paragraph{Geometric means.}
Liu \emph{et al.} treat Hermitian positive-definite geometric means via QSVT-based
block-encoding calculus.  The present paper uses a scaled sign embedding for
\[
G(A,B)=\begin{bmatrix}0&B\\ A^{-1}&0\end{bmatrix},
\]
which yields both \(A\#B\) and \((A\#B)^{-1}\) under a principal-branch condition and a
strip-resolvent bound.

\begin{table}[H]
\centering
\caption{Quantum algorithms for geometric means}
\label{tab:geommean-compare}
\scriptsize
\setlength{\tabcolsep}{2.2pt}
\renewcommand{\arraystretch}{1.22}
\begin{tabularx}{\textwidth}{@{}
>{\RaggedRight\arraybackslash}p{1.75cm}
>{\RaggedRight\arraybackslash}p{2.65cm}
>{\RaggedRight\arraybackslash}p{3.55cm}
>{\RaggedRight\arraybackslash}p{4.15cm}
>{\RaggedRight\arraybackslash}X
@{}}
\toprule
\textbf{Work} &
\textbf{Input/access model} &
\textbf{Assumptions} &
\textbf{Output} &
\textbf{Query complexity} \\
\midrule
Liu \emph{et al.}~\cite{LiuWangWildeZhang2025}
&
BEs of \(A,C\).
&
\makecell[l]{Hermitian positive definite:\\
\(I\succeq A\succeq I/\kappa_A,\)\\
\(I\succeq C\succeq I/\kappa_C\)}
&
\[
(2\kappa_A,\eps)\text{-BE of }A^{-1}\#C.
\]

&
\[
\widetilde O(\kappa_A\kappa_C)
\]
queries to \(U_C\), and
\[
\widetilde O(\kappa_A^2\kappa_C)
\]
queries to \(U_A\).
\\
\addlinespace

This paper
&
Unit BEs of \(A,B\).
&
\[
\spec(A^{-1}B)\cap(-\infty,0]=\emptyset,
\]
scaled sign bound for
\[
M_\#=\chi_\#G(A,B).
\]
FoV specialization assumes
\[
H(A)\succeq\mu_A I,\quad H(B)\succeq\mu_B I.
\]
&
BEs of both \(A\#B\) and \((A\#B)^{-1}\).  FoV/profile normalization:
\[
4/\sqrt{\mu_A\mu_B}.
\]
&
Generic profile form:
\[
\widetilde O(R_{\#,\rho}).
\]
FoV/profile regime, with
\(\mu_\#=\min\{\mu_A,\mu_B\}\):
\[
\widetilde O(\mu_\#^{-1}).
\]
\\
\bottomrule
\end{tabularx}
\end{table}

\newpage

\paragraph{Continuous-time algebraic Riccati equations.}
The closest Riccati-related quantum baseline uses a positive-definite geometric-mean
reduction.  The present paper treats the standard continuous-time CARE directly through the
Hamiltonian sign projector
\[
\Pi_-=\frac{I-\signf(H)}2,\qquad
H=\begin{bmatrix}A&-G\\ -Q&-A^*\end{bmatrix}.
\]

\begin{table}[H]
\centering
\caption{Quantum algorithms for algebraic Riccati equations}
\label{tab:care-compare}
\scriptsize
\setlength{\tabcolsep}{2.2pt}
\renewcommand{\arraystretch}{1.22}
\begin{tabularx}{\textwidth}{@{}
>{\RaggedRight\arraybackslash}p{1.75cm}
>{\RaggedRight\arraybackslash}p{2.55cm}
>{\RaggedRight\arraybackslash}p{3.65cm}
>{\RaggedRight\arraybackslash}p{4.55cm}
>{\RaggedRight\arraybackslash}X
@{}}
\toprule
\textbf{Work} &
\textbf{Input/access model} &
\textbf{Assumptions} &
\textbf{Output} &
\textbf{Query complexity} \\
\midrule
Liu \emph{et al.}~\cite{LiuWangWildeZhang2025}
&
BEs of \(A,B,C\).
&
Riccati-type equation
\[
Y^\dagger AY-B^\dagger Y-Y^\dagger B-C=0.
\]
Assumes
\[
\begin{aligned}
&I\succeq A\succeq I/\kappa_A,\\
&I\succeq C\succeq I/\kappa_C,\\
&B^\dagger B\preceq I,\\
&A^{-1}B=(A^{-1}B)^\dagger.
\end{aligned}
\]
&
\[
(2\kappa_A^{3/2},\eps)\text{-BE of }Y,
\]
where
\[
Y=A^{-1}\#(B^\dagger A^{-1}B+C)+A^{-1}B.
\]
&
\[
\widetilde O(\kappa_A\kappa_C)
\]
queries to \(U_B,U_C\), and
\[
\widetilde O(\kappa_A^2\kappa_C)
\]
queries to \(U_A\).
\\
\addlinespace

This paper
&
Unit BE of the normalized Hamiltonian \(H\), or built from BEs of \(A,G,Q\).
&
Standard CARE:
\[
A^*X+XA-XGX+Q=0.
\]
Assumes
\[
\spec(H)\cap i\mathbb R=\emptyset,
\]
a Hamiltonian strip/profile bound, and
\(\sigma=\sigma_{\min}(\Pi_{11})>0\).
&
Deterministic BE of
\[
X=\Pi_{21}\Pi_{11}^{-1}.
\]
Normalization:
\[
O\!\left(
\frac{1+\beta_{\rm sign}}{2\sigma-\eps_{\rm sign}}
\right).
\]
&
Sign stage:
\[
\widetilde O(R_H)
\]
queries to \(U_H,U_H^\dagger\).  Projector inversion multiplies by
\[
\widetilde O\!\left(
\frac{\beta_\Pi}{2\sigma-\eps_{\rm sign}}
\right).
\]
\\
\bottomrule
\end{tabularx}
\end{table}
\end{document}